\documentclass[12pt]{article}
\usepackage[english]{babel}
\usepackage{amsthm}
\usepackage{amssymb}
\usepackage{amsmath}
\usepackage{float}
\usepackage{caption}
\usepackage{natbib}
\usepackage{bm}
\usepackage{multirow}
\usepackage{xcolor}
\usepackage{pdfpages}
\usepackage{titling}
\usepackage{mathtools}
\usepackage[export]{adjustbox}   
\usepackage{makecell}  
\theoremstyle{plain}
\newtheorem{proposition}{Proposition}

\newtheorem{theorem}{Theorem}

\newtheorem{remark}{Remark}
\newtheorem{example}{Example}
\newtheorem{criterion}{Criterion}
\newtheorem{definition}{Definition}

\usepackage[bf,SL,BF]{subfigure}
\usepackage{graphicx}
\usepackage{enumerate}
\usepackage{url} 
\usepackage[colorlinks,linkcolor=blue,citecolor=blue,urlcolor=blue,filecolor=blue,backref=page]{hyperref}

\usepackage{algorithm}
\usepackage{algpseudocode}

\usepackage{algpseudocode}
\newcommand{\bbeta}{\bm{\beta}}
\newcommand{\z}{\bm{z}}

\newcommand{\T}{\text{T}}
\renewcommand{\tilde}{\widetilde}

\newcommand{\blind}{1}

\begin{document}

	\def\spacingset#1{\renewcommand{\baselinestretch}%
		{#1}\small\normalsize}
	\spacingset{1.1}
	\date{}
	
	\title{\bf \Large 
		On MCMC mixing for predictive inference under
		unidentified transformation models
	}
	\if1\blind
	{
		\author{ Chong Zhong\thanks{
				\scriptsize
				The author is a Research Associate of Department of Data Science and Artificial Intelligence, The Hong Kong Polytechnic University. },\hspace{.2cm}
			Jin Yang\thanks{  
				\scriptsize
				The author is a Senior Research Fellow of 
				Department of Applied Mathematics, The Hong Kong Polytechnic University.},\hspace{.2cm}
			Junshan Shen\thanks{
				\scriptsize
				The author is an Associate Professor of School of Statistics, Capital University of Economics and Business.},\hspace{.2cm}
			Zhaohai Li \thanks{
				\scriptsize
				The author is a Professor of Department of Statistics, George Washington University, Washington, DC.}, \\
			\hspace{.2cm} and 
			Catherine C. Liu\thanks{
				\scriptsize
				The author is an Associate Professor of Department of Data Science and Artificial Intelligence, The Hong Kong Polytechnic University.}
		}
	}\fi
	
	\if0\blind
	{
		
	} \fi
	\maketitle
	\begin{abstract}
		Reliable Bayesian predictive inference has long been an open problem under unidentified transformation models, since the Markov Chain Monte Carlo (MCMC) chains of posterior predictive distribution (PPD) values are generally poorly mixed. 
		We address the poorly mixed PPD value chains under unidentified transformation models through an adaptive scheme for prior adjustment. 
		Specifically, we originate a conception of sufficient informativeness, which explicitly quantifies the information level provided by nonparametric priors, and assesses MCMC mixing by comparison with the within-chain MCMC variance. 
		We formulate the prior information level by a set of hyperparameters induced from the nonparametric prior elicitation with an analytic expression, which is guaranteed by asymptotic theory for the posterior variance under unidentified transformation models. 
		The analytic prior information level consequently drives a hyperparameter tuning procedure to achieve MCMC mixing. 
		The proposed method is general enough to cover various data domains through a multiplicative error working model. 
		Comprehensive simulations and real-world data analysis demonstrate that our method successfully achieves MCMC mixing and outperforms state-of-the-art competitors in predictive capability.

	\end{abstract}
	
	\newpage
	\spacingset{1.9} 
	\section{Introduction}
	\label{sec:intro}
	
	We study the linear transformation model \citep{cuzick1988rank}, 
	\begin{align}\label{basicLTM}
		h(y) = \bbeta^\T  \bm{z} + \epsilon,
	\end{align}
	where $y \in \mathcal{Y} \subset \mathbb{R}$ is the response, $\bm{z} \in \mathcal{Z} \subset \mathbb{R}^p$ is the $p$-dimensional vector of covariates, $\bbeta \in \mathbb{R}^p$ is the corresponding vector of regression coefficients,  
	$h(\cdot)$ is a strictly increasing function, and $\epsilon$ is the continuous error term with cumulative distribution function (CDF) $F_\epsilon$.
	Over the past decades, numerous studies have contributed to statistical inference under the transformation model \eqref{basicLTM} \citep[][among others]{horowitz1996semiparametric, linton2008estimation, hothorn2018most, kowal2024monte},  
	and may be categorized into two approaches. 
	
	i)  A common strategy is the semiparametric regression approach that \textit{imposes a transformation on a specified reference distribution for the model error $\epsilon$} \citep[among others]{chen2002semiparametric,hothorn2014conditional, siegfried2023distribution, carlan2024bayesian, kowal2024monte, brachem2024bayesian}. 
	This strategy is straightforward and readily implementable, though it may encounter the risk of \textit{model misspecification}. 
	With this reference distribution strategy, most consistency results were established under log-concave-like assumptions on $F_\epsilon$ \citep[among others]{zeng2006efficient, hothorn2018most}. Such assumptions may fail in practice; in say, a normal mixture regression scenario \citep{soffritti2011multivariate, kasahara2015testing}. 
	
	ii) A second approach is to allow $h$ and $F_\epsilon$ to both be unspecified in model \eqref{basicLTM}. 
	\textit{Identification conditions such as scale and location normalization constraints were imposed on either $h$  or $F_\epsilon$} \citep{horowitz1996semiparametric, ye1997nonparametric, chiappori2015nonparametric}. 
	Despite their robustness and ideal theoretical properties, these methods are usually computationally intractable due to the employed kernel smoothing techniques. 
	Although \cite{chen2002rank} proposed a rank estimator of $h$ that does not require smoothing, he did not consider estimating $F_\epsilon$, and therefore cannot estimate the predictive distribution for future data. 
	\cite{mallick2003bayesian} imposed  a constrained Polya tree prior for $F_\epsilon$ to identify model \eqref{basicLTM}, but the posterior computation may not be stable since the posterior could suffer from slow mixing with an inappropriate center distribution \citep{Muller2015}. 
	
	In this article, we consider a third approach, where we allow the unspecified infinite-dimensional parameters $h$ and $F_\epsilon$ to be \textit{unidentified}. 
	That is, given the data, the likelihood is equal for a range of (infinite-dimensional) parameters $(h, \bm{\beta}, F_\epsilon)$; refer to \citet[pp. 105]{horowitz1996semiparametric} for explicit description. 
	We attempt to avoid complicated identification constraints for feasible computation. 
	Specifically, we focus on Bayesian predictive inference (BPI), i.e. estimating the posterior predictive distribution (PPD) for future observations. 
	
	Though the BPI under unidentified models is \textit{conceptually doable}, the key challenge that remains unresolved is the \textbf{poor mixing} of \textit{PPD value chains} due to unidentifiability.
	Given $n$ observed data pairs $\mathcal{D} = \{y_i, \bm{z}_i\}_{i=1}^n$ , suppose the future response $y^*$ is independent of $\mathcal{D}$ given 
	the future covariates $\bm{z}^*$. 
	The PPD value at a point $s \in \mathcal{Y}$ is 
	$    F_{y^*|\bm{z}^*}(s| \mathcal{D})
	= \int F_y(s|\bm{z}^*, h, \bm{\beta}, F_\epsilon) d\pi(h, \bm{\beta}, F_\epsilon|\mathcal{D}), 
	$
	where $F_y(s|\bm{z}^*, h, \bm{\beta}, F_\epsilon)$ is the conditional CDF at $y=s$ under model \eqref{basicLTM} given parameters $(h, \bm{\beta}, F_\epsilon)$, and $\pi(h, \bm{\beta}, F_\epsilon|\mathcal{D})$ is the joint posterior distribution. 
	In practice, the PPD value is numerically approximated by Markov Chain Monte Carlo (MCMC) draws. 
	Suppose that one draws $M$ parallel MCMC chains of the same length $N_d$, obtaining $MN_d$ draws of $(h^{(ml)}, \bm{\beta}^{(ml)},F_{\epsilon}^{(ml)})$, for $m=1, \ldots, M$, $l=1, \ldots, N_d$. 
	Then the PPD value at $s$ is approximated as the average of the PPD value chains: 
	\begin{align*}
		F_{y^*|\bm{z}^* }(s|\mathcal{D}) \approx (MN_d)^{-1} \sum_{m=1}^M \sum_{l=1}^{N_d} F_y(s|\bm{z}^*, h^{(ml)}, \bm{\beta}^{(ml)}, F_\epsilon^{(ml)}).
	\end{align*}
	However, under the unidentified model \eqref{basicLTM}, this approximation will NOT be reliable if the PPD value chains $F(s|\bm{z}^*, h^{(ml)}, \bm{\beta}^{(ml)}, F_\epsilon^{(ml)})$ are \textbf{poorly mixed} in the sense that the M-chain PPD value samples do NOT converge to the stationary distribution. 
	Poorly mixed PPD value chains usually incur poor BPI; see the lower expected log predictive densities \citep[Corollary 5]{yao2022stacking} for illustration.

	Our solution is an \textit{adaptive scheme that leverages prior adjustment to achieve MCMC mixing}. 
	Our scheme operates like a bridge, on the one side is a new insight that \textit{under unidentified transformation models}, \textit{the posterior variance is (asymptotically) dominated by the information level of the elicited Bayesian nonparametric priors} (BNPs). 
	This insight comes from a new asymptotic posterior variance decomposition, where \textit{the remainder term vanishes at a rate of $n^{-1}$}, and the \textit{dominating term is fully determined by the hyperparameters in BNP elicitation}; refer to Theorem \ref{theo: posterior variance}. 
	On the other side is the common principle that MCMC mixing occurs if the within-chain MCMC variance is sufficiently close to the posterior variance \citep[Section 6.1]{brooks2011handbook}. 
	
	The insight and the principle motivate us to conceptualize a \textit{sufficient informativeness criterion}: if the within-chain MCMC variance exceeds the dominating term (or its approximation) of the posterior variance, then the BNPs are sufficiently informative to reach MCMC mixing;
	accordingly, the \textit{prior information level} is defined by the inverse of the dominating term. 
	This criterion distinguishes the popular practice of computing the empirical between- and within-chain variances \citep[among others]{gelman1992inference, brooks1998general} for discrimination of mixing only, since the analytic expression of the  (approximated) prior information level (refer to Eq. \eqref{lowerbound: approximation}) can further activate expedient prior adjustment to achieve MCMC mixing; refer to Algorithm \ref{alg:Tuning}. 
	
	To derive the prior information level, we design an ideal BNP elicitation: a monotone spline model \citep{ramsay1988monotone} that possesses L\'{e}vy properties \citep{doksum1974tailfree}, and a Dirichlet process mixture model \citep{lo1984class} with a Weibull kernel. 
	The hyperparameters induced from the hyperpriors for the BNP elicitation yield a neat analytic expression for an upper approximation to the prior information level. 
	Consequently, prior adjustment is straightforwardly conducted by an adaptive hyperparameter tuning procedure without specific requirements on the initial values (refer to Section \ref{subsec: simu_illustration} for illustration).

	The major contributions of this article are summarized as follows. 
	\begin{itemize}
		\item We contribute a robust and computationally feasible method for predictive inference under transformation models. 
		Our methodology and theoretical results are general enough to cover the response types considered by conditional transformation models \citep{hothorn2014conditional, carlan2024bayesian}. 
		Specifically, we might be the first to establish the posterior inference theory under an unidentified nonparametric model, including the asymptotic posterior variance (Theorem \ref{theo: posterior variance}) and the properness of the joint posterior (Theorem \ref{theorem:proper}). 
		
		\item We contribute an easily implemented method to address the poor mixing of PPD value chains under unidentified transformation models.
		The hyperparameter tuning procedure is implemented under a general MCMC sampler \texttt{Stan} \citep{carpenter2017stan}, releasing us from developing tricky samplers for multimodal target distributions \citep[e.g.][]{pompe2020framework}. 
		
		\item We contribute a  quantile-knot I-spines BNP for nonnegative monotonic smooth functions. 
		The proposed I-spline model enjoys lower model complexity (a few knots are enough) compared with other I-spline variants \citep[e.g.][]{wang2011semiparametric, kim2017bayesian}, while maintaining the root-$n$ posterior contraction rate that guarantees the asymptotic mixture of normals (Theorem \ref{theo:BvM}).  
		
		\item We develop an \texttt{R} package \texttt{BuLTM}, 
		for BPI under the transformation model \eqref{basicLTM}. 
		Comprehensive numerical studies demonstrate that \texttt{BuLTM} achieves the mixing of PPD value chains, and outperforms other state-of-art (SOTA) competitors in prediction tasks.

	\end{itemize}

	\noindent{\textbf{Organization. }}
	Section \ref{sec: method} presents an equivalent working model to model \eqref{basicLTM} and formulates the BNPs. 
	Section \ref{sec: nonparametricposteriorinference} presents the adaptive scheme for prior adjustment to achieve MCMC mixing under the unidentified model. 
	Section \ref{sec:parametricprior} establishes the properness of the joint posterior and introduces estimation of the parametric component. 
	Simulations and applications to real-world data are presented in Sections \ref{sec:sim} and \ref{sec:app} respectively. 
	Section \ref{sec:disc} contributes brief discussion. 
	Technical proofs, additional simulation results, and other related details are collected in the \textit{Supplement}. 
	The companion \texttt{R} package \texttt{BuLTM} and the reproducible code for the numerical studies are available on GitHub \href{https://github.com/LazyLaker/BuLTM}{https://github.com/LazyLaker/BuLTM}.
	
	
	\section{Transformed modeling and nonparametirc priors}
	\label{sec: method}
	\subsection{Multiplicative error working model}
	\label{subsec: ModTrans}
	We first perform a transformation $\tilde{h}$ on the response $y$ to transfer its support to (a subset of) $(0, \tau)$, for an arbitrary positive constant $\tau$. 
	In this article, we consider the $\tau$-Sigmoid function $\tilde{h}(y) = \tau/(1+e^{-y})$. 
	Let $\tilde{y} \in (0, \tau)$ be the transformed response and let $\circ$ be the composition of two functions operator. 
	Based on model \eqref{basicLTM}, we still have a  transformation model $h^* (\tilde{y})= \bbeta^T \bm{z} + \epsilon$, 
	where  $ h^*  =  h \circ \tilde{h}^{-1}$ is n monotone increasing function, with $\tilde{h}^{-1}$ being the (known) inverse function of $\tilde{h}$. 

	Denote the transformation $\exp(h^*)$ by $H$, and accordingly $\xi = \exp({\epsilon})$ with CDF $F_{\xi}$. 
	We further have the following working model which is equivalent to model \eqref{basicLTM} in the sense of identical conditional distribution $F_{y|\bm{z}}$ under the two models
	\begin{align} \label{expmod1}
		H(\tilde{y}) = \xi\exp(\bbeta^\T  \z). 
	\end{align}
	The equivalence is based on the fact that
	$
	Pr\{y \le s|\bm{z}\} = Pr\{h(y) \le h(s)|\bm{z}\} =  Pr\{h^* (\tilde{y}) \le h^*(s) | \bm{z}\}= Pr\{H(\tilde{y}) \le H(s)|\bm{z}\}. 
	$
	Under working model \eqref{expmod1}, the following result holds naturally.
	\begin{proposition}
		\label{prop: H(0) = 0}
		$
		H(0) = 0 
		$ 
		if covariate $\bm{z}$ is independent of model error $\xi$.
	\end{proposition}
	The independence assumption between $\bm{z}$ and $\xi$ is general \citep{cuzick1988rank, horowitz1996semiparametric, chen2002rank}.
	As a result,  the space of $H$ is compressed to the space of nonnegative monotonic functions that passes through the origin. 
	
	\begin{remark}[\textbf{Nonlinearity}]
		\label{rmk: covariate transformation}
		The linear transformation model \eqref{basicLTM} and theworking model \eqref{expmod1} are sufficiently general  to incorporate nonlinear covariate effects. 
		Let $\bm{z} = (z_1, \ldots, z_p)$. 
		Let $\{\phi_{jk}\}_{k=1}^{K_j}$ be some basis functions (e.g. B-spline basis or Fourier basis) on $z_j$'s, for $j=1, \ldots, p$. 
		Let $\tilde{\bm{z}}_j = (\phi_{j1}(z_j), \ldots, \phi_{jK_j}(z_j))^T$ and $\otimes$ be the Kroneker product operator. 
		Based on the tensor product basis \citep{pya2015shape, carlan2024bayesian}, a smooth function $f: \mathbb{R}^p \to \mathbb{R}$ can be rewritten as $f(\bm{z}) = \tilde{\bm{\beta}}^T \tilde{\bm{z}}$, where $\tilde{\bm{z}} = \tilde{\bm{z}_1} \otimes \ldots \otimes \tilde{\bm{z}_p}$, and $\tilde{\bm{\beta}} \in \mathbb{R}^{\prod_{j=1}^p K_j}$. 
		To avoid the curse of dimensionality, one may consider an additive structure for $f$ \citep{linton2008estimation, chen2024semi} such that $f(\bm{z}) = \sum_{j=1}^p f_j(z_j)$, where $f_j(z_j) = \sum_{k=1}^{K_j} \beta_{jk} \phi_{jk}(z_j) \equiv \tilde{\bm{\beta}}^T \tilde{\bm{z}}$, where $\tilde{\bm{z}} = (\tilde{\bm{z}}_1, \ldots, \tilde{\bm{z}}_p)^T$ and $\tilde{\bm{\beta}} = (\beta_{11}, \ldots, \beta_{1K_1}, \ldots, \beta_{p1}, \ldots, \beta_{pK_p})^T$. 
		
	\end{remark}

	\subsection{Bayesian nonparametric priors}
	\label{subsec: nonpara prior}
	
	\subsubsection{Quantile-knot I-splines prior}
	\label{subsec: ispline}
	Given the aforementioned working model, the observed data $\mathcal{D}$ become independent pairs of $\{\tilde{y}_i, \bm{z}_i\}_{i=1}^n$. 
	For the transformed response $\tilde{y}_i$ observed on the interval $D = (0, \tau)$,   a natural method to model $H$ and its derivative $H'$ is to use the monotone spline basis,
	\begin{align}\label{ispline}
		H(s) = \sum_{j=1}^{K}\alpha_j B_j(s),  ~H'(s) = \sum_{j=1}^{K}\alpha_j B_j'(s),
	\end{align}
	where $\{\alpha_j\}_{j=1}^{K}$ are positive coefficients to guarantee nondecreasing monotonicity, $\{B_j(s)\}_{j=1}^{K}$ are I-spline functions \citep{ramsay1988monotone} on $D$ and $\{B_j'(s)\}_{j=1}^K$ are corresponding derivatives. 
	Once $\{\alpha_j\}_{j=1}^K$ are specified, $H$ and $H'$ are uniquely determined. 
	By Proposition \ref{prop: H(0) = 0}, we set $H(0) = 0$ directly, unlike existing I-splines approaches that include an unknown intercept.  
	A fundamental problem in spline modeling is how to specify the number of basis functions $K$, which is the sum of the number of interior knots and the order of smoothness $r$, defined by the existence of the $(r-1)$th order derivative. 
	Empirically, the degree $r$ may take a value from $2$ to $4$ and we take the default value $r=4$ in \texttt{R} package \texttt{splines2} \citep{splines2-paper}.
	The remaining task is to specify the number and locations of the interior knots. 
	
	We select interior knots from quantiles of the observed data, fitting a quantile-knots I-splines model, rather than using equally spaced knots. 
	Let $\hat{F}_{n}(s) = n^{-1}\sum_{i=1}^n I(\tilde{y}_i  \le  s)$ be the empirical CDF of $\tilde{y}$ and $\hat{Q}_{\tilde{y}}(q) = \hat{F}_{n}^{-1}(q) = \inf \{s: q\le \hat{F}_{n}(s)\}$ be the corresponding empirical quantile function, for $s\in (0, \tau)$ and $q \in (0, 1)$. 
	{We first specify $N_I$, the number of interior knots (we set $N_I = 4$ in this article as the default choice). 
		Then the interior knots are set as $s_j = \hat{Q}_{\tilde{y}}(j/N_I)$, for $j=0, \ldots, N_I - 1$. 
		Such a quantile-knot configuration guarantees that the observed data lie uniformly between the knots. }

	
	
	Our quantile-knot I-spline BNP is appealing since one only needs a few knots rather than an increasing number of interior equally spaced knots \citep{wang2011semiparametric, kim2017bayesian}, and hence has lower computational complexity. 
	This BNP is not sensitive to the choice of the number of initial knots; refer to \textit{Supplement} D.3. 
	By assigning independent and identically distributed hyperpriors for the coefficients $\alpha_j$, the proposed quantile-knot I-spline BNP is closely related to the L\'{e}vy process \citep{doksum1974tailfree}; refer to Proposition A.1 in \textit{Supplement}.  
	This proposition guarantees the local asymptotics in Theorem \ref{theo:BvM} below.

	\subsubsection{Dirichlet process mixture model}
	For the prior for $F_\xi$ we consider the common Dirichlet process mixture (DPM) model \citep{lo1984class}. 
	Here we employ a truncated stick-breaking construction of the DPM, denoted as 
	\begin{align*}
		F_\xi(\cdot) = \int F_0(\cdot|\bm{u})dG(\bm{u}), ~f_\xi(\cdot) = \int f_0(\cdot|\bm{u})dG(\bm{u}), ~
		G = \sum_{l=1}^L p_l\delta_{\bm{u}_l}, ~\bm{u}_l \sim G_0,
	\end{align*}
	where $F_0$ and $f_0$ are called kernels from a distribution family parameterized by $\bm{u}$, $L$ is a truncation number of the Dirichlet process, $p_l$ are corresponding sticking-breaking weights, and $\bm{u}_l$ are i.i.d. atoms from the base measure $G_0$. 
	More justifications for the truncation level $L$ are deferred to \textit{\textit{Supplement}} B. 
	
	Note that $\xi$ is an arbitrary \textit{continuous positive} random variable. 
	In this article, we select the Weibull kernel for the DPM model,
	\begin{align}
		\label{DPM}
		F_\xi(\cdot) = \sum_{l=1}^L p_l F_w(\cdot|\psi_l, \nu_l), ~ f_{\xi}(\cdot) =(F_\xi)' =\sum_{l=1}^L p_l f_w(\cdot|\psi_l, \nu_l),
	\end{align}
	where $F_w(x|\psi_l, \nu_l) = 1- \exp\{-(x/\psi_l)^{\nu_l}\}$ and $f_w(x|\psi_l, \nu_l) = \nu_l\psi_l^{-\nu_l}x^{\nu_l -1}\exp\{-(x/\psi_l)^{\nu_l}\} $ are the CDF and the pdf of the Weibull distribution with parameters $\{(\psi_l, \nu_l)\}_{l=1}^L$. 
	{Expression \eqref{DPM} yields a Weibull mixture model that has $L$ allocations of mixture components $(\psi_l, \nu_l)$, each with DP weights $p_l$. }

	The above Weibull kernel has at least two advantages: i) it can capture the shape of both monotone and nonmonotone hazards \citep{kottas2006nonparametric}, and ii) it guarantees that the joint posterior under the unidentified working model \eqref{expmod1} is proper; refer to Theorem \ref{theorem:proper}. 
	
	\subsubsection{Exponential hyperpriors and hyperparameters}
	{Our nonparametric prior elicitation is completed by assigning hyperpriors to the parameters in the quantile-knot I-spline prior \eqref{ispline} and DPM model \eqref{DPM}. 
		Let $\bm{\alpha} = \{\alpha_j\}_{j=1}^K$, $\bm{p} = \{p_l\}_{l=1}^L$, $\bm{\psi} = \{\psi_j\}_{j=1}^K$, and $\bm{\nu} = \{\nu_j\}_{j=1}^K$. 
		The hyperprior for $\bm{p}$ is naturally the stick-breaking prior \citep{sethuraman1994constructive}. 
		For $(\bm{\alpha}, \bm{\psi}, \bm{\nu})$, we assign exponential hyperpriors 
		\begin{align}
			\label{hyperprior: exponential}
			\pi(\bm{\alpha}) = \prod_{j=1}^K \text{Exp}(\alpha_j; \eta), ~\pi(\bm{\psi}) =\prod_{l=1}^L  \text{Exp}(\psi_l; \zeta), ~\pi(\bm{\nu}) = \prod_{l=1}^L \text{Exp}(\nu_l; \rho). 
		\end{align}
		The rationale for employing the exponential hyperpriors is straightforward. 
		For $\bm{\alpha}$, a Gamma hyperprior is preferable to link the I-splines model \eqref{ispline} with a Gamma process; refer to Proposition A.1 in the \textit{Supplement}; 
		for $(\bm{\psi}, \bm{\nu})$ in the DPM with Weibull kernels, Gamma hyperpriors are becoming popular choices  \citep{shi2019low}.
		We use exponential hyperpriors to avoid mathematically complicated formulations, though our theoretical results hold for arbitrary Gamma hyperpriors. 
		With exponential hyperpriors, the nonparametric priors for $(H, F_\xi)$ are parameterized by the hyperparameters $(\eta, \zeta, \rho)$. 
	}

	

	\section{Adaptive scheme for prior adjustment}
	\label{sec: nonparametricposteriorinference}
	The (infinite-dimensional) parameters in working model \eqref{expmod1} are still unidentified. 
	Suppose equation \eqref{expmod1} holds for a special triplet solution $(H_0, \bbeta_0, F_{\xi_0})$. 
	Then equation \eqref{expmod1} also holds on the set 
	$
	\mathcal{C}\{(H, \bbeta, F_\xi)\} = \{(c_1H_0^{c_2}, c_2\bbeta_0, F_{c_1\xi_0^{c_2}})\}
	$
	for any pair of positive constants $(c_1, c_2) \in \mathbb{R}_{+}^2$.
	In this section, we introduce an adaptive scheme to address the poorly mixing of PPD value chains under the unidentified working model \eqref{expmod1}. 
	In Section \ref{subsec: pv decomp}, we focus on the asymptotic posterior variance first. 
	If the posterior variance is divergent, no mixing results can be guaranteed; 
	otherwise, it is possible for a general MCMC sampler to sufficiently explore the posterior uncertainty. 
	In Section \ref{subsec: sufficient informative}, we formulate the sufficient informativeness criterion based on the theoretical results in Section \ref{subsec: pv decomp}, and elucidate how to use the criterion to adaptively tune the hyperparameters to achieve MCMC mixing for trustworthy BPI. 
	From now on, denote the expectation and variance operator with respect to a parameter $\theta$ under law $\pi(\theta)$ by $\mathbb{E}_{\theta}$ and $\mathbb{V}_{\theta}$ respectively, where $\pi(\theta)$ denotes the prior distribution of the parameter $\theta$.  
	
	\subsection{Posterior variance under transformation models}
	\label{subsec: pv decomp}
	
	Specifically, we focus on $\mathbb{V}\{H(s)|\mathcal{D}\}$, the posterior variance of $H(s)$ for some specific $s \in (0, \tau)$. 
	Our motivation is the following conditional cumulative hazard function of the transformed response $\tilde{y}$ given covariates $\bm{z}$. 
	With the nonparametric prior elicitation \eqref{ispline} and \eqref{DPM}, for $s \in (0, \tau)$, we have 
	\begin{eqnarray}
		\label{G_xi}
		\begin{aligned}
			\Lambda_{\tilde{y}|\bm{z}}(s)  = \log\left\{\sum_{l=1}^L p_l \exp\left(-\left\{\frac{\sum_{j=1}^K \alpha_j B_j(s) \exp(-\bbeta^T \z)}{\psi_l}\right\}^{\nu_l}\right)\right\}.  \\
		\end{aligned}
	\end{eqnarray}
	In \eqref{G_xi}, the DPM components $(\bm{p}, \bm{\psi}, \bm{\nu})$ encounter the label-switching issue since the conditional cumulative hazard $\Lambda_{\tilde{y}|\bm{z}}$ is invariant {under any permutations of the indices of the allocation $(\psi_l, \nu_l)$}; refer to \citep[pp. 1156]{mena2015bayesian} for general {illustration}. 
	As a result, it is impossible to identify these parameters individually even if $F_{\xi}$ is specified. 
	Fortunately, for a fixed $s \in (0, \tau)$, $H(s)$ does NOT encounter the label-switching issue as any permutations of DPM components has no impact on $\bm{\alpha}$ or $H(s)$, since $H$ is fully determined by $\bm{\alpha}$. 
	This fact partially explains why we focus on the posterior variance of $H(s)$. 
	
	\subsubsection{Preliminary: identified scenario}
	We start from a preliminary result in the case where $F_{\xi}$ is specified. 
	With DPM model \eqref{DPM}, specifying $F_\xi =  F_{\xi_0}$ is equivalent to specifying $(\bm{p}, \bm{\psi}, \bm{\nu})$ at the ground truth $(\bm{p}_0, \bm{\psi}_0, \bm{\nu}_0)$. 
	The following conditions are further assumed. \\
	(A1) All transformed response $\tilde{y_i}$ are distinct. \\
	(A2) There exists a constant $ 0< M_{\z} < \infty$ such that  $||\z||_1 < M_{\z}$ with probability 1.\\
	(A3) The prior $\pi(\bbeta)$ is continuous and $\pi(\bbeta) >0 $ on $\mathbb{R}^p$. \\
	(A4) The ``true" $F_{\xi_0}$ can be expressed in the form of \eqref{DPM}; in \eqref{DPM}, $p_l > \delta$ for some positive constant $\delta$, $\sum_{l=1}^L \nu_l < \infty$ for $l = 1, \ldots, L$. 
	
	Conditions (A1), (A2), and (A3) are general conditions in the literature for semiparametric Bernstein-von Mises (BvM) results \citep{kim2006bernstein, kim2017bayesian}. 
	Condition (A4) requires $F_{\xi0}$ to be from the Weibull-kernel DPM family. 
	In practice, (A4) can be relaxed so that the ``true" $F_{\xi_0}$ falls into Weibull-kernel DPM's 
	Kullback-Leibler neighborhood, which is quite general \citep[Theorem 13]{wu2008kullback}. 
	
	The following theorem describes the asymptotic marginal posterior distribution of $H(s_j)$, with ``ground truth" $F_{\xi_0}$ given. 
	The proof is deferred to \textit{Supplement} A. 
	\begin{theorem}[Asymptotic mixture of normals]
		\label{theo:BvM}
		Suppose the ``ground truth" $(\bm{p}_0, \bm{\psi}_0, \bm{\nu}_0)$ is known. 
		Let $H_0$ be the corresponding ``true" transformation. 
		Under conditions (A1) to (A4), with nonparametric priors \eqref{ispline} and \eqref{DPM}, and hyperprior \eqref{hyperprior: exponential}, for prespecified interior knots $s_j$ of the I-spline basis, for $j=1, \ldots, J$, as the data size $n \to \infty$, we have 
		$$
		\pi[\sqrt{n} \{H(s_j) - H_0(s_j)\}|\mathcal{D}, \bm{p}_0, \bm{\psi}_0, \bm{\nu}_0] \xrightarrow{~~\text{d}~~~} \sum_{l=1}^L p_{l0} N\left \{0, ~p_{l0}^{-1} \left(\frac{\psi_{l0}}{\nu_{l0}}\right)^2
		H_0(s_j)^{\frac{2}{\nu_{l0}}-2}U_l(s_j)\right\}, 
		$$
		where $U_l(s) = \int_0^s \{S_l^0(\mathcal{D}, \bbeta_0)\}^{-1} d\Lambda_
		{l0}(s)$, with $\Lambda_{l0}(s) = \{H_0(s)/\psi_l\}^{\nu_l}$  
		and $S_l^0(\mathcal{D}, \bbeta_0)$ is some positive constant depending on $\bbeta_0$, $\nu_{l0}$ and data $\mathcal{D}$, for $l=1, \ldots, L$. 
	\end{theorem}
	
	Theorem \ref{theo:BvM} relies on the fact that $H(s_j)$ are sampled from a L\'{e}vy process $\mathcal{H}$ (refer to Proposition A.1 in the \textit{Supplement}), which guarantees that $n^{1/2}(\mathcal{H} - H_0)|\mathcal{D}$ weakly converges to a mixture of Gaussian processes. 
	Consequently, the local posterior on a specific $s_j$ converges to a mixture of normals. 
	The mixture of normals in Theorem \ref{theo:BvM} comes from the mixture structure of $F_{\xi 0}$. 
	Theorem \ref{theo:BvM} also holds for censored data under the condition $\lim \limits_{n \to \infty} n_1/n  >0$, where $n_1$ denotes the number of uncensored observations. 
	In the special proportional hazard case where $L=1$ and $\psi_{10} = \nu_{10} = 1$, Theorem \ref{theo:BvM} reduces to the BvM theorem \citep[Theorem 3.3]{kim2006bernstein}. 
	\begin{remark}
		Theorem \ref{theo:BvM} can be extended to establish $\sqrt{n}$-consistency of $H(s)$ for all $s \in (0, \tau)$ by further assuming the following conditions: i) the number of knots $J\equiv J_n \to \infty$ as $n \to \infty$ such that $\max\limits_{j=1, \ldots, J} |s_j - s_{j+1}| \lesssim n^{-1/2}$; ii) $H_0$ is absolutely continuous on $[0, \tau]$. 
		Nevertheless,  empirically a few knots are sufficient for estimation of PPDs. 
		Meanwhile, to derive the sufficient informativeness criterion (refer to Criterion \ref{criterion: MCMC mixing} in the next subsection), we only need the $\sqrt{n}$-consistency of $H$ with respect to each knot $s_j$.     
	\end{remark}


	\subsubsection{Unidentified scenario}
	Under the unidentified model \eqref{expmod1}, where $F_\xi$ is drawn from the DPM model \eqref{DPM},  the ``ground truth" $F_{\xi_0}$ is no longer a fixed distribution, but, a sample of random functions.
	In this case, the posterior becomes multi-modal and the posterior variance will not vanish anymore.
	The following theorem formulates the asymptotic posterior variance of $H(s_j)$ in the unidentified scenario. 
	
	\begin{theorem}[Asymptotic posterior variance]
		\label{theo: posterior variance}
		Assume conditions (A1) to (A4). 
		Let $w_{j'} = B_{j+j'}(s_j) -  B_{j+j'}(s_{j-1})$ for $j' = 1, \ldots, r$ in the I-splines model \eqref{ispline}. 
		As $n \to \infty$, under model \eqref{expmod1}, for parameter $\nu_1$ in DPM model \eqref{DPM},  with hyperprior \eqref{hyperprior: exponential}, there exist series of positive constants $\{c_{lj}\}_{l=1}^L$ and $\{r_{l}\}_{l=1}^L$ with $r_{1} = 1$,  such that for $j=1, \ldots, J$, 
		$g_{s_j}(\nu_1, \eta, \zeta)  \equiv \zeta \sum_{l=1}^L c_{lj}^{\frac{1}{r_l\nu_1}} + \eta (j+ \sum_{j'=1}^r w_{j'})^{-1},$ and 
		\begin{eqnarray}
			\label{lowerbound: posteriorvariance}
			{
				\begin{aligned}
					\mathbb{V}\{H(s_j)| \mathcal{D}\} &= \left[\mathbb{V}_{\nu_1}\left\{
					g_{s_j}^{-1}(\nu_1, \eta, \zeta) 
					\right\} + \mathbb{E}_{\nu_1}\left\{g_{s_j}^{-2}(\nu_1, \eta, \zeta) \right\}\right]  + O(n^{-1}), \\
					&\equiv \mathcal{V}_{s_j}(\eta, \zeta, \rho) +  O(n^{-1}).  
				\end{aligned}
			}
		\end{eqnarray}
	\end{theorem}
	The first term on the RHS of \eqref{lowerbound: posteriorvariance} can be fully expressed in terms of the hyperparameters $(\eta, \zeta, \rho)$ since $\nu_1$ is integrated out. 
	The second term is a remainder that vanishes at a rate of $n^{-1}$, which is a direct consequence of Theorem \ref{theo:BvM}. 
	Indeed, Theorem \ref{theo: posterior variance} is an explicit form of the following law of total variance under working model \eqref{expmod1}
	\begin{eqnarray}
		\label{total variance}
		\begin{aligned}
			\mathbb{V}\{H(s_j)|\mathcal{D}\} = \underbrace{\mathbb{V}_{F_{\xi_0}}\{\mathbb{E}(H(s_j)|\mathcal{D}, F_{\xi_0})\}}_{\mathcal{V}_{s_j}(\eta, \zeta, \rho) } + \underbrace{\mathbb{E}_{F_{\xi_0}}\{\mathbb{V}(H(s_j)|\mathcal{D}, F_{\xi_0})\}}_{O(n^{-1})}. 
		\end{aligned}
	\end{eqnarray}
	In \eqref{total variance}, we call $\mathbb{V}_{F_{\xi_0}}\{\mathbb{E}(H(s_j)|\mathcal{D}, F_{\xi_0})\}$ the \textbf{mode variance} since it is the variance of the posterior modes $\mathbb{E}(H(s_j)|\mathcal{D}, F_{\xi_0})$, and call $\mathbb{E}_{F_{\xi_0}}\{\mathbb{V}(H(s_j)|\mathcal{D}, F_{\xi_0})\}$ the \textbf{local variance} since it is the average of the variance around each local mode of the posterior.
	Obviously, if the model is identified, the mode variance disappears since the posterior mode is unique and fixed. 
	In this unidentified scenario, the take-home messages of Theorem \ref{theo: posterior variance} are: i) under unidentified transformation models, asymptotically, the posterior variance will be dominated by the mode variance, which is expressed by the hyperparameters $(\eta, \zeta, \rho)$;
	ii) the large mode variance accounts for the poor mixing of MCMC chains, since the multiple chains should be sufficiently dispersed, if the single-chain variation is not enough to recover the mode variance.

	
	
	To calculate the mode variance, we still needs to know the two positive constant series $\{c_{lj}\}_{l=1}^L$ and $\{r_l\}_{l=1}^L$, which are, however, unobservable. 
	Fortunately, they have specific interpretations. 
	Based on \eqref{G_xi}, given ``true" $\bm{\nu}_0$, at each knot $s_j$, the ratio between ``true" parameters $\bm{\psi}_0$ and $H_0$
	$$
	c_{lj} \equiv \left\{\frac{\psi_{l0}}{H_0(s_j)} \right\}^{\nu_{l0}} = \left\{\frac{\psi_{l0}}{\sum_{j=1}^K \alpha_{0j} B_j(s_j)} \right\}^{\nu_{l0}}, ~ j =1, \ldots, K
	$$
	is uniquely determined, for $l = 1, \ldots, L$.
	Furthermore, we can show that all ``true" $\bm{\nu}_0$ fall in the space $\{(\nu_{10}, \ldots, \nu_{L0}): \nu_{l0}/\nu_{10} = r_l, l=2, \ldots, L\}$, where $r_l$ are some fixed positive constants; refer to Proposition A.4 in the \textit{Supplement}. 
	Based on these interpretations,  we present an approximation to $\{c_{lj}\}_{l=1}^L$ in our adaptive scheme in the next subsection. 
	
	\subsection{Sufficient informativeness for MCMC mixing}
	\label{subsec: sufficient informative}
	The above decomposition of the posterior variance  motivates an adaptive scheme for prior adjustment to achieve the mixing of PPD value chains under unidentified transformation models. 
	For a specific knot $s_j$, let $\mathbb{V}_{\text{WI}}(H(s_j)|\mathcal{D})$ be the within-chain MCMC variance of $H(s_j)$ of $M$ MCMC chains of length $N_d$:  
	$$
	\mathbb{V}_{\text{WI}}(H(s_j)|\mathcal{D}) = M^{-1}\sum_{m=1}^M\sum_{l=1}^{N_d} \{H^{(ml)}(s_j) - \bar{H}^{(m)}(s_j)\}^2, ~ \bar{H}^{(m)}(s_j) = N_d^{-1} \sum_{l=1}^{N_d}H^{(ml)}(s_j). 
	$$
	The following criterion assesses the MCMC mixing via the mode variance $\mathcal{V}_{s_j}$. 
	
	\begin{criterion}[Sufficient informativeness criterion]
		\label{criterion: MCMC mixing}
		Under working model \eqref{expmod1}, the chains of PPD value of the transformed response $\widetilde{y}$ at the point $s_j$ are well mixed if $
		\mathbb{V}_{\text{WI}}\{H(s_j)|\mathcal{D}\} \ge \mathcal{V}_{s_j}(\eta, \zeta, \rho),  
		$ 
		for $j=1, \ldots, J$. 
		Then the BNPs for $H$ and $F_\xi$ are sufficiently informative. 
	\end{criterion}
	
	Criterion \ref{criterion: MCMC mixing} identifies the mixing of PPD value chains if the within-chain MCMC variance $\mathbb{V}_{\text{WI}}\{H(s_j)|\mathcal{D}\}$ exceeds the mode variance.
	Criterion \ref{criterion: MCMC mixing} stands on two facts: i) in well mixed MCMC chains, $\mathbb{V}_{\text{WI}}\{H(s_j)|\mathcal{D}\}$ should approach (from below) the posterior variance $\mathbb{V}\{H(s_j)|\mathcal{D}\}$ based on the ergodic theorem \citep{birkhoff1942ergodic}, and ii) the mode variance $\mathcal{V}_{s_j}(\eta, \zeta, \rho)$ is smaller than but dominates $\mathbb{V}\{H(s_j)|\mathcal{D}\}$ based on Theorem \ref{theo: posterior variance}. 
	Consequently, we compare the within-chain variance with the mode variance to examine the convergence of MCMC chains to the target distribution. 
	
	
	\begin{remark}
		Note that the mode variance $\mathcal{V}_{s_j}$ is a function of the hyperparameters $(\eta, \zeta, \rho)$ for BNP elicitation \eqref{ispline} and \eqref{DPM}. 
		The hyperparameters $(\eta, \zeta, \rho)$ fully determine the uncertainties of the hyperpriors. 
		Consequently, we call $\mathcal{V}_{s_j}^{-1}$ the ``\textbf{prior information level}": the smaller the mode variance, the more informative the BNPs are. 
		Nevertheless, too informative BNPs can hinder the prior-to-posterior updating. 
		Thus, our criterion uses the inverse of the within-chain MCMC variance $\mathbb{V}_{\text{WI}}^{-1}\{H(s_j)|\mathcal{D}\}$ as a lower bound to determine the ``sufficient" prior information level that achieves MCMC mixing and avoids slow posterior sampling. 
	\end{remark}
	
	Intuitively, by tuning the hyperparameters $(\eta, \zeta, \rho)$, we can increase the prior information level (or equivalently, decrease the mode variance $\mathcal{V}_{s_j}$) to satisfy Criterion \ref{criterion: MCMC mixing}. 
	The remaining question is to approximate the unobservable 
	constant series $\{c_{lj}\}_{l=1}^L$ and $\{r_l\}_{l=1}^L$. 
	
	By observing the  form  of $g_{s_j}(\nu_1, \eta, \zeta)$ on the RHS of \eqref{lowerbound: posteriorvariance}, 
	if there exists a knot $s_j$ such that $c_{lj} \approx 1$, we can cancel $\nu_1$ and obtain a simple closed-form approximation for $\mathcal{V}_{s_j}(\eta, \zeta, \rho)$. 
	Particularly, we only need a lower bound for $\mathcal{V}_{s_j}(\eta, \zeta, \rho)$ since Criterion \ref{criterion: MCMC mixing} requires that the within-chain variance exceeds the information level. 
	Therefore, to apply Criterion \ref{criterion: MCMC mixing}, it suffices to distinguish whether $c_{lj} < 1$ or not. 
	Based on \eqref{G_xi}, we have 
	\begin{align*}
		F_{\tilde{y}|\bm{z} = \bm{0}_p}(s_j) =1- \sum_{l=1}^L p_l \exp\left[-\left\{\frac{H(s_j)}{\psi_l}\right\}^{\nu_l}\right]\equiv 1- \sum_{l=1}^L p_l \exp\left(-c_{lj}^{-\nu_l}\right). 
	\end{align*}
	Suppose there exists $s_{j_0}$ such such that $F_{y^*|\bm{z} = \bm{0}_p}(s_j) \ge 1 - e^{-1}$. 
	We have $\sum_{l=1}^L p_l c_{lj_0}^{-\nu_l} \ge 1$. 
	That is, for $l=1, \ldots, L$, there exists at least one $c_{l_{j_0}} < 1$. 
	Then, we obtain the analytic expression of a lower approximation to $\mathcal{V}_{s_{j_0}}$ by replacing $c_{lj_0}$ to 1
	\begin{align}
		\label{lowerbound: approximation}
		{
			\mathcal{V}_{s_{j_0}} \ge \left(L\zeta   + \frac{\eta}{ j+ \sum_{j'=1}^r w_{j'}} \right)^{-1} + \left(L\zeta   + \frac{\eta}{j+ \sum_{j'=1}^r w_{j'}} \right)^{-2} \equiv \tilde{\mathcal{V}}_{s_{j_0}}(\eta, \zeta). 
		}
	\end{align}
	To use this approximation, we have to first specify the knot $s_j$. 
	In practice, we consider the knot $s_{j_0}$ in the I-splines model \eqref{ispline} such that $s_{j_0}$ is the smallest among the knots that are greater than the $1-e^{-1}$ quantile of the transformed responses. 
	We summarize this as the following criterion, an applicable version of Criterion \ref{criterion: MCMC mixing}. 
	{
		\begin{criterion}[Applicable sufficient informativeness criterion]
			\label{cri: weaklyinformative NTM}
			Under the working model \eqref{expmod1}, suppose we draw $M>1$ parallel MCMC chains. 
			In the I-splines model \eqref{ispline}, let 
			$
			s_{j_0} = \hat{Q}_{\tilde{y}} (q_0/N_I)$ be the specific knot used for our criterion , where $q_0 = \min\limits_{q = 0, \ldots, N_I-1} \left\{1- q/N_I < e^{-1} \right\}
			$. 
			
			Then the BNPs are sufficiently informative if
			\begin{align}
				\label{threshold}
				\mathbb{V}_{WI}\left\{H(s_{j_0})|\mathcal{D}\right\} \ge \tilde{\mathcal{V}}_{s_{j_0}}(\eta, \zeta).  
			\end{align}
		\end{criterion}
	}
	Criterion \ref{cri: weaklyinformative NTM} requires pre-configuration of the hyperparameter $\rho$ before MCMC sampling. 
	We recommend specifying $\rho = 1$ such that $E(\nu_l) = 1$, which is the same as the expectation of the LIO Weibull kernel hyperprior \citep[pp. 690]{shi2019low}.

	\begin{algorithm*}[!htb]\footnotesize
		\caption{Adaptive tuning of hyperparameters $(\eta, \zeta)$ to reach MCMC mixing. }\label{alg:Tuning}
		\begin{algorithmic}[1]
			\State Specify $s_{j_0}$ in Criterion \ref{cri: weaklyinformative NTM}  based on data $\mathcal{D}$;   set initial values for $(\eta, \zeta) \leftarrow  (\eta_0, \zeta_{0})$. 
			\State  Draw $M > 1$ MCMC chains with $N_d$ draws and examine the mixing by Criterion \ref{cri: weaklyinformative NTM}. 
			\If{inequality \eqref{threshold} does not hold}
			\State Select candidates $(\eta_{\text{new}}, \zeta_{\text{new}})$ and set $\eta \leftarrow \eta_{\text{new}}$, $\zeta \leftarrow \zeta_{\text{new}}$;  repeat 2 until Criterion \ref{cri: weaklyinformative NTM} is met. 
			\EndIf
		\end{algorithmic}
	\end{algorithm*}
	
	The adaptive tuning procedure to select hyperparameters $(\eta, \zeta)$ is summarized in Algorithm \ref{alg:Tuning}.
	The choice of the initial values $(\eta_0, \zeta_{0})$ is arbitrary. 
	We recommend starting from very small values $(\eta_0, \zeta_0)$ to elicit ``noninformative" BNPs. 
	The selection of tuning candidates $\eta_{\text{new}}$ and $\zeta_{\text{new}}$ and the tuning procedure is illustrated and visualized in Section \ref{subsec: simu_illustration}. 
	The number of MCMC draws $N_d$ in each chain is related to the effective sample size and the MCMC sampler used. 
	In \texttt{Stan}, we recommend using a chain length of $N_d = 500$ (after a warm-up phase of the same length).
	Sensitivity analysis finds that longer MCMC chains will not change the tuning result. 
	Detailed discussions on the chain length needed are deferred to \textit{Supplement} D.1.

	\section{Joint posterior and parametric estimation}
	\label{sec:parametricprior}
	In this section, we attempt to answer the following two questions related to Bayesian inference under the working model \eqref{expmod1}:
	the first and most basic question is ``is the joint posterior proper without identifiability?" and 
	the next question is ``can we estimate $\bm{\beta}$ with correct uncertainty quantification?"

	\subsection{The joint posterior is proper}
	\label{subsec: proper posterior}
	Let $\bm{\theta} = (\bm{\alpha}, \bm{\psi}, \bm{\nu},  \bm{p}, \bm{\beta} )$ be the collection of all parameters under working model \eqref{expmod1}. 
	For the parametric component $\bbeta$, we consider the objective improper uniform  prior $\pi(\bm{\beta}) \propto 1$. 
	The following general conditions are assumed.  \\
	(B1) $\pi(\bm{p})$, $\pi(\bm{\psi})$, and $\pi(\bm{\nu
	})$ in model \eqref{DPM} and $\pi(\bm{\alpha})$ in model \eqref{ispline} are proper; \\
	(B2) $0 < K, L < \infty$ in models \eqref{DPM} and \eqref{ispline}; \\
	(B3) The kernel $f_w$ in model  \eqref{DPM} satisfies $xf_w(x) < \infty$ for all $x>0$; \\
	(B4) The $n \times p$ covariate matrix $\bm{Z}$ is of full rank $p$. \\
	Conditions (B1) and (B2) are naturally satisfied, and (B3) is satisfied by the Weibull kernel. 
	(B4) is similar to the condition $(ii)$ in \cite{de2014bayesian}, which is practical and easily validated. 
	The following theorem tells us that, even with an improper prior for $\bbeta$, the joint posterior of $\bm{\theta}$ is still proper. 
	The proof is deferred to \textit{Supplement} A.5. 
	\begin{theorem}
		\label{theorem:proper}
		Assume conditions (A1) to (A3) and (B1) to (B4). 
		With the improper uniform prior for $\bbeta$, under model \eqref{expmod1}, the posterior of $\bm{\theta}$ is proper. 
	\end{theorem} 
	
	Theorem \ref{theorem:proper} contradicts the results for unidentified parametric linear models, where proper priors lead to improper posteriors \citep{GelfandSahu1999JASA}. 
	This observation may imply that the infinite-dimensional parameters play a dominant role if a nonparametric model also has parametric components. 
	Theorem \ref{theorem:proper} can be further extended to right-censored data by relaxing (B4) to $\bm{Z}^*$, the $n_1 \times p$ covariate matrix of uncensored observations is of full rank $p$. 
	

	\subsection{Parametric estimation with posterior projection}
	\label{subsec: identified parametric estimation}
	Under unidentified model \eqref{expmod1}, the marginal posterior intervals of $\bbeta$ are generally too long to correctly quantify the uncertainty \citep{gelman2013bayesian}. 
	Therefore, we are driven to obtain the posterior of $\bbeta^*$, the identified counterpart of $\bbeta$ with certain normalization. 
	Specifically, we consider unit-norm normalization such that $||\bbeta^*||_2 = 1$, where $||\cdot||_2$ denotes the Euclidian norm on $\mathbb{R}^p$. 
	This differs from the element-one constraint, which needs extra effort to choose the covariate with coefficient fixed at 1 (\cite{song2007semiparametric}; \cite{lin2017robust}; among others). 
	
	Rather than sampling $\bbeta^*$ from the constrained space directly, we adopt the posterior projection \citep{sen2022constrained} to project the marginal posterior of unconstrained $\bbeta$ to the constrained parameter space of $\bbeta^*$, the unit hyper-sphere $\text{St}(1, p)$ in $\mathbb{R}^p$.
	The metric projection operator $m_\mathcal{A}: \mathbb{R}^p \to \mathcal{A}$ of a set $\mathcal{A}$ is 
	\begin{align*}
		m_{\mathcal{A}}(\bm{x}) = \{\bm{x}^*\in \mathcal{A}: ||\bm{x}-\bm{x}^*||_2=\inf \limits_{\bm{v}\in \mathcal{A}} ||\bm{x} - \bm{v}||_2\}.
	\end{align*}
	By definition, the metric projection of a vector $\bm{\beta} \in \mathbb{R}^p$ into $\text{St}(1, p)$ is $m_{\text{St}(1, p)}(\bm{\beta}) = \bm{\beta}/||\bm{\beta}||_2$.
	Note that projecting the posterior of the unconstrained $\bbeta$ to $\bbeta^*$ does not cause any extra computational burden. 
	Meanwhile, it is anticipated that the posterior of the projected $\bbeta^*$ is $\sqrt{n}$-consistent based on Theorem \ref{theo:BvM}, since the posterior contraction rate of the projected posterior is at least that of the original posterior \citep[Theorem 2]{sen2022constrained}.
	Numerical studies valid the claim that the projection leads to accurate estimation of $\bbeta$ with uncertainty correctly quantified by the induced posterior interval; refer to \textit{Supplement} C.5. 
	
	\section{Simulations}
	\label{sec:sim}
	Extensive simulations are conducted to evaluate two aspects of the proposed method: 
	i) {how the proposed adaptive scheme guarantees the mixing of PPD value chains}; and ii) 
	{how the proposed \texttt{BuLTM} package performs in predictive inference}. 
	In Section \ref{subsec: simu_illustration}, we present examples that illustrate the hyperparameter tuning procedure for mixing PPD value chains;
	in Section \ref{subsec: sim_predictive}, we compare \texttt{BuLTM} with other competitors.

	\noindent{\textbf{Simulation setting}}. 
	Our data setting covers two domains of response: (a) a real-valued response and (b) a positive response.
	In both settings, the simulated data are generated from model \eqref{basicLTM}. 
	In setting (a), the transformation $h$ is set to be the inverse (signed) Box-Cox function with $\lambda = 0.5$, the same as the \texttt{box-cox} setting in \cite{kowal2024monte}. 
	Two types of model error distributions are considered:
	(a.1) a standard normal distribution, a benchmark setting; 
	(a.2) a normal mixture distribution, a model-misspecification setting for semiparametric methods. 
	In setting (b), we allow the observations to be right-censored in a noninformative censoring scheme. 
	Two types of model error distributions are considered:
	(b.1) an extreme-value distribution, the popular proportional hazard setting \citep{cox1972regression}; 
	(b.2) a normal mixture distribution, a model-misspecification setting for semiparametric
	methods.
	
	In each simulation, we generate $n=200$ samples as the training set and $n_{\text{test}} = 20$ independent samples as the test set, and independently replicate the simulation runs $100$ times. 
	An additional simulation setting (c) generated from nonlinear transformation models is deferred to \textit{Supplement} C.

	\subsection{Visualization of the adaptive scheme for prior adjustment}
	\label{subsec: simu_illustration}
	We use examples from Setting (a.2) to illustrate the adaptive prior adjustment Algorithm \ref{alg:Tuning} achieving well mixed PPD value chains. 
	We examine two aspects of the mixing of PPD value chains: i) visualizing the trace plots of MCMC chains; ii) checking whether the rank normalized $\hat{R}$ statistic \citep{vehtari2021rank} exceeds $1.01$.
	We use the chains of the sum of the log posterior density of the observed data $\mathcal{D}$ given by MCMC samples, denoted by \texttt{lp\_\_} in \texttt{Stan}, as an alternative to the PPD value chains for simplicity. 
	In each example, we use the $\tau$-Sigmoid transformation with $\tau = 5$ as the data transformation mentioned in Section \ref{subsec: ModTrans}, and set the chain length for tuning as the default $N_d = 500$. 
	Examples in other settings and sensitivity analysis of the chain length are deferred to \textit{Supplement} C and D respectively. 
	In all examples, we set the hyperparameter $\rho = 1$ as stated in Section \ref{subsec: sufficient informative}. 
	
	\begin{example}
		\label{example: 1}
		Set initial values $(\eta_0, \zeta_0) = (0.01, 0.01)$, yielding very vague priors for $H$ and $F_\xi$.
		After drawing MCMC samples, we compare the within-chain variance of $H(s_{j_0})$ with the inverse of the prior information level $\tilde{\mathcal{V}}_{s_{j_0}}$. 
		As shown in Figure \ref{zeta_plot_example1}, the within-chain MCMC variance is much less than $\tilde{\mathcal{V}}_{s_{j_0}}(\eta_0, \zeta_0)$. 
		Thus, we assert that the BNPs are not sufficiently informative to achieve mixing of PPD value chains. 
		As evidence, Figure \ref{lp_example1} shows that the MCMC traces of  ``lp\_\_" are poorly mixed, with $\hat{R} = 1.25$. 
		Accordingly, the effective sample size (ESS) of \texttt{lp\_\_} is only 7, which is definitely insufficient. 
	\end{example}
	
	\begin{figure}[!htb]
		\centering
		\subfigure[]{
			\begin{minipage}[t]{0.45\linewidth}
				\centering
				\includegraphics[width=2.0in]{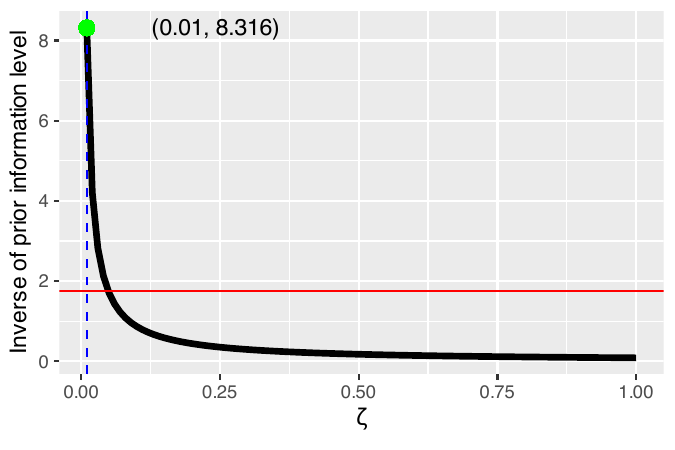}
				\label{zeta_plot_example1}
			\end{minipage}
		}
		\vspace{-.5cm}
		\subfigure[]{
			\begin{minipage}[t]{0.45\linewidth}
				\centering
				\includegraphics[width=2.0in]{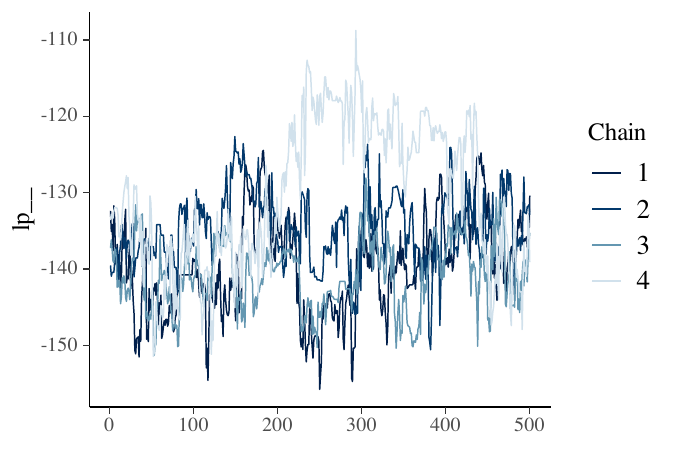}
				\label{lp_example1}
			\end{minipage}
		}
		
		\caption{\footnotesize (a) The curve of $\tilde{\mathcal{V}}_{s_{j_0}}(\eta, \zeta)$ with $\eta = 0.01$ fixed; horizontal line: the within-chain MCMC variance sampled with hyperparameters $(\eta, \zeta) = (0.01, 0.01)$. (b) Trace plot of chains of \textit{lp\_\_}. }
		\label{Fig: example 1}
	\end{figure}
	
	\begin{figure}[!htb]
		\centering
		\subfigure[]{
			\begin{minipage}[t]{0.45\linewidth}
				\centering
				\includegraphics[width=2.0in]{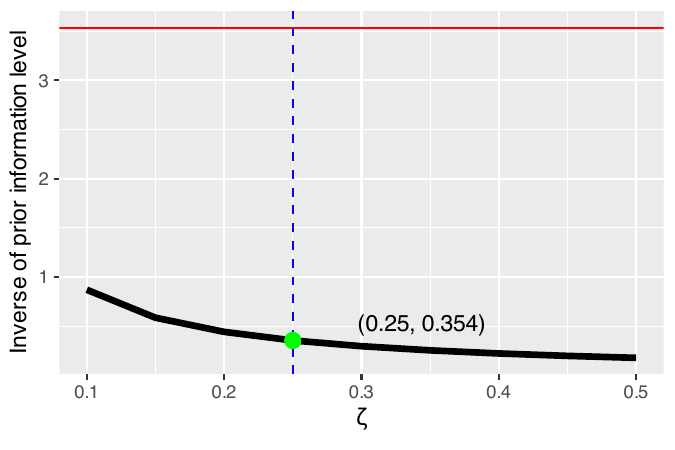}
				\label{zeta_plot_example2}
			\end{minipage}
		}
		\vspace{-.5cm}
		\subfigure[]{
			\begin{minipage}[t]{0.45\linewidth}
				\centering
				\includegraphics[width=2.0in]{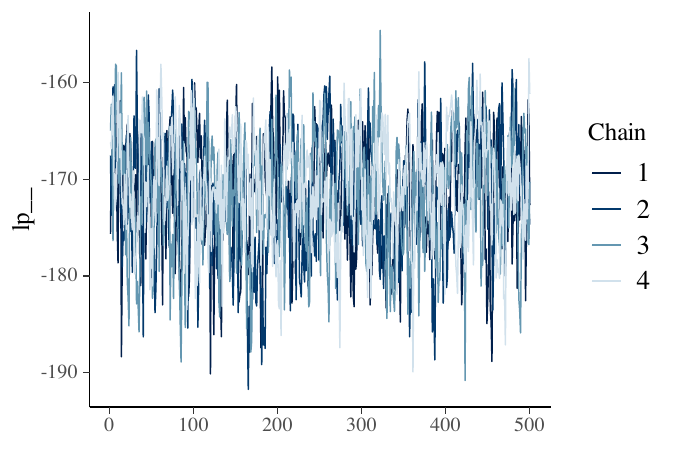}
				\label{lp_example2}
			\end{minipage}
		}
		
		\caption{\footnotesize (a) The curve of $\tilde{\mathcal{V}}_{s_{j_0}}(\eta, \zeta)$ with $\eta = 0.01$ fixed; horizontal line: the within-chain MCMC variance sampled with hyperparameters $(\eta, \zeta) = (0.01, 0.25)$. (b) Trace plot of chains of \textit{lp\_\_}. }
		\label{fig: example 2}
	\end{figure}

	\begin{example}
		\label{example: 2}
		Figure \ref{zeta_plot_example1} illustrates hyperparameter tuning on $(\eta, \zeta)$. 
		With $\eta = 0.01$ fixed, candidates for updating $\zeta$ should enable the curve of $\tilde{\mathcal{V}}_{s_{j_0}}$ against $\zeta$  to fall below the within-chain MCMC variance (the horizontal line). 
		Meanwhile, the curve falls sharply on the interval $(0,. 0.25]$, and decreases gently on the interval $[0,25, 1]$. 
		Consequently, we set $(\eta, \zeta) = (0,01, 0.25)$ as the updated tuning hyperparameters. 
		Figure \ref{zeta_plot_example2} shows that the within-chain MCMC variance exceeds $\tilde{\mathcal{V}}_{s_{j_0}}$, indicating that the BNPs are sufficiently informative. 
		As a result, the MCMC chains of \textit{lp\_\_} mix well as shown by Figure \ref{lp_example2} with $\hat{R} = 1.006$, demonstrating the efficacy of the tuning procedure. 
		Meanwhile, the obtained ESS of 520 is sufficient to represent the log posterior densities. 
	\end{example}
	According to \cite{margossian2023many}, 
	reliability diagnostics ($\hat{R}$ and ESS) demonstrate that the estimated PPD in Example \ref{example: 2} is reliable. 
	This example also illustrates that the hyperparameter configuration $(\eta, \zeta, \rho) = (0.01, 0.25, 1)$ achieves mixing of PPD value chains under setting (a.2). 
	We further use this hyperparameter setting as the initial values throughout all numerical studies. 
	Interestingly, this hyperparameter setting achieves mixing of PPD value chains in all our numerical studies.

	We do NOT recommend increasing prior information by increasing $\eta$ only since in \eqref{lowerbound: approximation}, $|\partial \tilde{\mathcal{V}}_{s_{j_0}}/\partial \zeta|$ is much larger than $|\partial \tilde{\mathcal{V}}_{s_{j_0}}/\partial \eta|$. 

	\begin{figure}
		\centering
		\subfigure[]{
			\begin{minipage}[t]{0.45\linewidth}
				\centering
				\includegraphics[width=2.0in]{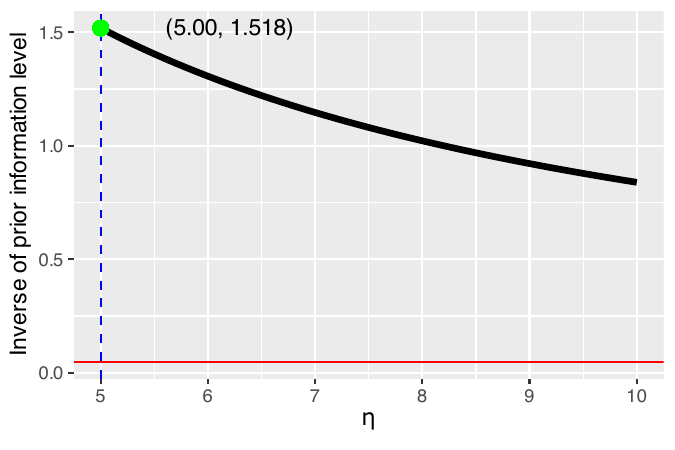}
				\label{eta_plot_example3}
			\end{minipage}
		}
		\vspace{-.5cm}
		\subfigure[]{
			\begin{minipage}[t]{0.45\linewidth}
				\centering
				\includegraphics[width=2.0in]{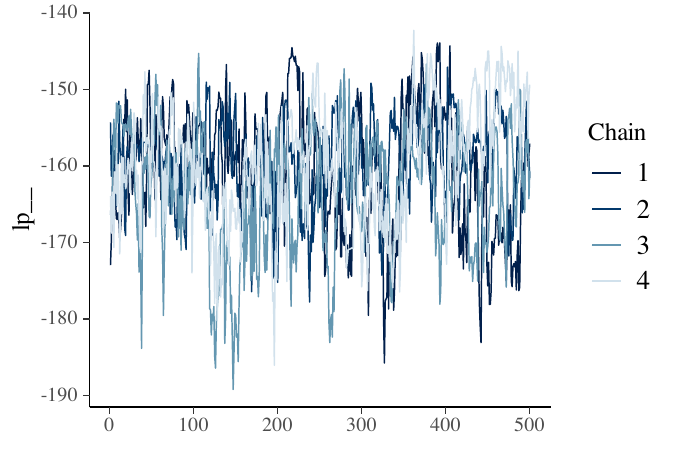}
				\label{lp_example3}
			\end{minipage}
		}
		
		\caption{\footnotesize (a) The curve of $\tilde{\mathcal{V}}_{s_{j_0}}(\eta, \zeta)$ with $\zeta = 0.01$ fixed; horizontal line: the within-chain MCMC variance sampled with hyperparameters $(\eta, \zeta) = (5, 0.01)$. (b) Trace plot of chains of \textit{lp\_\_}. }
		\label{fig: example 3}
	\end{figure}

	\begin{example}
		As a counter example, based on the initial values in Example \ref{example: 1}, we fix $\zeta = 0.01$ and set $\eta = 5$, yielding a vague prior for $f_\xi$ and a highly informative prior for $H$.
		Unfortunately, this hyperparameter setting is insufficiently informative since the within-chain MCMC variance falls below $\tilde{\mathcal{V}}_{s_{j_0}}$ as shown in Figure \ref{eta_plot_example3}. 
		Thus, the PPD value chains are poorly mixed as shown by Figure \ref{lp_example3}, with $\hat{R} = 1.056$. 
		
	\end{example}
	Note that it is meaningless to further increase $\eta$ since too informative a prior for $H$ leads to extremely slow sampling. 
	This counter example illustrates another aspect of the utility of the proposed prior adjustment scheme:  
	it strikes a balance between the noninformative priors that yield poor mixing and the too informative priors that hinder sampling.


	\subsection{Predictive capability evaluation}
	\label{subsec: sim_predictive}
	This subsection evaluates the predictive capability of the \texttt{BuLTM} package under transformation models. 
	In \texttt{BuLTM}, we use the estimated PPD as the predictive distribution; 
	for the predicted value, in Setting (a), we use the median of the estimated PPD; in Setting (B), we use the quantile of the estimated PPD that corresponds to the censoring rate.

	\noindent{\textbf{Competitors}}. 
	In Setting (a), competitors are  the packages or open-source algorithms for fitting semiparametric transformation models. 
	All competitors adopt the standard normal distribution as the reference distribution. 
	\begin{itemize}
		\item \texttt{R} package \texttt{SeBR} \citep{kowal2024monte}.
		We use the empirical CDF of predicted Monte Carlo samples as the predictive distribution, and use the default predicted value. 
		
		\item \texttt{R} code \texttt{BCTM.lin} \citep[BCTM,][]{carlan2024bayesian}. 
		We use the default PPD as the predictive distribution and use predictive median as the predicted value.

		\item Add-on \texttt{R} package \texttt{tram} \citep{siegfried2023distribution}. 
		An add-on package in \texttt{mlt} \citep{hothorn2020most}, the implementation of conditional transformation models \citep{hothorn2014conditional, hothorn2018most}.
		We use the default predictive distribution and the predicted meadian as the predicted value.
		
		\item \texttt{Python} library \texttt{liesel\_ptm} \citep[PTM,][]{brachem2024bayesian}. 
		We use the default predictive distribution and predicted value. 
	\end{itemize}

	\noindent{In Setting (b), competitors are  the packages for semiparametric survival models.}
	
	\begin{itemize}
		\item \texttt{R} package \texttt{spBayesSurv} \citep{zhou2020spbayessurv}, a Bayesian package for semiparametric survival model fitting and model selection.

		\item \texttt{R} package \texttt{mlt} \citep{hothorn2020most}. 
		
		\item \texttt{R} package \texttt{TransModel} \citep{zhou2022transmodel}, fitting a semiparametric transformation model based on \cite{chen2002semiparametric}.
		
	\end{itemize}
	In Setting (b.1), all competitors use the correct reference distribution; 
	in Setting (b.2), competitors use the reference distribution selected by \texttt{spBayesSurv} following the model selection procedure in \cite{zhou2018unified}.

	\noindent{\textbf{Assessments}}. 
	We evaluate two capabilities: (i) the capability of recovering the predictive distribution; (ii) the performance of a single predicted value. 
	Capability (i) is evaluated by the root integrated mean square error (RIMSE) between estimated predictive distribution $\hat{f}$ and the truth $f$: $\text{RIMSE}(\hat{f}, f) = \sqrt{\int_a^b (\hat{f}(s) - f(s))^2ds}$ on an interval $(a, b)$. 
	For capability (ii) assessment, we use the mean absolute error (MAE) in Setting (a), and the C-index \citep{harrell1982evaluating} in Setting (b). 
	
	\noindent{\textbf{Knot interpolation with censored data}}. 
	{
		We introduce a knot interpolation procedure to incorporate information from censored observations. 
		Let $\tilde{y}$ be the uncensored transformed observations, and $\tilde{y}_c$ be the collection of both censored and uncensored observations. 
		We begin with the $N_I$ interior knots specified by the quantiles of the uncensored observations, denoted by $s_0, \ldots, s_{N_i-1}$. 
		Then we interpolate the quantiles of $\tilde{y}_c$ that are located far from the same quantiles of $\tilde{y}$ and a complement of the knots. 
		We summarize the two-step procedure in Algorithm \ref{alg:quantile-knot}. 
		An example that illustrates and visualizes the procedure is deferred to \textit{Supplement} E. 
		
		\begin{algorithm*}[!htb]\footnotesize
			\caption{Knot interpolation with censored data}\label{alg:quantile-knot}
			\begin{algorithmic}[1]
				\State Configure initial knots. 
				Let $N_I > 1$ be the number of initial knots. 
				For $j=0, \ldots, N_I - 1$, let $s_j = \hat{Q}_{\tilde{y}}(j/N_I)$.
				Sort initial knots $0<s_0<\cdots<s_{N_I - 1} < \tau$. 
				\If{ the transformed observations $\tilde{y}$ are right-censored } 
				\State 
				Let $\tilde{y}_c$ be the collection of all observations, and $\tilde{y}$ be the uncensored observations. 
				For $s_j$ such that $|\hat{F}_{\tilde{y}_c}(s_j) - \hat{F}_{\tilde{y}}(s_j)| \ge 0.05$, interpolate a new knot $s_j^* = \hat{Q}_{\tilde{y}_c}(j/N_I)$. 
				\State Output sorted series of $\{s_1, \ldots, s_j, s_j^*, \ldots, s_{N_I-1}\}$ as final interior knots. 
				\EndIf
			\end{algorithmic}
		\end{algorithm*}
	}
	
	\noindent{\textbf{Setting (a)}}. 
	\begin{figure}[!htp]
		\centering
		\subfigure[]{
			\begin{minipage}[t]{0.45\linewidth}
				\centering
				\includegraphics[width=1.75in]{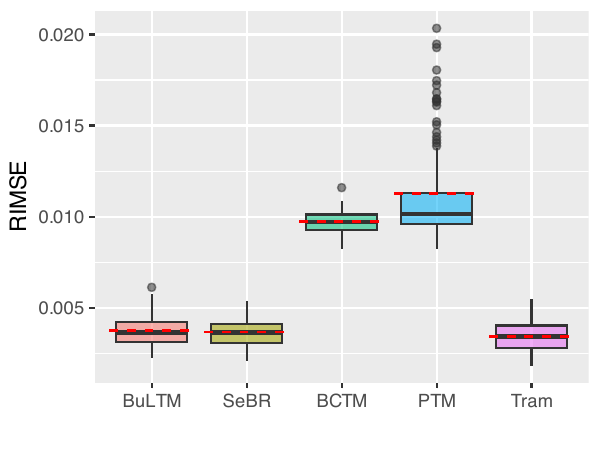}
				\label{RIMSE_Gauss}
			\end{minipage}
		}
		\vspace{-.5cm}
		\subfigure[]{
			\begin{minipage}[t]{0.45\linewidth}
				\centering
				\includegraphics[width=1.75in]{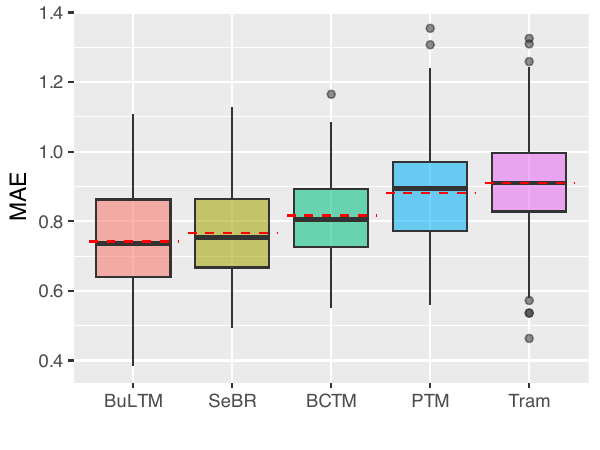}
				\label{MAE_Gauss}
			\end{minipage}
		}
		
		\caption{\footnotesize Box-plots of predictive assessments under Setting (a.1). (a), RIMSE; 
			(b), MAE. }
		\label{Box_Gauss}
	\end{figure}
	The box-plots of assessments among all replicative simulations in Settings (a.1) and (a.2) are presented in Figures \ref{Box_Gauss} and \ref{Box_Mix} respectively.
	In setting (a.1), where all competitors correctly specify the model, \texttt{BuLTM} is competitive with \texttt{SeBR} in recovering predictive distributions (two-sided paired t-test $p$-value: 0.15 against \texttt{SeBR}), and significantly outperforms the remaining competitors except \texttt{tram}.
	However, \texttt{BuLTM} outperforms all competitors including \texttt{tram} in fitting the predicted values (one-sided paired t-test $p$-values : $1.147 \times10^{-5}$ against \texttt{SeBR}; $2.14\times 10^{-15}$ against BCTM). 
	In setting (a.2), where all semiparametric methods encounter model misspecification, \texttt{BuLTM} significantly outperforms all competitors in both recovering predictive distributions and fitting the predicted values (one-sided paired t-test $p$-values : $0.0001$ against \texttt{SeBR}; $3.16\times 10^{-8}$ against BCTM).

	\begin{figure}[!htp]
		\centering
		\subfigure[]{
			\begin{minipage}[t]{0.45\linewidth}
				\centering
				\includegraphics[width=1.75in]{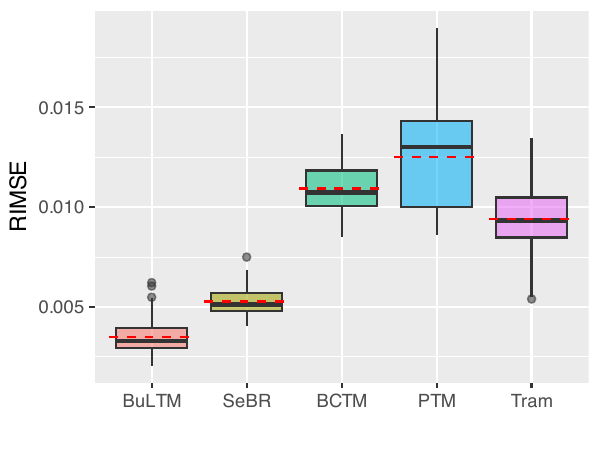}
				\label{RIMSE_Mix}
			\end{minipage}
		}
		\vspace{-.5cm}
		\subfigure[]{
			\begin{minipage}[t]{0.45\linewidth}
				\centering
				\includegraphics[width=1.75in]{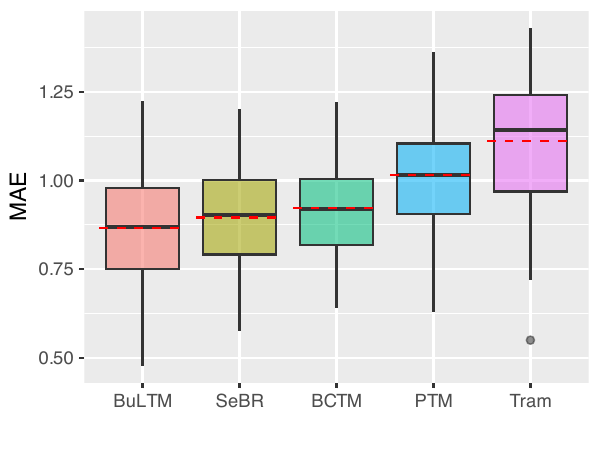}
				\label{MAE_Mix}
			\end{minipage}
		}
		
		\caption{\footnotesize Box-plots of predictive assessments under Setting (a.2). (a), RIMSE; 
			(b), MAE. }
		\label{Box_Mix}
	\end{figure}

	\begin{figure}[!htp]
		\centering
		\subfigure[]{
			\begin{minipage}[t]{0.45\linewidth}
				\centering
				\includegraphics[width=1.75in]{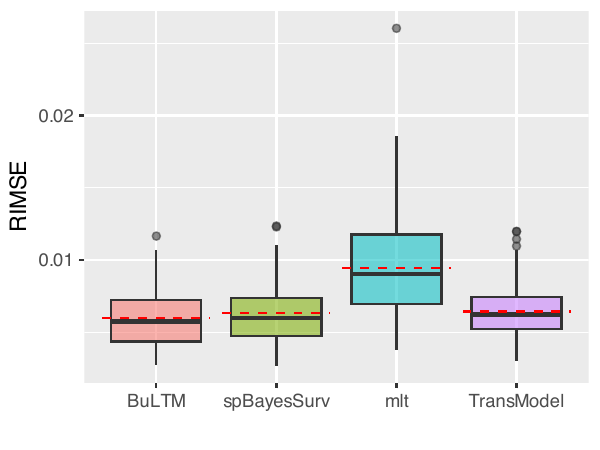}
				\label{RIMSE_PH}
			\end{minipage}
		}
		\vspace{-.5cm}
		\subfigure[]{
			\begin{minipage}[t]{0.45\linewidth}
				\centering
				\includegraphics[width=1.75in]{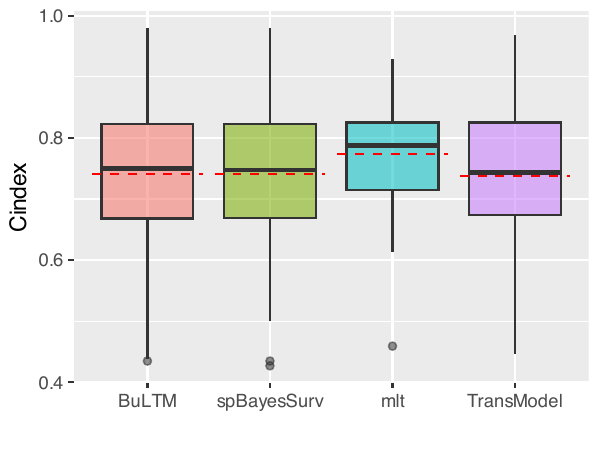}
				\label{Cind_PH}
			\end{minipage}
		}
		
		\caption{\footnotesize Box-plots of predictive assessments under Setting (b.1). (a), RIMSE; 
			(b), C-index. }
		\label{Simu_PH}
	\end{figure}
	
	\noindent{\textbf{Setting (b)}}. The box-plots of assessments among all replicative simulations in Settings (b.1) and (b.2) are presented in Figures \ref{Simu_PH} and \ref{Simu_Non} respectively. 
	In Setting (b.1), the commonly used proportional hazard setting, \texttt{BuLTM} significantly outperforms \texttt{spBayesSurv} (one-sided paired t-test $p$-values: $0.0002$) and \texttt{TransModel} ($2.03\times 10^{-5}$) in recovering the predictive distributions, and is comparable with \texttt{spBayesSurv} and \texttt{TransModel} in C-index, while \texttt{mlt} outperforms. 
	In Setting (b.2), where the competitors encounter model misspecification, \texttt{BuLTM} significantly outperforms the competitors in recovering the predictive distributions, and slightly outperforms in C-index. 
	
	In summary, \texttt{BuLTM} is robust against model misspecification for both real-valued and positive responses, and is competitive (generally significantly outperforms) in both predictive distribution recovery and single-value predictions.

	\begin{figure}[!htp]
		\centering
		\subfigure[]{
			\begin{minipage}[t]{0.45\linewidth}
				\centering
				\includegraphics[width=1.75in]{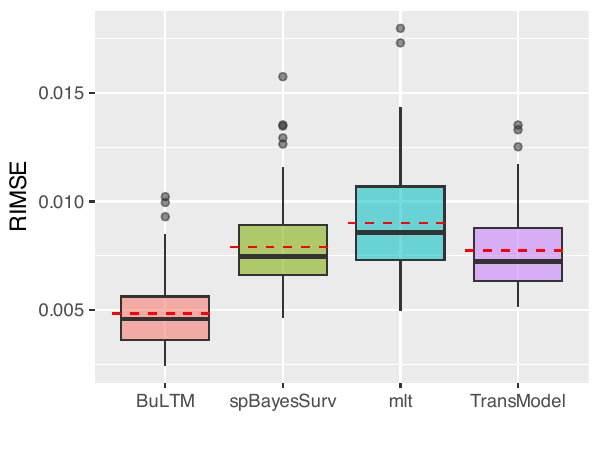}
				\label{RIMSE_Non}
			\end{minipage}
		}
		\vspace{-.5cm}
		\subfigure[]{
			\begin{minipage}[t]{0.45\linewidth}
				\centering
				\includegraphics[width=1.75in]{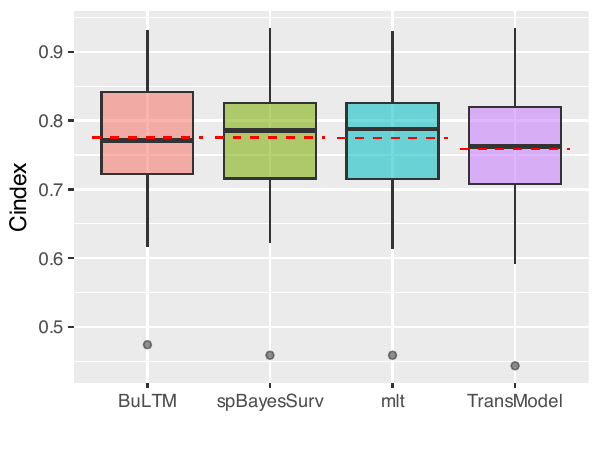}
				\label{Cind_Non}
			\end{minipage}
		}
		
		\caption{\footnotesize Box-plots of predictive assessments under Setting (b.2). (a), RIMSE; 
			(b), C-index. }
		\label{Simu_Non}
	\end{figure}
	

	\section{Applications}
	\label{sec:app}
	
	\subsection{Auto MPG data}
	We first apply \texttt{BuLTM} to Auto MPG \citep{MPGdata}, a benchmark machine learning dataset. 
	The response is city-cycle fuel consumption in miles per gallon (MPG) and the predictors are 3 multivalued discrete and 5 continuous covariates. 
	We preprocess the data by transforming all continuous predictors to $(0, 1)$ and center the response to $\mathbb{R}$. 
	We split the data into a 90\% training set and a 10\% test set and repeat the split for 10 runs to compare the out-of-sample predictive performance of \texttt{BuLTM} with other competitors. 
	We allow for nonlinear covariate effects through covariate transformation with an additive structure such that $ h(y) = \sum_{j=1}^8 f_j(z_j) + \epsilon$, where $f_j(z_j) = \sum_{k=1}^K \beta_{jk} \phi_j(z_j)$ and the $\phi_j$ are basis functions. 
	On this dataset, we select the Fourier basis and set $K = 8$. 
	We apply the nonlinear model transformation to both \texttt{BuLTM} and \texttt{SeBR}, namely \texttt{BuLTM.nonlin} and \texttt{SeBR.nonlin}, respectively.

	Two metrics on the test sets are used for assessment: i) the MAE between the predicted and the true values, and ii) the coverage probability (CP) of the 95\% prediction intervals. 
	We compare \texttt{BuLTM} with the competitors in Section \ref{sec:sim} under simulation Setting (A). 
	To reduce computational burden, we fit the linear model for BCTM without the smooth transformation. 
	Thus, the results of BCTM can also be treated as a baseline model.

	\begin{figure}[!htp]
		\centering
		\subfigure[]{
			\begin{minipage}[t]{0.45\linewidth}
				\centering
				\includegraphics[width=2.0in]{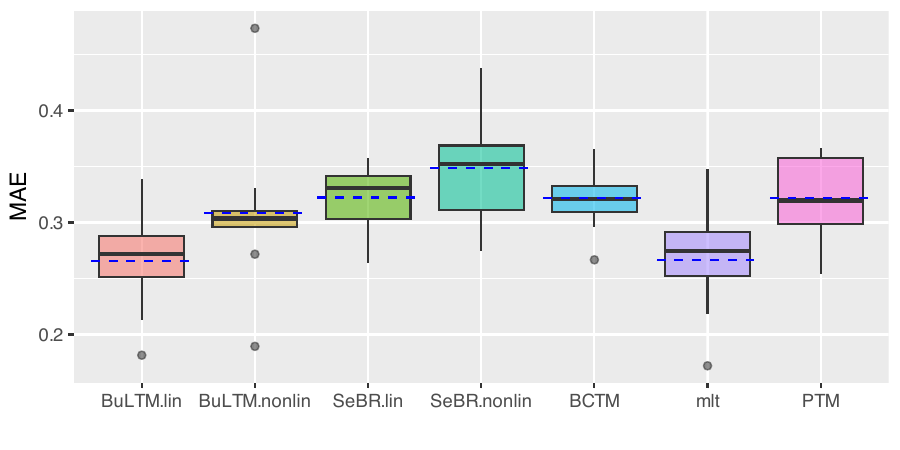}
				\label{MAE_MPG}
			\end{minipage}
		}
		\vspace{-.5cm}
		\subfigure[]{
			\begin{minipage}[t]{0.45\linewidth}
				\centering
				\includegraphics[width=2.0in]{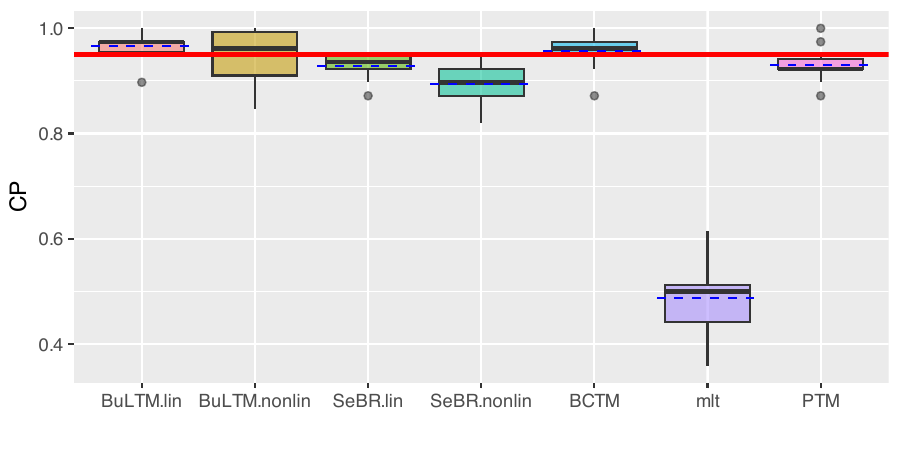}
				\label{CP_MAE}
			\end{minipage}
		}
		
		\caption{\footnotesize Box-plots of predictive assessments on the MPG dataset. (a), MAE; 
			(b), CP, the horizontal line is the nominal level of coverage. }
		\label{App: MPG}
	\end{figure}
	
	The box-plots of the assessment metrics are presented in Figure \ref{App: MPG}. 
	We find that for both \texttt{BuLTM} and \texttt{SeBR}, the linear model enjoys lower MAE than the nonlinear model, indicating that a linear transformation model is adequate to fit the MPG data. 
	On the MAE metric, \texttt{BuLTM.lin} is competitive with \texttt{mlt} (two-sided paired t-test $p$-value: 0.062) and significantly outperforms the other competitors. 
	On the CP metric, both \texttt{BuLTM} and \texttt{BuLTM.nonlin} achieve the nominal coverage, while \texttt{mlt} fails to do so. 
	This real-word example demonstrates the superiority of \texttt{BuLTM} in both fitting predicted values and recovering predictive distributions for real-valued data.

	\subsection{Heart failure clinical records data}
	The second real-world example is the heart failure clinical records data. 
	The dataset records 299 heart failure patients collected at the Faisalabad Institute of Cardiology and at the Allied Hospital in Faisalabad, from April to December in 2015 \citep{ahmad2017survival}. 
	The dataset consists of 105 women and 194 men, with a range of ages between 40 and 95 years old. 
	In the dataset, 96 subjects are recorded as dead and the remaining 203 are censored, leading to a censoring rate of 67.9\%, which is relatively high. 
	The dataset contains 11 covariates reflecting subject's clinical, body, and lifestyle information. 
	In this dataset, \texttt{spBayesSurv} selects the PH model and thus, \texttt{TransModel} specifies $r=0$, and \texttt{mlt} uses the reference distribution \texttt{"MinExtrVal"}. 
	We conduct 10 runs of 5-fold cross validation. 
	The results of estimation of regression coefficients are deferred to \textit{Supplement} F.1. 
	
	On the heart failure dataset, we use two metrics to assess the predictive capabilities: 
	i) the C-index, where we use the $70\%$ quantile of the predictive distribution (close to the censoring rate) as the predicted survival time of a future observation;
	ii) the Integrated Brier Score \citep[IBS][]{graf1999assessment} on the follow-up time interval (the lower the IBS, the better the prediction). 
	
	\begin{figure}[!htp]
		\subfigcapskip=-10pt
		\centering
		\subfigure[]{
			\begin{minipage}[t]{0.45\linewidth}
				\centering
				\includegraphics[width=1.5in]{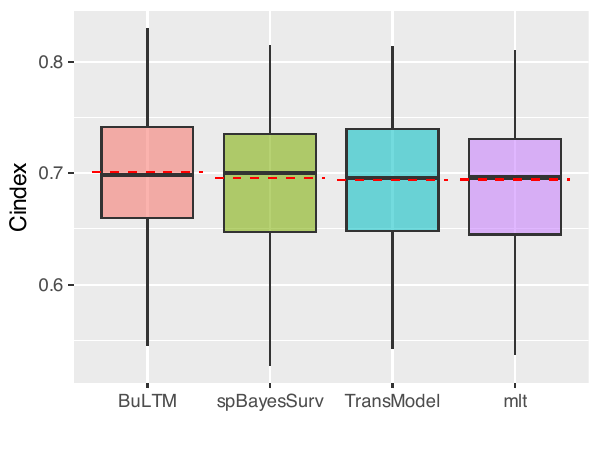}
				\label{CindexHeart}
			\end{minipage}
		}
		\vspace{-.5cm}
		\subfigure[]{
			\begin{minipage}[t]{0.45\linewidth}
				\centering
				\includegraphics[width=1.5in]{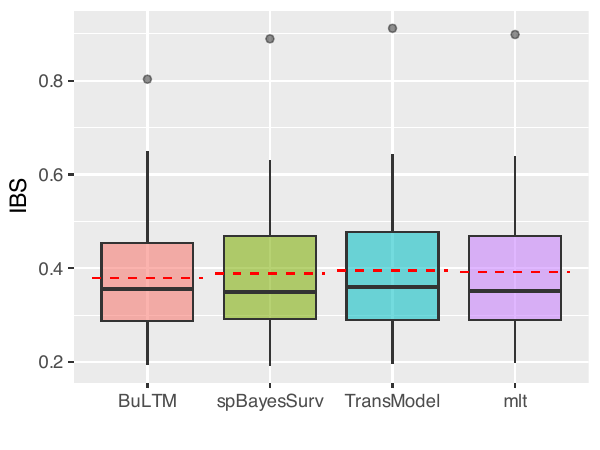}
				\label{IBSHeart}
			\end{minipage}
		}
		\caption{\footnotesize{Prediction comparison between BuLTM, spBayesSurv, and TransModel on the heart failure dataset; (a), C index; (b), Integrated Brier score; red dashed lines: the mean of the metrics. } }
		\label{Transplant}
	\end{figure}

	Box-plots of the assessment metrics are presented in Figure \ref{Transplant}. 
	\texttt{BuLTM} significantly outperforms other competitors in both C-index and IBS:  
	for C-index, the one-sided paired t-test $p$-values are $0.01$ against \texttt{spBayesSurv}, $0.0038$ against \texttt{TransModel}, and $0.0046$ against \texttt{mlt}; 
	for IBS, the one-sided paired t-test $p$-values are $0.002$ against \texttt{spBayesSurv}, $1.47 \times 10^{-5}$ against \texttt{TransModel}, and $2.8 \times 10^{-4}$ against \texttt{mlt}. 
	This example demonstrates the superiority of \texttt{BuLTM} in the  prediction of censored data.

	\section{Discussion}
	\label{sec:disc}

	Under unidentified transformation models, the proposed sufficient informativeness criterion extends the Gelman-Rubin (G-R) statistic \citep{gelman1992inference} from MCMC checking to covering prior adjustment, by \textit{taking another view of MCMC convergence}. 
	The G-R statistic diagnoses whether the MCMC transitions converge to the stationary distribution \citep{roy2020convergence}. 
	By contrast, we examine whether the within-chain MCMC variance is close to the true posterior variance (dominated by the inverse of prior information level; refer to Theorem \ref{theo: posterior variance}) under unidentified transformation models.

	As the AE and referees' sharp insights have helped us clarify, the application scope of \texttt{BuLTM} covers general data domains through the multiplicative error working model. 
	We want to emphasize that \texttt{BuLTM} offers a robust  alternative toolbox to survival analysis in addition to \texttt{spBayesSurv} \citep{zhou2018unified} and \texttt{TransModel} \citep{zhou2022transmodel}: 
	it can estimate conditional survival functions and conditional hazard functions for future data, and provides a reliable Bayesian estimate of identified regression coefficient $\bm{\beta}$ with a tractable unit-norm constraint.

	
	In our unidentified scenario caused by a flat likelihood, only a few discussions have mentioned that ``weakly informative priors" may resolve the poor mixing phenomenon (\cite{reich2019bayesian}; \cite{mcelreath2020rethinking}), but they did not quantify how ``weak" the priors can be. 
	In this sense, this article might be the first to quantitatively link prior informativeness with MCMC mixing under unidentified models: 
	we rigorously demonstrate that the variance of multi-modal posteriors does not vanish with increasing data size, but rather is dominated by the prior information level; 
	and we comprehensively illustrate how to achieve MCMC mixing by increasing the prior information through an analytic expression for prior information, an algorithm of hyperparameter tuning, and visualization of the whole procedure. 
	
	Our method addresses the poor mixing of PPD value chains under unidentified transformation models. 
	Nonetheless, mixing of other parameters such as the Dirichlet process mixture components remains unsolved due to the label-switching issue. 
	It is anticipated that an ordered Dirichlet process stick-breaking construction \citep{zarepour2012rapid} can resolve the problem and speed up our MCMC sampling, but the implementation is so far unavailable in \texttt{rstan} since it needs the Boost \texttt{C++} libraries. 
	Our method may be further extended to other unidentified models such as latent Dirichlet allocation \citep[LDA,][]{blei2003latent} and Bayesian additive regression trees \citep[BART,][]{chipman2012bart}, where new prior elicitation and new quantifications of prior information are needed.


	\newpage
	\begin{center}
		{\large\bf SUPPLEMENTARY MATERIAL}
	\end{center}
	
\section{Techniacal proofs}
\subsection{Proof of Proposition 1}
\begin{proof}
Suppose $H(0) = a$ , where $a$ is a positive constant. 
Since $\widetilde{y}$ lay on interval $(0, \tau)$, $Pr\{\widetilde{y} \ge 0\}=1$. 
Thus we have
$$
Pr\{\widetilde{y}\ge 0\} = \int_{D}  Pr\{\widetilde{y}\ge 0|\bm{z} = \bm{s}\}f_{\z}(\bm{s}) d\bm{s} = 1,
$$
where $D$ is the support of covariate $\z$ and $f_{\z}$ denotes the density of $\z$.  
Since $H$ is monotonic, we have 
$$
Pr\{\widetilde{y}\ge 0|\z=\bm{s}\} = Pr\{H(\widetilde{y})\ge a|\z=\bm{s}\}=Pr\{\xi\exp(\bbeta^T\bm{s})\ge a\} = Pr\{\xi\ge a\exp(-\bbeta^T \bm{s})\}.
$$ 
As a counterexample,  suppose the covariate $\z  \sim N(0, 1)$ is univariate, the model error $\xi \sim \exp(1)$, and $\bbeta = \beta_1=-1$. 
Since $\xi$ and $\z$ are independent, we have 
$
Pr\{\widetilde{y}\ge 0\} = \int_{\mathbb{R}}\int_{a\exp(z)}^{+\infty} \exp(-s) \phi(z)dsdz <1, 
$
where $\phi(\cdot)$ denotes the density of $N(0,1)$. 
This contradicts the fact that $Pr\{\widetilde{y}\ge 0\}=1$. 
\end{proof}

\subsection{Propositions of the quantile-knot I-spline prior}
\label{propSpline}
We first present the following proposition, which reveals the relationship between the proposed quantile-knot I-splines prior and the L\'{e}vy process.
Such a proposition is the key to prove the results like the Bernstein-von Mises theorem. 
\begin{proposition}[I-splines prior and L\'{e}vy process]
\label{prop:SplineOrdinator}
Suppose all transformed data $\widetilde{y}_i$ are observed distinctly on $ (0, \tau)$. 
For $H(s)$ modeled by quantile-knot I-spline prior with fixed knots taken on $s_1 < s_2 < \cdots < s_J$ and prespecified smooth order $r$, 
there exists a L\'{e}vy process $\mathcal{H}$ such that $H(s_j)$ are samples of $\mathcal{H}(s_j)$, for $j=1, \ldots, J$. 
Specifically, given that $\alpha_j \sim \text{Gamma}(a, b)$, there exists a Gamma process $\mathcal{H} = \Gamma \mathcal{P}(bc(s), b)$ from which $H(s_j)$ are sampled, where $c(s)$ is a nonnegative nondecreasing function determined by the Gamma hyperparameter $a$  and the I-spline functions $\sum_{j=1}^K B_j(s)$. 
\end{proposition}

Proposition \ref{prop:SplineOrdinator} suggests the use of Gamma hyperprior for $\bm{\alpha}$, since the Gamma process naturally leads to the nonparametric Bernstein-von Mises (BvM) theorem \citep{kim2006bernstein}. 
Meanwhile, it also justifies why we cannot identify $H$ by imposing constrained priors to $\bm{\alpha}$, as the property of independent increments no longer holds. 
In the following, we present the proof of Proposition \ref{prop:SplineOrdinator}. 

\subsubsection*{Properties of I-spline functions}
We begin with some properties of I-splines functions \citep{ramsay1988monotone}. 
Let $s_0=0 < s_1 < s_2 < \cdots < s_J=\tau$ and we get $J$ disjoint partitions  $[0, s_1], (s_1, s_2], \cdots, (s_{J-1}, s_J]$ of $(0, \tau]$. 
Note that each I-spline function starts at $0$ in an initial flat region, increases in
the mid region, and then reaches $1$ at the end \citep{wang2011semiparametric}.
Therefore, the range of all I-spline functions is $[0, 1]$. 
Then we determine the I-spline basis functions with knots $s_0=0 < s_1 < s_2 < \cdots < s_J=\tau$ and smoothness order $r > 1$. 
When $r = 1$, the I-spline functions are defined as the piecewise linear function connecting the knots $s_1, \ldots, s_J$. 
\begin{definition}[Joint and disjoint I-splines]
Two I-spline functions $B_{j_1}(s)$ and $B_{j_2}(s)$ are joint on a certain interval $D_j$ for $j=1, \cdots, J$, if $\exists s' \in D_i$ such that $B_{j_1}(s'), B_{j_2}(s') \in (0, 1)$. 
Otherwise, they are disjoint on $D_i$. 
\end{definition}

\begin{definition}[Crossing of I-splines]
An I-spline function $B_j(s)$ crosses an interval $D_i$ if $\exists s' \in D_i$ such that $0 <B_j(s')<1$.
\end{definition}
The following propositions are direct results of the definition of I-splines functions \citep[Eq. (5)]{ramsay1988monotone}. 
\begin{proposition}
\label{property: crossspline}
For each interval $D_j$, there are at most $r$ I-spline functions crossing the interval. 
\end{proposition}
\begin{proposition}
\label{property: croos spline2}
For $j =1, \ldots, K = J+r$, for the intervals $D_j$, the I-spline function $B_j(s)$ will cross at least one interval and cross no moren than $r$ intervals. 
\end{proposition}

With $J$ interior knots and smooth degree $r$, the I-spline functions are uniquely determined, denoted as $\{B_j(s)\}_{j=1}^{K}$, where $K = J+r$ is the total number of I-spline functions. 
We divide all $K = J+r$ I-spline basis functions into $(r)$ groups. 
For $\iota = 1, \ldots, r$, the $\iota$th group consists of $\{B_\iota, B_{\iota+r}, B_{\iota+2r}, \ldots\}$. 
Propositions \ref{property: croos spline2} and \ref{property: crossspline} guarantee that all I-spline functions within the same group are disjoint. 
That is, for any $D_i$, only one of the I-spline functions within the $\iota$th group crosses the interval $D_i$.  
We define the combination of I-spline functions within the $\iota$th group as
$$
H_\iota(s) = \sum_{k\ge1} \alpha_{\iota+kr}B_{\iota+kr}(s). 
$$
Then $H_\iota(s)$ has independent increments on all knots $s_0=0 < s_1 < s_2 < \cdots < s_J=\tau$, if the coefficients $\{\alpha_{\iota+kr}\}_{k\ge1}$ are independent positive variables. 
Then we rewrite the equation (6) in the manuscript, the I-splines model into the sum of $H_\iota$
\begin{align}
\label{Groupspline}
H(s) = \sum_{j=1}^{K=J+r}\alpha_j B_j(s) = \sum_{\iota=1}^{r} H_\iota(s).
\end{align}

Before we prove Proposition \ref{prop:SplineOrdinator}, we present the definition of the L\'{e}vy process first. 
\begin{definition}[L\'{e}vy process \citep{doksum1974tailfree}]
\label{Levyprocess}
A process $A(s)$ is a L\'{e}vy process such that:, $(i),$ $A(s)$ has independent increments for any $m$ and $0 < s_1 < \cdots < s_m < \infty$; $(ii)$, $A(s)$ is nondecreasing a.s; $(iii)$, $A(t)$ is right continuous a.s; $(iv)$, $A(s) \to \infty$ a.s as $s \to \infty$; $(v)$, $A(s) = 0$ a.s.  
\end{definition}

We are now in a position to complete the proof of Proposition \ref{prop:SplineOrdinator}. 
\begin{proof}
We start from the case where $r = 1$, the linear spline model. 
When $r=1$, all the I-spline functions are disjoint, and therefore, the increments between two knots $s_j$ and $s_{j+q}$ are independent variables for all $q \ge 1$ and do not depend on $s_j$. 
Meanwhile, $H(s_j)$ are nondecreasing and $H(0) =0$ surely. 
Therefore, when $r = 1$, $H(s_j)$ can be sampled from a Levy process. 

When $r > 1$, Eq. \eqref{Groupspline} tells that the I-splines model can be represented as the sum of $r$ groups of disjoint I-spline functions. 
Based on the above proof, the knots of each group of disjoint I-spline functions can be sampled from a L\'{e}vy process. 
Note that a finite sum of independent L\'{e}vy processes is still a L\'{e}vy process. 
Hence, at the knots $s_j$, $H(s_j)$, the sum of $r$ groups of disjoint I-spline functions can be sampled from a L\'{e}vy process too. 

Given that increments on the knots are gamma-distributed, we have 
$
H(s_{j+1}) - H(s_j) \sim \text{Gamma}(ra, b). 
$
Therefore, as those knots $s_j$ are fixed, one can construct a nonnegative nondecreasing function $c(t)$ such that 
$$
c(s_{j+1}) - c(s_{j}) = ra/b
$$
holds for all knots $s_j$. 
By the definition of Gamma processes \citep[pp. 50]{ibrahim2001bayesian}, we conclude that there exists a Gamma process from which $H(s_j)$ are sampled.

\end{proof}

\subsection{Proof of Theorem 1}
Proposition \ref{prop:SplineOrdinator} unveils the relationship between the proposed I-splines model and the L\'{e}vy process and hence enables us to prove Theorem 1.  

\begin{proof}
Note that once $(\bm{p}_0, \bm{\psi}_0, \bm{\nu}_0)$ is specified, the unidentified nonparametric transformation model reduces to an identified semiparametric transformation model \citep{cheng1995analysis}. 
Hence, the ``ground truth'' $H_0$ is unique and fixed. 

Our proof starts from the proportional hazard (PH) case. 
In the PH case, the conditional cumulative hazard function 
$$
\Lambda_{\widetilde{y}|\bm{z}} (s) = H(s) \exp(-\bbeta^T \bm{z}). 
$$
That is, the transformation $H$ plays the role of baseline cumulative hazard function in the PH case. 
Note that with our exponential hyperprior configuration, the posterior samples of $H(s_j)$ are sampled from the posterior of a Gamma process under the PH model (based on Proposition \ref{prop:SplineOrdinator}).
Such insight allows us to employ the Bernstein-von Mises theorem by \cite{kim2006bernstein} in the PH case. 

Let $\mathcal{H}(\cdot)$ be the prior L\'{e}vy process for $H$. 
We introduce the following notations. 
\begin{align*}
U_0(s) = \int_0^s \frac{dH_0(s)}{S^0(s:\bbeta_0)}, ~~
S^0(s:\bbeta) = E(\exp(-\bbeta^T \z) I \{\widetilde{y} > s)\}. 
\end{align*}
Note that the Gamma process satisfies the conditions (C1) and (C2) and our assumptions match conditions (A1) to (A5) in \cite{kim2006bernstein}. 
By integrating out $\bbeta$ in \citet[Theorem 3.3]{kim2006bernstein}, one immediately obtains that 
$$
\pi\{\sqrt{n}(\mathcal{H}(\cdot) - \hat{H}(\cdot) |\mathcal{D}\} \xrightarrow{~~\text{weakly}~~}  \mathcal{GP}(0, K(\cdot, \cdot)), ~K(s, t) = \min(U_0(s), U_0(t)), 
$$
where $\hat{H}$ is the nonparametric maximum likelihood estimator (MLE),
$\mathcal{GP}(0, K(\cdot, \cdot))$ is the Gaussian process with mean function $0$ and the kernel $K$. 
Note that for the knots $s_j$, the posterior of $H(s_j)$ is marginalized from $\mathcal{H}(s_j)$. 
Hence, marginally for all $s_j$, under the PH model, we have 
$$
\pi\{\sqrt{n} (H(s_j) - \hat{H}(s_j))|\mathcal{D}\} \xrightarrow{~~\text{d}~~} N(0, U_0(s_j)). 
$$

Next, we extend the results under PH models to general semiparametric transformation models. 
Under the Weibull mixture model, we have the following relationship between $H$ and the conditional survival function $S_{\widetilde{y}|\z}$, 
\begin{eqnarray}
\label{model:survival}
\begin{aligned}
S_{\widetilde{y}|\z}(t) &= \sum_{l=1}^L p_{l0} \exp\left(-\left\{\frac{H(s)\exp(-\bbeta^T \z)}{\psi_{l0}}\right\}^{\nu_{l0}}\right) \\
&\equiv \sum_{l=1}^L p_{l0} S_l(s)^{\exp(-\nu_{l0} \bm{\beta}^T \bm{z})}, ~~\text{where}~ S_l(s) = \exp\left\{-\left(\frac{H(s)}{\psi_{l0}}\right)^{\nu_{l0}}\right\}. 
\end{aligned}
\end{eqnarray}
Looking into each $S_l(t)$ for $l=1, \ldots, L$ separately, one finds that $S_l(t)$ is the conditional survival function of a Cox's model with reregssion coefficient $\nu_l \bm{\beta}$. 
Hence, we conclude that with the Weibull mixture prior for $f_\xi$, the conditional survival model for $T$ becomes a mixture of PH models. 

For the $l$th component in the mixture of PH models, the baseline cumulative hazard function is 
$
\{H(s)/\psi_{l0}\}^{\nu_{l0}}.
$
Given that $H(t)$ are sampled from a L\'{e}vy prior process, 
$
\{H(s)/\psi_{l0}\}^{\nu_{l0}}
$
is also sampled from a L\'{e}vy process (subordinator). 
Thus, to prove the Bernstein-von Mises results, we only need to verify the conditions (C1) and (C2) in \cite{kim2006bernstein} for $
\{H(s)/\psi_{l0}\}^{\nu_{l0}}.
$

Under the Gamma process prior $\Gamma\mathcal{P}(c(t), b)$ in Proposition \ref{prop:SplineOrdinator}, the L\'{e}vy measure of the prior L\'{e}vy process for $H(s)$ is 
$$
\mu_H(ds, dx) = x^{-1} \exp(-x/b)dxdc(s), ~ s,  x \in (0, +\infty). 
$$
Since $
\{H(s)/\psi_l\}^{\nu_l}
$
is a point-wise transformation on $H(s)$, 
the prior process for $
\{H(s)/\psi_{l0}\}^{\nu_{l0}}
$
share the same Poisson random measure as that of $H(s)$. 
Let  $x(u) = \psi_{l0} u^{1/\nu_{l0}}$ be the inverse map from $H(s)$ to $
\{H(s)/\psi_{l0}\}^{\nu_{l0}}
$ for any fixed $s$, where $u$ is a placeholder. 
The Levy measure of the prior Levy process for $
\{H(s)/\psi_l\}^{\nu_l}
$
is then given by 
\begin{align*}
\mu_{(H, \psi_{l0}, \nu_{l0})}(du, dt) &= \{bx(u)\}^{-1} \exp(-x(u)/b)dx(u) d \{(c(s) /\psi_{l0})^{\nu_{l0}}\}\\
& = \frac{\exp(-\psi_{l0} u ^{1/\nu_{l0}}/ b)}{b\nu_{l0} u}dud \{(c(s) /\psi_{l0})^{\nu_{l0}}\}\\
& \equiv \frac{f(u)}{u} du \lambda(s) ds,  
\end{align*}
where $f(u) = (b\nu_{l0})^{-1}\exp(-\psi_{l0} u ^{1/\nu_{l0}}/ b)$. 
Let $g_s(u) = f(u) / \lambda(s)$. 
Then following \cite{kim2006bernstein}, we need to verify:\\
$i)$, there exists a positive constant $h$ such that 
$$\sup \limits_{s \in [0, \tau], u \in (0, \infty)} (1-u)^{h} g_s(u) < \infty. $$ 
$ii)$, there exits a function $0 < \inf \limits_{s \in [0, \tau]} \gamma(s) < \sup\limits_{s \in [0, \tau]} \gamma(s) < \infty$, such that 
for some $m > 1/2$ and $0 < M < \infty$, 
$$
\sup\limits_{s \in [0, \tau], u \in (0, M)} \left|\frac{g_s(u) - \gamma(s)}{u^m}\right| < \infty. 
$$

Suppose that $ 0 <dc(s) < \infty$ for $s \in [0, \tau]$. 
This can be easily obtained if $(r, a, b, \tau)$ are finite. 
Given that $(\psi_{l0}, \nu_{l0})$ are finite positive constants,  condition $i)$ is verified since for any $s \in [0, \tau]$ and fixed positive constant $h$,  $\lim\limits_{u \to \infty}(1-u)^h g_s(u) \to 0$. 
Condition $ii)$ is verified by selecting a continuous function $\gamma(s)$ such that $\gamma(0) = g_0(0)$. 

Now we are in the position to give the Bernstein-von Mises results. 
Let $\Lambda_{l0}$ be the ``true" baseline cumulative hazard function of the $l$th component of the mixture of PH models under the prespecified ground truth $H_0$. 
That is, 
$$
\Lambda_{l0}(s) = \left\{\frac{H_0(s)}{\psi_{l0}}\right\}^{\nu_{l0}}. 
$$
Suppose that there are $n_l$ data from the $l$th PH component. 
Condition (B4) guarantees that $n_l \to \infty$ as $n \to \infty$. 
Without loss of generality, assume that $\widetilde{y}_{l(1)} < \cdots < \widetilde{y}_{l(n_l)}$ be the ordered sequence transformed responses from the $l$th PH component. 
We define 
\begin{align*}
U_l(s) = \int_0^t \frac{d\Lambda_
	{l0}(s)}{S_l^0(s:\bbeta_0)}, ~~
S_l^0(s:\bbeta) = E(\exp(-\nu_{l0}\bbeta^T \z_{l(1)}) I \{\widetilde{y}_{l(1)} > s)\}), 
\end{align*}
where $\z_{l(1)}$ and $\widetilde{y}_{l(1)}$ and the covariate vector and the transformed response of the $l$th PH component. 
Hence, the Bernstein-von Mises result holds such that
$$
\pi(\sqrt{n_l} (\Lambda_{l}(s_j) - \hat{\Lambda}_l(s_j))|\mathcal{D}) \xrightarrow{~~\text{d}~~~} N(0, U_l(s_j)), 
$$
where $\Lambda_l = (H/\psi_{l0})^{\nu_{l0}}$, and $\hat{\Lambda}_l$ is the corresponding nonparametric MLE. 
Note that under this semiparaemtric setting, $\hat{\Lambda}_l$ converges to the true cumulative hazard 
$
\Lambda_{l0}(t)
$
uniformly \citep{zeng2006efficient}. 
Then by employing the delta method we have, for the $l$th PH component, at each fixed $s_j$,
$$
\pi[\sqrt{n_l} \{\psi_{l0} \Lambda_l^{1/\nu_{l0}}(s_j) - H_0(s_j)\}|\mathcal{D}, \bm{p}_0, \bm{\psi}_0, \bm{\nu}_0] \xrightarrow{~~\text{d}~~~} N\{0, (\psi_{l0}/\nu_{l0})^2H_0(s_j)^{2/\nu_{l0} -2} U_l(s_j)\}.  
$$

Finally, we aggregate all the $L$ components. 
Note that $n_l/n \to p_{l0}$ as $n \to \infty$. 
Hence, we obtain
$$
\pi[\sqrt{n} \{H(s_j) - H_0(s_j)\}|\mathcal{D}, \bm{p}_0, \bm{\psi}_0, \bm{\nu}_0] \xrightarrow{~~\text{d}~~~} \sum_{l=1}^L p_{l0} N(0, p_[l0]^{-1} (\psi_{l0}/\nu_{l0})^2H_0(s_j)^{2/\nu_{l0} -2}U_l(s_j)).  
$$
\end{proof}

\subsection{Proof of Theorem 2}
Recall that specifying $F_{\xi_0}$ is equivalent to specifying the parameters $(\bm{p}_0, \bm{\psi}_0, \bm{\nu}_0)$. 
Consequently, Theorem 2 can be obtained from Theorem 1 by integrating out $(\bm{p}_0, \bm{\psi}_0, \bm{\nu}_0)$ step by step. 
Without loss of generality, we assume that $\bm{p}_0$ is in an decreasing order. 
Any permutation of the DPM index will not change the result on $H$. 

\begin{proof}
\noindent{\textbf{Step 1}}\\
By the total variance formula, we have 
$$
\mathbb{V}(H(s_j)|\mathcal{D}, \bm{\psi}_0, \bm{\nu}_0) = \mathbb{E}\{\mathbb{V}(H(s_j)|\mathcal{D}, \bm{p}_0,  \bm{\psi}_0, \bm{\nu}_0)\} + \mathbb{V}\{E(H(s_j)|\mathcal{D}, \bm{p}_0,  \bm{\psi}_0, \bm{\nu}_0)\}.
$$
Based on Theorem 1, we obtain that
\begin{align}
\mathbb{V}(H(s_j)|\mathcal{D}, \bm{p}_0, \bm{\psi}_0, \bm{\nu}_0) &\to n^{-1}\left(\sum_{l=1}^L p_{l0} \psi_{l0}^2 H_0(s_j)^{2/\nu_{l0}} U_l(s_j)\right) \label{E variance}, \\
\mathbb{E}(H(s_j)|\mathcal{D}, \bm{p}_0,  \bm{\psi}_0, \bm{\nu}_0) &=  H_0(s_j)| \bm{p}_0, \bm{\psi}_0, \bm{\nu}_0. \label{Var E} 
\end{align}
Taking expectation with respect to $\bm{p}$ on the right-hand side of \eqref{E variance}, we obtain that 
$$
\mathbb{E}\{\mathbb{V}(H(s_j)|\mathcal{D}, \bm{p}_0, \bm{\psi}_0, \bm{\nu}_0)\} \to n^{-1}\mathbb{E}_{\bm{p}}\left(\sum_{l=1}^L p_{l0} \psi_{l0}^2 H_0(s_j)^{2/\nu_{l0}} U_l(s_j)\right) \equiv n^{-1} q_j(\bm{\psi}_0, \bm{\nu}_0). 
$$
For the right-hand side of \eqref{Var E}, note that $ H_0$ is uniquely specified once $(\bm{p}_0, \bm{\psi}_0, \bm{\nu}_0)$ are specified. 
Consequently, we have 
$$
\mathbb{V}\{\mathbb{E}(H(s_j)|\mathcal{D}, \bm{p}_0,  \bm{\nu}_0, \bm{\psi}_0)\} =  0. 
$$
\noindent{\textbf{Step 2}}\\
Again employing the total variance formula, we have 
\begin{align*}
\mathbb{V}(H(s_j)|\mathcal{D}, \bm{\nu}_0) &= \mathbb{E}\{\mathbb{V}(H(s_j)|\mathcal{D}, \bm{\psi}_0, \bm{\nu}_0)\} + \mathbb{V}\{\mathbb{E}(H(s_j)|\mathcal{D}, \bm{\psi}_0, \bm{\nu}_0)\}\\
& = n^{-1}\mathbb{E}_{\bm{\psi}_0}q_j(\bm{\psi}_0, \bm{\nu}_0) + \mathbb{V}_{\bm{\psi}_0}(H_0(s_j)|\bm{\nu}_0, \bm{\psi}_0) \\
&\equiv
n^{-1}q_j + \mathbb{V}(H_0(s_j)|\bm{\nu}_0).
\end{align*}
The remaining is to derive $\mathbb{V}(H_0(s_j)|\bm{\nu}_0)$, the variance of ``ground truth" of $H_0(s_j)$. 
Recall that the sample space of ``true" model parameters $(H_0(s_j), \bm{\psi}_0)$ given $\bm{\nu}_0$ is
$$
\Omega(H_0(s_j), \psi_{10}, \ldots, \psi_{L0}) = \{(H_0(s_j), c_{1j}^{1/\nu_{10}}H_0(s_j), \ldots, c_{Lj}^{1/\nu_{L0}}H_0(s_j))\}. 
$$
Define the function $\mathcal{X}: \Omega \to \mathbb{R}_+$ such that $\mathcal{X}\{H_0(s_j), \bm{\psi}_0\} = H_0(s_j)$. 
Obviously, $\mathcal{X}$ is a one-to-one map and thus $\mathcal{X}$ is a random variable ($\mathcal{X}$ is a measurable map). 
Now, we only need to specify the density of $\mathcal{X}$, denoted as $f_{\mathcal{X}}$. 

Base on the definition of $\mathcal{X}$, for $x \in \mathbb{R}_+$, we have 
$$
f_{\mathcal{X}}(x) \propto f_{H_0(s_j)}(x) \prod_{l=1}^L f_{\psi_l}(c_{lj}^{1/\nu_{l0}} x) . 
$$
At its margin, $H_0(s_j)$ is fully determined by the quantile-knot I-sline model. 
Given that the I-spline coefficients $\alpha_j \sim \exp(\eta)$ i.i.d., one can obtain the marginal distribution of $H_0(s_j)$. 
By the Proposition \ref{property: croos spline2}, we have, in priori, 
$$
H_0(s_j) = \sum_{j'=1}^j \alpha_j + \sum_{j'1=1}^r w_{j'}\alpha_{j + j'} 
$$
Hence, given that $\pi(\alpha_j) = \text{Exp}(\eta)$, at the margin of $H_0(s_j)$, the distribution is 
$\pi\{H_0(s_j)\} = \text{Exp}(\eta/(j+ \sum_{j'=1}^r w_{j'}) )$. 

Note that on each margin of $\psi_l$, $\pi(\psi_l) = \text{Exp}(\zeta)$. 
Therefore, we have 
\begin{eqnarray}
\label{density: H and psi}
\begin{aligned}
f_{\mathcal{Y}}(y) &\propto \exp\left(- \frac{y \eta}{j+ \sum_{j'=1}^r w_{j'}}\right)\prod_{l=1}^L \exp(-y\zeta c_{lj}^{1/\nu_l})  \\
& = \exp\left[-y\left\{ \zeta \sum_{l=1}^L c_{lj}^{1/\nu_l} + \frac{\eta}{j+ \sum_{j'=1}^r w_{j'}}
	\right\} \right]. 
\end{aligned}
\end{eqnarray}
Therefore, we have 
$$
\mathbb{V}(H_{0}(s_j)|\bm{\nu}_0) = \mathbb{V}(\mathcal{X})=\left\{\zeta \sum_{l=1}^L c_{lj}^{1/\nu_l} + \eta/(j+ \sum_{j'=1}^r w_{j'}) \right\}^{-2}. 
$$


The next step is to integrate $\bm{\nu}_0$ out. 
We start from a trivial proposition on parameters $\bm{\nu}_0$ of  ``true" $f_\xi$. 
\begin{proposition}
\label{prop: nu}
Under the conditions of Theorem 1, suppose exists a ``true" value of $f_{\xi 0^*}$ with parameters $(\nu_{10}, \ldots, \nu_{L0})$.
Then for all ``true" $f_{\xi0}$, the parameters $\bm{\nu}_0$ are of the form 
$\{\bm{\nu} = (c\nu_{10}, \ldots, c\nu_{L0})\}$, where $c$ is an arbitrary positive constant. 
\end{proposition}
\begin{proof}
Recall that the recast NTM holds on the set $C\{(H, \bbeta, \xi)\} = \{c_1H_0^{c_2}, c_2\bbeta_0, c_1 \xi_0^{c_2}\}$. 
Without loss of generality, we fix $c_1 = 1$. 
Consider the Weibull mixture model
$$
f_{\xi_0}(t) = \sum_{l=1}^L p_l f_w(t| \psi_{l0}, \nu_{l0}).
$$
It is trivial to show that 
$
f_{\xi_0^{c_2}}(t) = \sum_{l=1}^L p_l f_w(t| \psi_{l0}^{c_2}, \nu_{l0}/c_2).
$

\end{proof}
Based on Proposition \ref{prop: nu}, one can conclude that 
\begin{itemize}
\item For all ``true" $f_{\xi0}$, for $l=2, \ldots, L$, the ratio $\nu_{l0}/\nu_{10}$ is fixed, denoted as $r_{l}$. 
Specifically, we note $r_{1} = 1$. 
\item There exists a ``true" $f_{\xi0}$ such that $\nu_{10} = 1$. 
Conditional on $\bm{\nu}_0 = (1, \ldots, \nu_{L0})$, the constants $c_{lj}$ in Theorem 1 are uniquely specified. 
\end{itemize}

\noindent{\textbf{Step 3}}\\
Again by the total variance formula we have 
$$
\mathbb{V}(H(s_j)|\mathcal{D}) = \mathbb{E}_{\bm{\nu}_0}\{\mathbb{V}(H(s_j)|\mathcal{D}, \bm{\nu}_0)\} +  \mathbb{V}_{\bm{\nu}_0}\{\mathbb{E}(H(s_j)|\mathcal{D}, \bm{\nu}_0)\}. 
$$
By Theorem 1 and \eqref{density: H and psi}, we have 
\begin{align*}
\mathbb{V}(H(s_j)|\mathcal{D}, \bm{\nu}_0)  &= \mathbb{V}(H_0(s_j)|\bm{\nu}) + O(n^{-1}) = \left(\zeta \sum_{l=1}^L c_{lj}^{1/\nu_{l0}} + \frac{\eta}{j+ \sum_{j'=1}^r w_{j'}}\right)^{-2} + O(n^{-1}),\\
\mathbb{E}(H(s_j)|\mathcal{D}, \bm{\nu}_0) &= H_0|\bm{\nu}_0 = \mathbb{E}(\mathcal{Y})=\left(\zeta \sum_{l=1}^L c_{lj}^{1/\nu_{l0}} + \frac{\eta}{j+ \sum_{j'=1}^r w_{j'}} \right)^{-1}. 
\end{align*}
Note that for ``true" $f_{\xi0}$, $\bm{\nu}_0$ are constrained on the subset of $\mathbb{R}^L$
$$
\{(\nu_{10}, \ldots, \nu_{L0})\} = c(1, r_2=, \ldots, r_L) \equiv(\nu_{10}, r_2\nu_{10}, \ldots, r_K \nu_{10}).   
$$
Meanwhile, based on the definition of $d_{lj}$, for any $\nu_{10}$, we have  we have $\psi_{l0}/H_0(s_j) = c_{lj}^{1/(r_l \nu_{10})}$. 
Consequently, we have the following unified expression 
\begin{align*}
\mathbb{V}(H(s_j)|\mathcal{D}, \nu_{10})  &= \left(\zeta \sum_{l=1}^L c_{lj}^{
	\frac{1}{r_l \nu _{10}}} + \frac{\eta}{j+ \sum_{j'=1}^r w_{j'}} \right)^{-2} + O(n^{-1}), \\
E(H(s_j)|\mathcal{D}, \nu_{10}) &= \left(\zeta \sum_{l=1}^L c_{lj}^{\frac{1}{r_l \nu_{10}}} + \frac{\eta}{j+ \sum_{j'=1}^r w_{j'}} \right)^{-1}. 
\end{align*}
Taking expectation and variance on the above two terms respectively completes the proof. 

\end{proof}

\subsection{Proof of Theorem 3}
\begin{proof}
Let $\bm{\theta} = (\bm{\alpha}, \bbeta, \bm{p}, \bm{\psi}, \bm{\nu})$ and $\pi(\bm{\theta})$ be their priors. 
To show the posterior $\pi(\theta|\mathcal{D})$ is proper is equivalent to show that $\int_{\bm{\Theta}} \pi(\bm{\theta}|\mathcal{D})d\bm{\theta} < \infty$, where $\bm{\Theta}$ is the domain of $\bm{\theta}$.

Let $B_j$ be the I-splines functions, for $j=1, \ldots, K$. 
Let $f_w\{\cdot; \psi_{l}, \nu_{l}\}$ be the Weibull PDFs with parameters $\psi_{l}$ and $\nu_{l}$, for $l = 1, \ldots, L$. 
Based on our BNP elicitation, we have 
\begin{align*}
\mathcal{L}(\bm{\theta})&= \prod_{i=1}^{n} f_\xi\{H(\widetilde{y}_i) \exp(-{\bbeta}^T \z_i)\}H'(\widetilde{y}_i)\exp(-{\bbeta}^T \z_i)\\
&= \prod_{i=1}^{n}  \sum_{j=1}^{K}\{\alpha_j B_j'(\widetilde{y}_i) \exp(-{\bbeta}^T \z_i)\sum_{l=1}^{L} p_{l} f_w\{\exp(-{\bbeta}^T \z_i)\}\sum_{j=1}^{K}\alpha_j B_j(\widetilde{y}_i); \psi_{l}, \nu_{l}\}. 
\end{align*}
In the right-censored case, let $n_1$ be the number of uncensored observations. 
We have 
\begin{align*}
\mathcal{L}(\bm{\theta}) < {\mathcal{L}}^*(\bm{\theta}) &\equiv \prod_{i=1}^{n_1} f_\xi\{H(\widetilde{y}_i) \exp(-{\bbeta}^T \z_i)\}H'(\widetilde{y}_i)\exp(-{\bbeta}^T \z_i)\\
&= \prod_{i=1}^{n_1}  \sum_{j=1}^{K}\alpha_j B_j'(\widetilde{y}_i) \exp(-{\bbeta}^T \z_i)\sum_{l=1}^{L} p_{l} f_w\{\exp(-{\bbeta}^T \z_i)\sum_{j=1}^{K}\alpha_j B_j(\widetilde{y}_i); \psi_{l}, \nu_{l}\} 
\end{align*}
Thus, the proof in the complete data and right-censored data scenarios is just the same as the complete data, except for replacing the data size $n$ to $n_1$.

By condition (B1), we first integrate out all $p_l$ and it remains to show that 

\begin{align*}
\mathcal{A}_{l}=&\int_{{\bm{\Theta}_{-p_{l}}}}\prod_{i=1}^{n}  \left[\exp(-{\bbeta}^T \z_i)f_w\left\{\exp(-{\bbeta}^T \z_i)\sum_{j=1}^{K}\alpha_j B_j(\widetilde{y}_i); \psi_{l}, \nu_{l}\right\}\sum_{j=1}^{K}\alpha_j B_j'(\widetilde{y}_i)\right]\\
&\times p(\bm{\theta}_{-p_l})d\bm{\theta}_{-p_l} < \infty, 
\end{align*}
for all $l$, where $\bm{\theta}_{-p_{l}}$ denotes all parameters except $p_{l}$ and ${\bm{\Theta}_{-p_{l}}}$ denotes corresponding domains. 

Let $\bm{\phi}_i = (B_1'(\widetilde{y}_i), \cdots, B_K'(\widetilde{y}_i))^T$ and $\bm{\Phi}_i = (B_1(\widetilde{y}_i), \cdots, B_K(\widetilde{y}_i))^T$.
Let 
\begin{align*}
M_1^* &= \max\{\max(\bm{\phi}_1), \ldots, \max(\bm{\phi}_{n_1})\}\\
M_2^* & = \min\{\min(\bm{\Phi}_1^+), \ldots, \min(\bm{\Phi}_n^+)\}, 
\end{align*}
where $\bm{\Phi}_i^+$ denotes the set of positive entries in $\bm{\Phi}_i$. 
By the definition of I-spline functions, if $B_j(s) = 0$, $B_j'(s) = 0$ too. 
Hence, we have, for any $\bm{\alpha}$, 
$$
0< \frac{\bm{\alpha^T} \bm{\phi}_i}{\bm{\alpha^T} \bm{\Phi}_i} \le  \frac{M_1^*}{M_2^*} \equiv M_0. 
$$
Thus, we have 
\begin{align}
\label{inequality: M0}
\mathcal{A}_l \le M_0^{n} \int_{\bm{\Theta}_{-p_l} }\prod_{i=1}^{n}  \left[\exp(-{\bbeta}^T \z_i) \bm{\alpha}^T \bm{\Phi}_i f_w\left\{\exp(-{\bbeta}^T \z_i)\bm{\alpha}^T \bm{\Phi}_i; \psi_{l}, \nu_{l}\right\}\right] p(\bm{\theta}_{-p_l})d\bm{\theta}_{-p_l}. 
\end{align}

By condition (B4), we can find $p$ observations such that the corresponding $p \times p$ submatrix of covariates, with each row being the vector of covariates of one observation, is full rank. 
Let $\bm{z}^*$ denote that full rank $p$ matrix and let $\bm{\gamma} = -\bm{z}^*\bbeta = (\gamma_1, \cdots, \gamma_p)^T$. 
Note that by condition (B3), for all $x\in \mathbb{R}$, we can find a constant $M_1< \infty$ such that
\begin{align*}
\exp(x)f_w\{\exp(x)\bm{\alpha}^T\bm{\phi}_i; \psi_{l}, \nu_{l}\}\bm{\alpha}^T\bm{\Phi}_i
\le M_0 \{\exp(x)\bm{\alpha}^T\bm{\phi}_i\} f_w\{\exp(x)\bm{\alpha}^T\bm{\phi}_i; \psi_{l}, \nu_{l}\} < M_1. 
\end{align*}
Therefore, from inequality \eqref{inequality: M0}, we further have 
\begin{align*}
\mathcal{A}_l \le M_0^{n} M_1^{n - p}\int_{\bm{\Theta}_{-p_l}}  \prod_{h=1}^p  \left[\exp(\gamma_h) \bm{\alpha}^T \bm{\Phi}_i f_w\left\{\exp(\gamma_h)\bm{\alpha}^T \bm{\Phi}_i; \psi_{l}, \nu_{l}\right\}\right] p(\bm{\theta}_{-p_l})d\bm{\theta}_{-p_l}.
\end{align*}

Since $\bm{z}^*$ is a one-on-one linear operation of $\bbeta$, the integrand $\bbeta$ can be transferred to $\bm{\gamma} = (\gamma_1, \ldots, \gamma_p)^T$. 
Therefore, there exists a finite constant $ M_2$ such that 
\begin{align*}
\mathcal{A}_l &\le M_0^{n} M_1^{n - p} M_2 \int_{\bm{\Theta}_{-\{p_l, \bm{\beta}\}}}p(\bm{\theta}_{\{-p_l, \bm{\beta}\}})d\bm{\theta}_{\{-p_l, \bm{\beta}\}}\int_{\mathbb{R}^p} \prod_{h=1}^p  \left[\exp(\gamma_h) \bm{\alpha}^T \bm{\Phi}_i f_w\left\{\exp(\gamma_h)\bm{\alpha}^T \bm{\Phi}_i; \psi_{l}, \nu_{l}\right\}\right] \\
&\quad \quad \quad \times d\gamma_1\ldots d\gamma_p \\   
& \le M_0^{n} M_1^{n - p} M_2 \int_{\bm{\Theta}_{-\{p_l, \bm{\beta}\}}}p(\bm{\theta}_{\{-p_l, \bm{\beta}\}})d\bm{\theta}_{\{-p_l, \bm{\beta}\}} \prod_{h=1}^p \int_{-\infty}^{+\infty}\left[\exp(\gamma_h) \bm{\alpha}^T \bm{\Phi}_i f_w\left\{\exp(\gamma_h)\bm{\alpha}^T \bm{\Phi}_i; \psi_{l}, \nu_{l}\right\}\right]\\
& \quad \quad \quad \times d\gamma_1\ldots d\gamma_p \\   
& = M_0^{n} M_1^{n - p} M_2 \int_{\bm{\Theta}_{-\{p_l, \bm{\beta}\}}}p(\bm{\theta}_{\{-p_l, \bm{\beta}\}})d\bm{\theta}_{\{-p_l, \bm{\beta}\}} \prod_{h=1}^p \int_{0}^{+\infty} f_w(u_h; \psi_l, \nu_l)du_1\ldots du_p < \infty. 
\end{align*}
\end{proof}

\section{The DPM model for $S_\xi$}
\label{PriorLambda}
A regular Dirichlet process mixture (DPM) model \citep{lo1984class} is assigned for $S_\xi$, the survival probability function of the positive random variable $\xi$. 
The DPM is a kernel convolution to the Dirichlet process (DP).
We use the stick breaking representation for $G \sim \text{DP}(c, G_0)$ \citep{sethuraman1994constructive}
$$
G(\cdot) = \sum_{l=1}^{\infty} p_{l} \delta_{\theta_{l}}(\cdot), \theta_{l} \sim G_0, p_{l} \sim \text{SB}(1, c)
$$
where $\delta(\cdot)$ is the point mass function, and $\text{SB}$ is the stick-breaking representation. 
We call $G_0$ as the base measure and $c$ as the total mass parameter, acting as the center and precision of the DP, respectively. 

Following the above stick-breaking representation, we construct the truncated DPM priors for $S_\xi$ and $f_\xi$ with the Weibull kernel such that
\begin{align*}
S_\xi = 1- \sum_{l=1}^{L}p_{l} F_w(\psi_{l}, \nu_{l}), f_\xi = \sum_{l=1}^{L}p_{l} f_w(\psi_{l}, \nu_{l}),
p_{l} \sim \text{SB}(1, c), (\psi_{l}, \nu_{l}) \sim G_0, 
\end{align*}
where $L$ is the truncation number, and $F_w$ and $f_w$ denote CDF and density of Weibull distribution, respectively. 
We fix the truncation number $L$ rather than sampling it to simplify computation as a common strategy \citep{rodriguez2008nested}.  
Let $S_\xi^{(\infty)}$ denote the limit of the DPM model, and $S_\xi ^{(L)}$ denote the truncated form.
The truncation number $L$ is generally selected such that the $L_1$ error between the limit form and the truncated form, denoted as $\int_{0}^{+\infty} |S_\xi^{(\infty)}(s) - S_\xi ^{(L)}(s)|ds$, is as small as possible. 
As shown by \cite{ishwaran2002approximate}, this $L_1$ error is bounded by $4n\exp\{-(L-1)/c\}$, where $n$ denotes the sample size. 
In practice, an error bound of $0.01$ is considered to be sufficiently small \citep{ohlssen2007flexible}. 
Since we fix the total mass parameter $c=1$ as a common practice \citep{gelman2013bayesian}, for sample size $n<600$, $L=12$ is a suitable choice of truncation number. 
In our numerical studies, we find that an $L$ in the range of $10-15$ is appropriate to approximate the DPM model well. 
Users of \texttt{BuLTM} are free to adjust the truncation number according to the data size.

\section{Additional simulation results}
\label{sec:addsim}
We report additional simulations here. 
We first introduce the reproducibility of all simulations, and report the results of simulations in highly-censored cases, results of parametric estimation under AFT models, results of effective sample size (ESS) given by \texttt{BuLTM} in simulations, and results of prediction and estimation on data sets with size $100$. 
\subsection{Reproducibility}
\label{subsec:reprod}
This subsection is about details for the reproducibility of our simulation results. 
In all simulations, we run four independent parallel chains in \texttt{BuLTM} as the default setting in \texttt{Stan}. 
The length of each chain is $1500$ with the first $500$ iterations burn-in and we aggregate four chains to obtain a total of $4000$ posterior samples without any thinning. 
In all numerical studies, we set $L=12$ as the truncation number of DPM, $c=1$ for the total mass parameter, and $r=4$ for the order of smoothness of I-spline functions as the default of \texttt{splines2}. 
We configure 4 interior knot series for all simulations settings. 
That is, selecting interior knots from $(20\%, 40\%, 60\%, 80\%)$ percentiles of the observed or uncensored data. 
Throughout all numerical studies, we use the hyperparameter configuration of $(\eta, \zeta, \rho) = (0.01, 0.5, 1)$; refer to Section \ref{subsec: Increasing prior informativeness} for justifications.

The credible interval for $\bm{\beta}$ estimation given by \texttt{BuLTM} is the default central posterior interval in \texttt{Stan}. 
All numerical studies are realized in \texttt{R} version 4.3.0 with \texttt{rstan} version 2.26.4.  
The pointwise bias of \texttt{BuLTM} for parametric estimation should be computed in a different way from usual. 
Among all simulations, we re-scale the mean vector of estimated $\hat{\bbeta}^*$ into a unit vector and then compute the pointwise bias.
Otherwise, the result is surely biased no matter what kind of unit-norm estimator is used. 
The reason is that \texttt{BuLTM} provides an estimate of a unit vector in each replication of simulations, while the element-wise mean of a series of unit vectors is not a unit vector anymore since for unit vectors $v_1, \ldots, v_n \in \text{St}(1, p)$, $||n^{-1}\sum_{i=1}^n v_i||\le 1$ by triangle inequality. 
All the code and data to reproduce the results in the article are collected on GitHub \href{https://github.com/LazyLaker/BuLTM}{https://github.com/LazyLaker/BuLTM}.

\subsection{Detailed simulation settings}
\noindent{\textbf{Setting} (a)}. In this setting, the covariate vector is generated as $\bm{z} \sim \text{MVN}_3(\bm{0}_3, \Sigma_3)$, where $\Sigma_3 \equiv \sigma_{ij} = 0.75^{|i-j|}$ for $i, j = 1, \ldots, 3$.  
We set the true regression coefficients $\bm{\beta} = (\sqrt{3}/3, \sqrt{3}/3, \sqrt{3}/3)$ as a unit-norm vector. 
The transformation $h$ is set as the inverse (signed) Box-Cox function with $\lambda = 0.5$ (the same as \cite{kowal2024monte}). 
In setting (a.1), the model error is generated by $\epsilon \sim N(0, 1)$; 
in setting (a.2), the model error is generated by $\epsilon \sim 0.5N(-0.5, 0.5^2) + 0.5N(1.5, 1^2)$, yielding a bi-modal, asymmetric, and non-central distribution. 

We add two additional simulation settings (a.3) and (a.4). 
Setting (a.3) is the proportional hazard case where $\epsilon$ follows a Gumbel distribution such that 
$F_\epsilon(s) = \exp\{\exp(-s)\}$;
Setting (a.4) is the proportional odds case where $\epsilon$ follows a standard logistic distribution. \\

\noindent{\textbf{Setting } (b)}. In this setting, the first covariate is generated by $z_1 \sim \text{Bernoulli}(0.5)$ independently to the remaining, and the remaining are generated by $(z_2, z_3) \sim \text{MVN}_2(\bm{0}_2, \Sigma_2)$, where $\Sigma_2 \equiv \sigma_{ij} = 0.75^{|i-j|}$ for $i, j = 1, 2$. 
The transformation $h$ is set as 
$$
h(x) = \log\left[(0.8x + \sqrt{x} + 0.825) \{0.5\Phi_{1, 0.3}(x) + 0.5\Phi_{3, 0.3}(x) +C_0 \}\right], 
$$
where $\Phi_{\mu, \sigma}$ denotes the CDF of $N(\mu, \sigma^2)$, and $C_0$ is the constant such that $\exp\{h(0)\} = 0$. 
In Setting (b.1), the model error distribution is set as the standard extreme value distribution, yielding a Cox's proportional hazard model; 
in setting (b.2), the model error is also generated by $\epsilon \sim 0.5N(-0.5, 0.5^2) + 0.5N(1.5, 1^2)$. 
In Setting (b.1), the independent censoring variable is generated by $C = \min \{\text{Exp(1)}, 1.5\}$; 
in setting (b.2), the the independent censoring variable is generated by $C \sim U(1, 3.5)$. \\

\noindent{\textbf{Setting (c)}}. 
In this setting, the covariate vector $\bm{z} = (z_1, z_2)$, where $z_1, z_2 \sim U(-2, 2)$.
The data generation model is given by 
\begin{align*}
h(y) = f_1(z_1) + f_2(z_2) + \epsilon, 
\end{align*}
where $h$ is set as the inverse
(signed) Box-Cox function similar to setting (a), $f_1(x) =-x + \pi \sin(\pi x) $, $f_2(x)= 0.5x + 15\phi (2(x - 0.2)) - \phi(x + 0.4)$, with $\phi$ being the density function of $N(0, 1)$. 
In Setting (c.1), the model error $\epsilon \sim N(0,1)$; in Setting (c.2), the model error $\epsilon \sim 0.5N(-0.5, 0.5^2) + 0.5N(1.5, 1^2)$. 
We employ nonlinear additive covariate transformation for both \texttt{BuLTM} and \texttt{SeBR} such that $f(\bm{z}) = \sum_{j=1}^p f_j(\bm{z}_j)$, where $f_j(s) = \sum_{k=1}^K \beta_{jk} \phi_{jk}(s)$. 
In this setting, we set $\phi_{jk}$ as the B-spline function with 5 interior knots. 

\subsection{More visualizations of hyperparameter tuning }
\label{subsec: Hyperparameter tuning in other data settings}
In this subsection, we examine the proposed sufficient informativeness criterion in other numerical studies. 
As we mentioned in Section 5.1, we will use hyperparameter configuration $(\eta, \zeta, \rho) = (0.01, 0.25, 1)$ as the initial value. 

\noindent{\textbf{Setting (a.1)}}~ 
Figure \ref{zeta_plot_A1} shows that that the within-chain variance exceeds the inverse of prior information level $\mathcal{V}_{s_{j_0}}$. 
Figure \ref{lp_plot_A1} shows that the chains of \texttt{lp\_\_} are well mixed. 
The well mixing is evidenced by $\hat{R} = 1.006$, with an effective sample size (ESS) of $490$. 
Therefore, we conclude that hyperparameter configuration $(\eta, \zeta, \rho) = (0.01, 0.25, 1)$ leads to reliable prediction in setting (a.1). 

\noindent{\textbf{Setting (b.1)}} ~
Figure \ref{zeta_plot_B1} shows that that the within-chain variance exceeds the inverse of prior information level $\mathcal{V}_{s_{j_0}}$. 
Figure \ref{lp_plot_B1} shows that the chains of \texttt{lp\_\_} are well mixed. 
The well mixing is evidenced by $\hat{R} = 1.006$, with an ESS of $503$. 
Therefore, we conclude that hyperparameter configuration $(\eta, \zeta, \rho) = (0.01, 0.25, 1)$ leads to reliable prediction in setting (b.1).

\begin{figure}[!htb]
\centering
\subfigure[]{
	\begin{minipage}[t]{0.45\linewidth}
	\centering
	\includegraphics[width=2.0in]{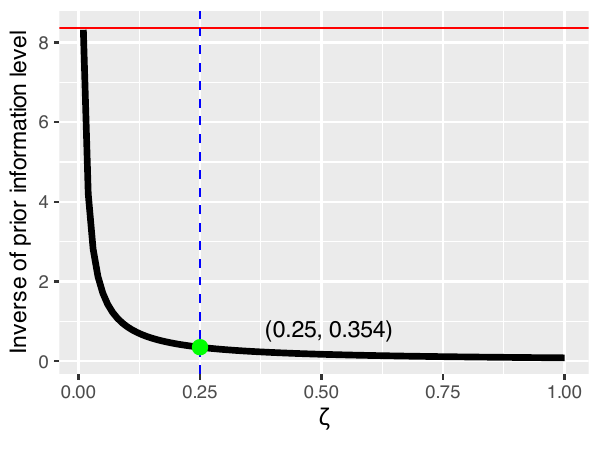}
	\label{zeta_plot_A1}
	\end{minipage}
}
\vspace{-.5cm}
\subfigure[]{
	\begin{minipage}[t]{0.45\linewidth}
	\centering
	\includegraphics[width=2.0in]{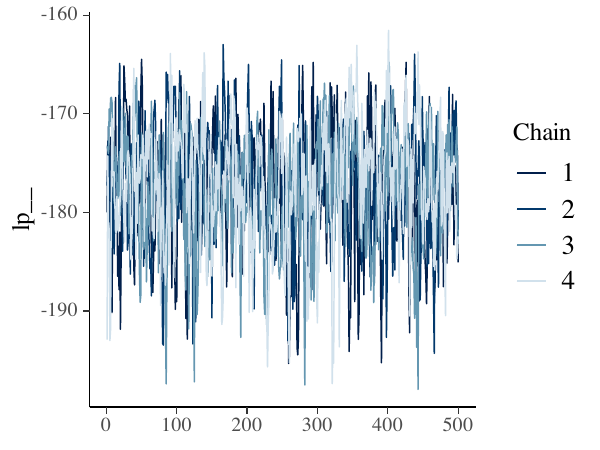}
	\label{lp_plot_A1}
	\end{minipage}
}

\caption{\footnotesize MCMC checking results in setting (a.1) with hyperparameter configuration of $(\eta, \zeta, \rho) = (0.01, 0.25, 1)$. 
	(a), the curve of $\widetilde{\mathcal{V}}_{s_{j_0}}(\eta, \zeta)$ with $\eta = 0.01$ fixed; horizontal line: the within-chain MCMC variance sampled in setting (a.1). (b), trace plot of chains of \texttt{lp\_\_}. }
\label{fig: mixing A1}
\end{figure}

\begin{figure}[!htb]
\centering
\subfigure[]{
	\begin{minipage}[t]{0.45\linewidth}
	\centering
	\includegraphics[width=2.0in]{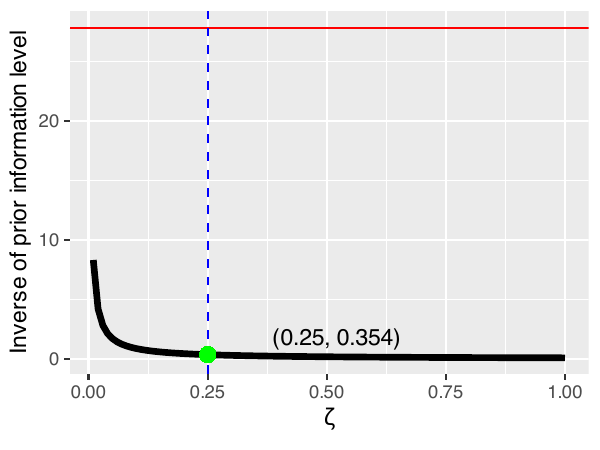}
	\label{zeta_plot_B1}
	\end{minipage}
}
\vspace{-.5cm}
\subfigure[]{
	\begin{minipage}[t]{0.45\linewidth}
	\centering
	\includegraphics[width=2.0in]{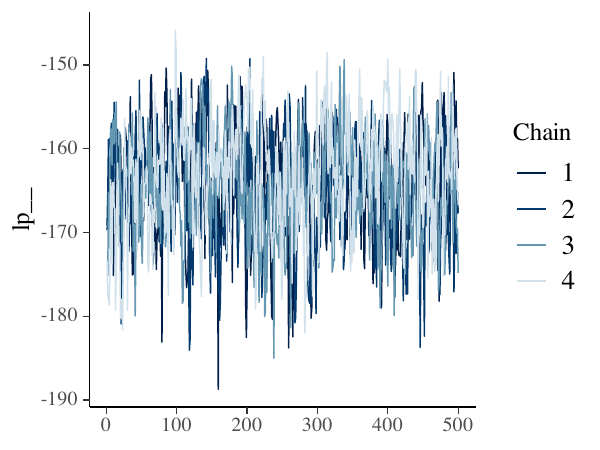}
	\label{lp_plot_B1}
	\end{minipage}
}

\caption{\footnotesize MCMC checking results in setting (b.1) with hyperparameter configuration of $(\eta, \zeta, \rho) = (0.01, 0.25, 1)$. 
	(a), the curve of $\widetilde{\mathcal{V}}_{s_{j_0}}(\eta, \zeta)$ with $\eta = 0.01$ fixed; horizontal line: the within-chain MCMC variance sampled in setting (b.1). (b), trace plot of chains of \texttt{lp\_\_}. }
\label{fig: mixing B1}
\end{figure}

\noindent{\textbf{Setting (b.2)}} ~
Figure \ref{zeta_plot_B2} shows that that the within-chain variance exceeds the inverse of prior information level $\mathcal{V}_{s_{j_0}}$. 
Figure \ref{lp_plot_B2} shows that the chains of \texttt{lp\_\_} are well mixed. 
The well mixing is evidenced by $\hat{R} = 1.010$, with an ESS of $509$. 
Therefore, we conclude that hyperparameter configuration $(\eta, \zeta, \rho) = (0.01, 0.25, 1)$ leads to reliable prediction in setting (b.2).

\begin{figure}[!htb]
\centering
\subfigure[]{
	\begin{minipage}[t]{0.45\linewidth}
	\centering
	\includegraphics[width=2.0in]{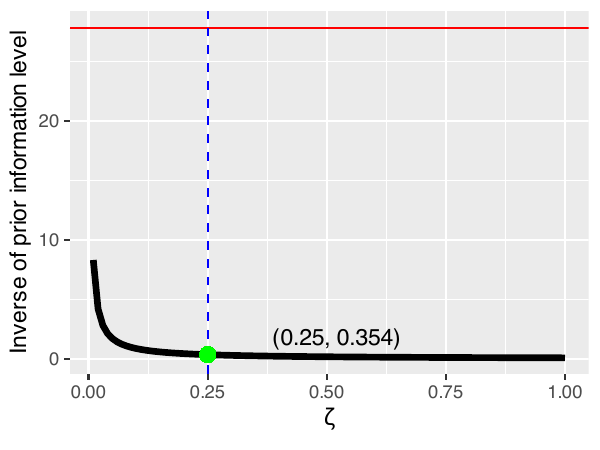}
	\label{zeta_plot_B2}
	\end{minipage}
}
\vspace{-.5cm}
\subfigure[]{
	\begin{minipage}[t]{0.45\linewidth}
	\centering
	\includegraphics[width=2.0in]{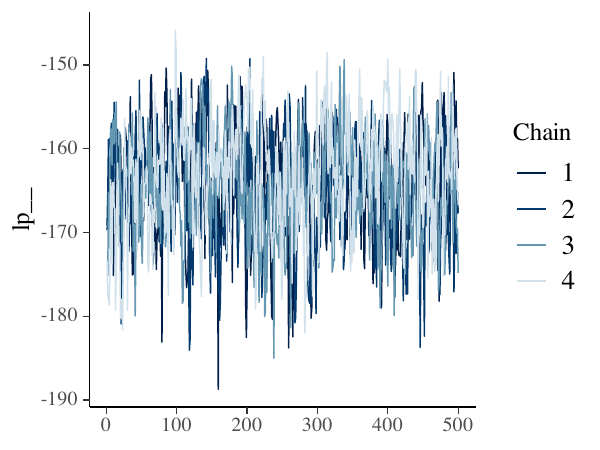}
	\label{lp_plot_B2}
	\end{minipage}
}

\caption{\footnotesize MCMC checking results in setting (b.2) with hyperparameter configuration of $(\eta, \zeta, \rho) = (0.01, 0.25, 1)$. 
	(a), the curve of $\widetilde{\mathcal{V}}_{s_{j_0}}(\eta, \zeta)$ with $\eta = 0.01$ fixed; horizontal line: the within-chain MCMC variance sampled in setting (b.2). (b), trace plot of chains of \texttt{lp\_\_}. }
\label{fig: mixing B2}
\end{figure}

\subsection{Results in other simulation settings}

\noindent{\textbf{Setting (a.3)}}. 
Results under Setting (a.3) are presented in Figure \ref{fig: CTM_Logistic}. 
We find that \texttt{tram} outperforms in RIMSE since it correctly specifies the reference distribution; 
\texttt{SeBR} ranks second with the normal reference distribution, since the shape of the true logistic reference is close to that of normal reference; 
\texttt{BuLTM} ranks third, since the DPM model with Weibull kernel does not approximate the logistic reference well in the tail. 
On the other hand, the proposed \texttt{BuLTM} significantly outperforms all competitors in MAE (one-sided paired t-test $p$-values: $2.57 \times 10^{-8}$ against \texttt{SeBR}, $2.20 \times 10^{-16}$ against \texttt{PTM}). 
These results demonstrate the superiority of \texttt{BuLTM} in estimating the predicted values, or equivalently, the middle of predictive distributions.

\begin{figure}[!htb]
\centering
\subfigure[]{
	\begin{minipage}[t]{0.45\linewidth}
	\centering
	\includegraphics[width=2.0in]{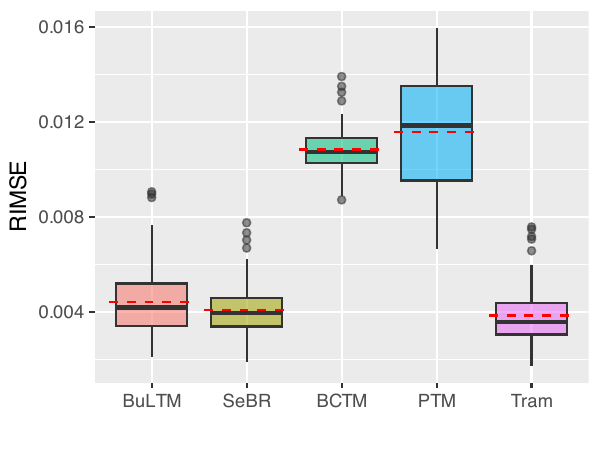}
	\label{RIMSE_logistic_CTM}
	\end{minipage}
}
\vspace{-.5cm}
\subfigure[]{
	\begin{minipage}[t]{0.45\linewidth}
	\centering
	\includegraphics[width=2.0in]{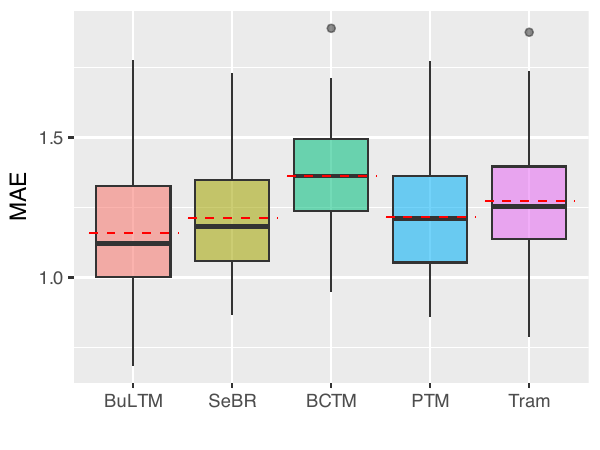}
	\label{MAE_Logistic_CTM}
	\end{minipage}
}

\caption{\footnotesize Box-plots of predictive assessments under Setting (a.3). (a), RIMSE; (b), MAE. }
\label{fig: CTM_Logistic}
\end{figure}

\noindent{\textbf{Setting (a.4)}}. 
Results under Setting (a.4) are presented in Figure \ref{fig: CTM_PH}. 
We find that \texttt{BuLTM} and \texttt{tram} outperform the remaining competitors in the RIMSE, and \texttt{BuLTM} outperforms other competitors in MAE. 
The reason is that the exponential of the extreme value distribution is exactly a Weibull distribution. 
Thus, the DPM model with Weibull kernel characterizes the model error pretty well. 

\begin{figure}[!htb]
\centering
\subfigure[]{
	\begin{minipage}[t]{0.45\linewidth}
	\centering
	\includegraphics[width=2.0in]{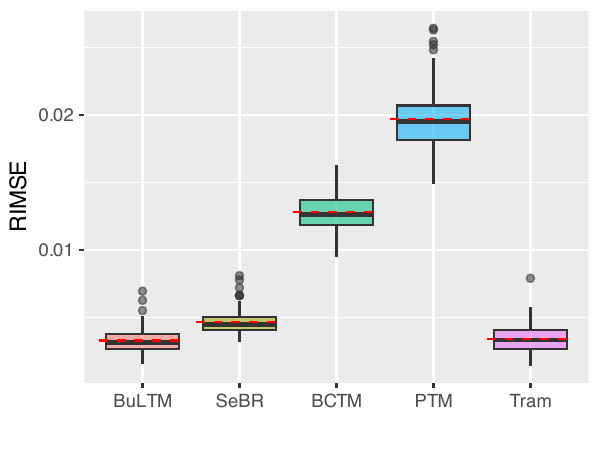}
	\label{RIMSE_PH_CTM}
	\end{minipage}
}
\vspace{-.5cm}
\subfigure[]{
	\begin{minipage}[t]{0.45\linewidth}
	\centering
	\includegraphics[width=2.0in]{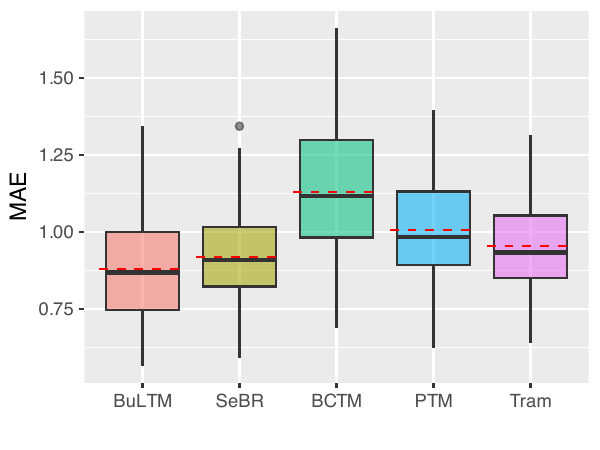}
	\label{MAE_PH_CTM}
	\end{minipage}
}

\caption{\footnotesize Box-plots of predictive assessments under Setting (a.4). (a), RIMSE; (b), MAE. }
\label{fig: CTM_PH}
\end{figure}

\noindent{\textbf{Setting (c.1)}}. 
Results under Setting (c.1) are presented in Figure \ref{fig: CTM_nonlinGauss}. 
We find that \texttt{BuLTM} and \texttt{SeBR} are comparable and outperform the remaining competitors in both RIMSE and MAE. 
It demonstrate the predictive capability of \texttt{BuLTM} in nonlinear cases. 

\begin{figure}[!htb]
\centering
\subfigure[]{
	\begin{minipage}[t]{0.45\linewidth}
	\centering
	\includegraphics[width=2.0in]{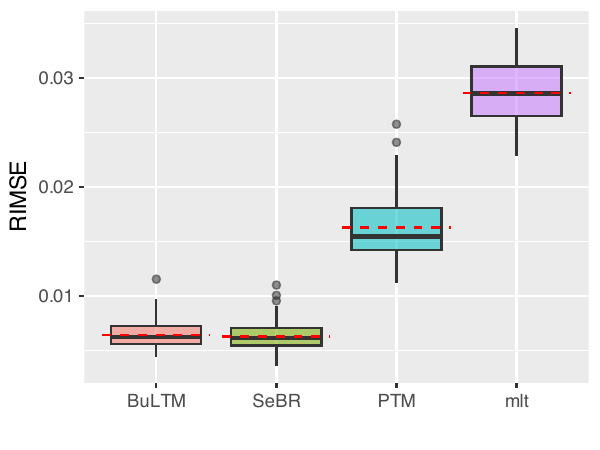}
	\label{RIMSE_nonlinGauss}
	\end{minipage}
}
\vspace{-.5cm}
\subfigure[]{
	\begin{minipage}[t]{0.45\linewidth}
	\centering
	\includegraphics[width=2.0in]{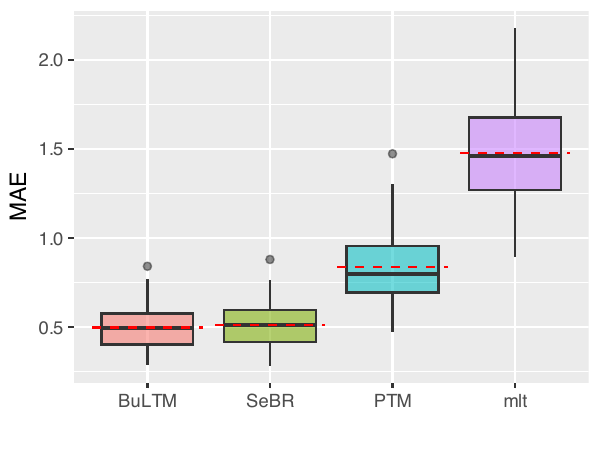}
	\label{MAE_nonlinGauss}
	\end{minipage}
}

\caption{\footnotesize Box-plots of predictive assessments under Setting (c.1). (a), RIMSE; (b), MAE. }
\label{fig: CTM_nonlinGauss}
\end{figure}

\noindent{\textbf{Setting (c.2)}}. 
Results under Setting (c.2) are presented in Figure \ref{fig: CTM_nonlinMix}. 
We find that \texttt{SeBR} performs the best in both RIMSE and MAE. 
\texttt{BuLTM} is slightly worse than \texttt{SeBR} but outperforms the remaining competitors. 
We conjecture that the DPM sampling may be more difficult with higher dimensionality of $\bm{\beta}$. 
We conjecture that implementing horseshoe priors in \texttt{Stan} \citep{piironen2017sparsity} may improve the performances, which is one of our future works. 

\begin{figure}[!htb]
\centering
\subfigure[]{
	\begin{minipage}[t]{0.45\linewidth}
	\centering
	\includegraphics[width=2.0in]{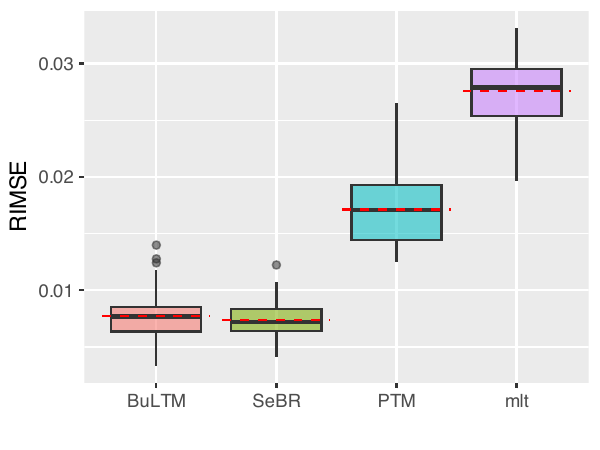}
	\label{RIMSE_nonlinMix}
	\end{minipage}
}
\vspace{-.5cm}
\subfigure[]{
	\begin{minipage}[t]{0.45\linewidth}
	\centering
	\includegraphics[width=2.0in]{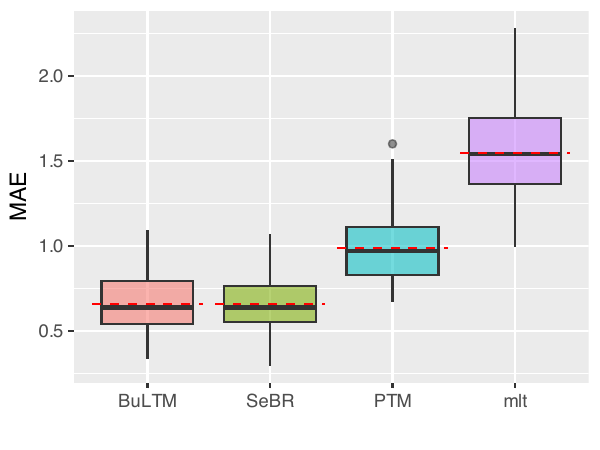}
	\label{MAE_nonlinMix}
	\end{minipage}
}

\caption{\footnotesize Box-plots of predictive assessments under Setting (c.2). (a), RIMSE; (b), MAE. }
\label{fig: CTM_nonlinMix}
\end{figure}

\subsection{Results of parametric estimation}
Parametric estimation results under linear transformation model Settings (a.1), (a.2), (a.3), (a.4), (b.1), and (b.2) are collected in Table \ref{ParaEst}. 
We find that the estimators given by \texttt{BuLTM} have low bias, and the coverage of the posterior interval is close to the nominal level.

\subsection{Mixing of other parameters}
We visualize the trace plot of other parameters to examine their MCMC mixing under Settings (a.1) and (a.2). 
Specifically, we examine the mixing of $\bm{\beta}$, and the first three components of $\bm{\psi}$ and $\bm{\nu}$.

According to Figure \ref{beta_mixing_A1}, in Setting (a.1), the MCMC chains of $\bm{\beta}$ are mixed. 
The $\hat{R}$ statistics for $\beta_1$, $\beta_2$, and $\beta_3$ are $1.013$, $1.011$, and $1.010$ respectively, demonstrating that the MCMC mixing is acceptable. 
The ESS of $\beta_1$, $\beta_2$, and $\beta_3$ are 275, 360, and 385 respectively, indicating that longer chains are needed for reliable estimation of $\bm{\beta}$. 
For comparison, the MCMC chains of $\bm{\beta}$ mix better in Setting (a.2), according to Figure \ref{beta_mixing_A2}. 
The corresponding $\hat{R}$ statistics for $\beta_1$, $\beta_2$, and $\beta_3$ are closer to $1$ (1.000, 1.004, 1.001 respectively), demonstrating the better mixing. 
As a result, the ESS also increases. 
We conjecture the reason is that the Weibull mixture prior correctly captures the mixture nature of the model error, leading to better fitting. 

By contrast, some of the DPM components may NOT achieve MCMC mixing as $\bm{\beta}$. 
Figures \ref{psi_mixing_A1} and \ref{nu_mixing_A1} show that the chains of the first three components of $\bm{\psi}$ and $\bm{\nu}$ are poorly mixed in Setting (a.1), and Figure \ref{psi_mixing_A2} and \ref{nu_mixing_A2} show that these chains are also poorly mixed in Setting (a.2). 
As a result, the corresponding $\hat{R}$ statistics exceed $1.05$, and the corresponding ESS is very low. 
This poor mixing is caused by the label-switching issue caused by the invariance against the permutations of the allocation of $(p_l, \psi_l, \nu_l)$. 
As an illustration, in Setting (a.2), we clearly find that the chains $\psi_1$ and $\psi_2$ are symmetrical to each other.
The reason is that, the ``true" $F_\epsilon$ is a two-component mixture model with equally weights, and thus, its posterior easily tends to be two-component. 
The label-switching issue naturally occurs between the labels first two components. 

\begin{table}[!htb]
\centering
\footnotesize
\tabcolsep 10pt
\caption{\label{ParaEst}\footnotesize{Results of estimation of $\bbeta$ given by \texttt{BuLTM}. 
PSD: average posterior standard deviation; SDE: empirical standard error of the estimators; CP: the coverage probability of 95\% credible interval. }}\vspace{.3cm}
\scriptsize
\begin{tabular}{cc|cccc|cccc}
\hline
\multicolumn{2}{c}{} & \multicolumn{4}{c}{(a.2)} & \multicolumn{4}{c}{(a.1)} \\
\hline
& Parameter &   BIAS &   PSD &   SDE & CP &   BIAS &   PSD &   SDE & CP \\ 
\hline
& $\beta_1$ & -0.019 & 0.087 & 0.091 & 94 & -0.025 & 0.100 & 0.115 & 91 \\ 
& $\beta_2$ & -0.019 & 0.115 & 0.130 & 91 & -0.024 & 0.128 & 0.144 & 92 \\ 
& $\beta_3$ & -0.019 & 0.089 & 0.102 & 90 & -0.024 & 0.100 & 0.109 & 95 \\ 
\hline
\multicolumn{2}{c}{} & \multicolumn{4}{c}{(a.3)} & \multicolumn{4}{c}{(a.4)} \\
\hline
& Parameter &   BIAS &   PSD &   SDE & CP &   BIAS &   PSD &   SDE & CP \\ 
\hline
& $\beta_1$ & -0.057 & 0.158 & 0.177 & 91 & -0.021 & 0.097 & 0.100 & 95 \\ 
& $\beta_2$ & -0.062 & 0.189 & 0.213 & 89 & -0.021 & 0.128 & 0.122 & 97 \\ 
& $\beta_3$ & -0.059 & 0.154 & 0.174 & 93 & -0.021 & 0.098 & 0.100 & 92 \\ 
\hline
\multicolumn{2}{c}{} & \multicolumn{4}{c}{(b.2)} & \multicolumn{4}{c}{(b.1)} \\
\hline
& Parameter &   BIAS &   PSD &   SDE & CP &   BIAS &   PSD &   SDE & CP \\ 
\hline
& $\beta_1$ & -0.015 & 0.096 & 0.110 & 93 & -0.035 & 0.153 & 0.157 & 91 \\ 
& $\beta_2$ & -0.015 & 0.084 & 0.096 & 87 & -0.034 & 0.127 & 0.126 & 95 \\ 
& $\beta_3$ & -0.015 & 0.082 & 0.082 & 95 & -0.034 & 0.122 & 0.133 & 92 \\ 
\hline         
\end{tabular}
\end{table}

\begin{figure}[!htb]
\centering
\subfigure[]{
\begin{minipage}[t]{0.8\linewidth}
\centering
\includegraphics[width=4.0in]{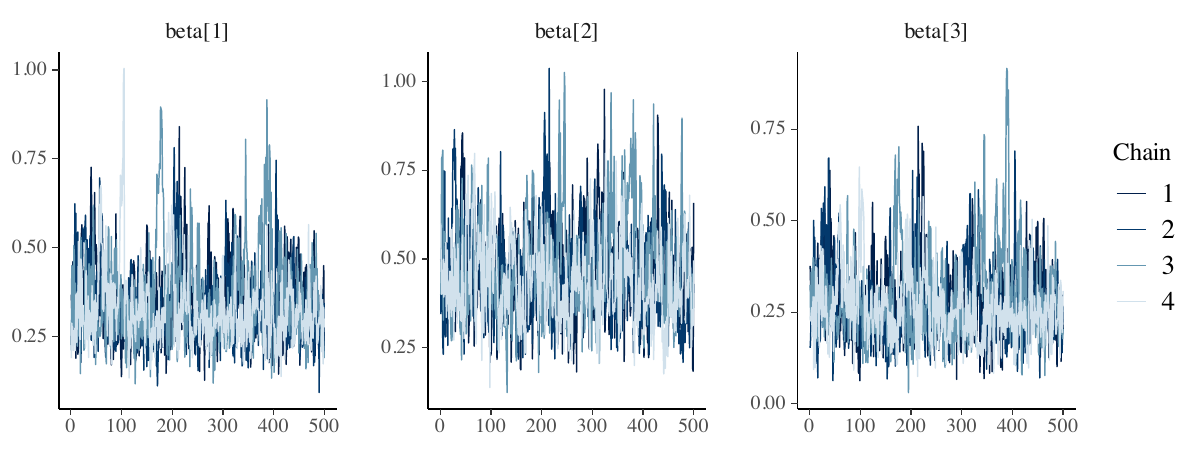}
\label{beta_mixing_A1}
\end{minipage}
}

\subfigure[]{
\begin{minipage}[t]{0.8\linewidth}
\centering
\includegraphics[width=4.0in]{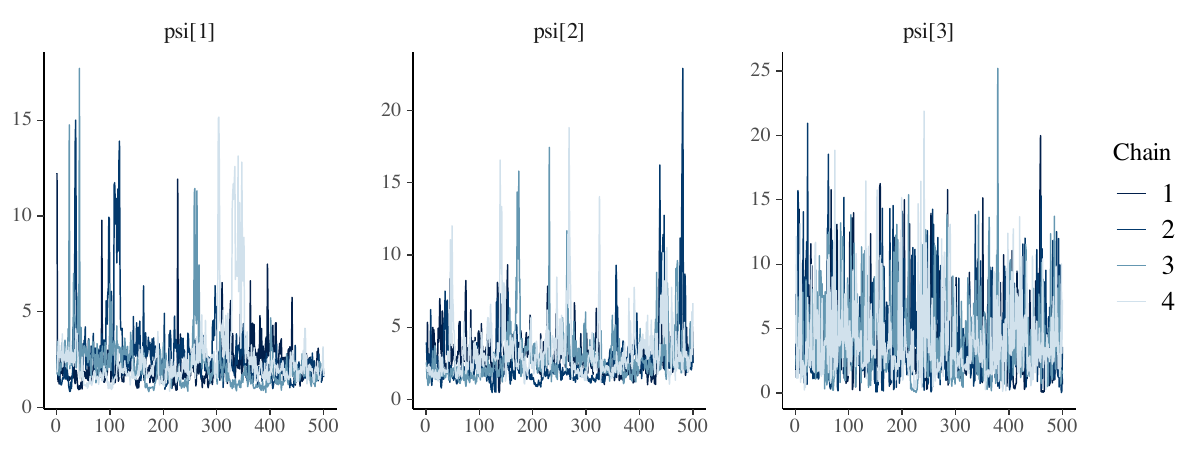}
\label{psi_mixing_A1}
\end{minipage}
}

\subfigure[]{
\begin{minipage}[t]{0.8\linewidth}
\centering
\includegraphics[width=4.0in]{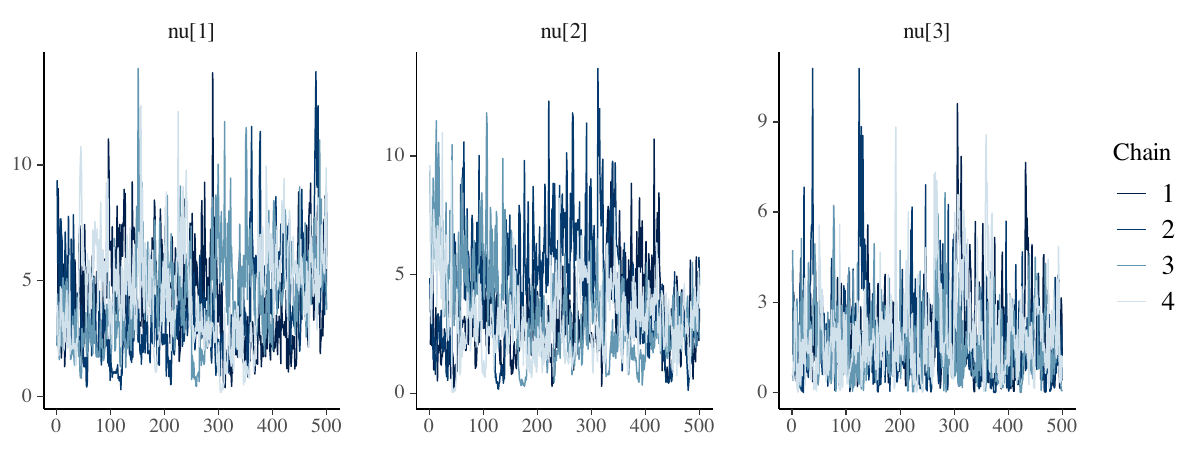}
\label{nu_mixing_A1}
\end{minipage}
}

\caption{\footnotesize Visualization of MCMC traces of other parameters in Setting (a.1). 
(a), trace plots of $\bm{\beta}$; (b), trace plots of the first three components of $\bm{\psi}$; 
(c), trace plots of the first three components of $\bm{\nu}$. }
\label{fig: mixing A1 others}
\end{figure}

\begin{figure}[!htb]
\centering
\subfigure[]{
\begin{minipage}[t]{0.8\linewidth}
\centering
\includegraphics[width=4.0in]{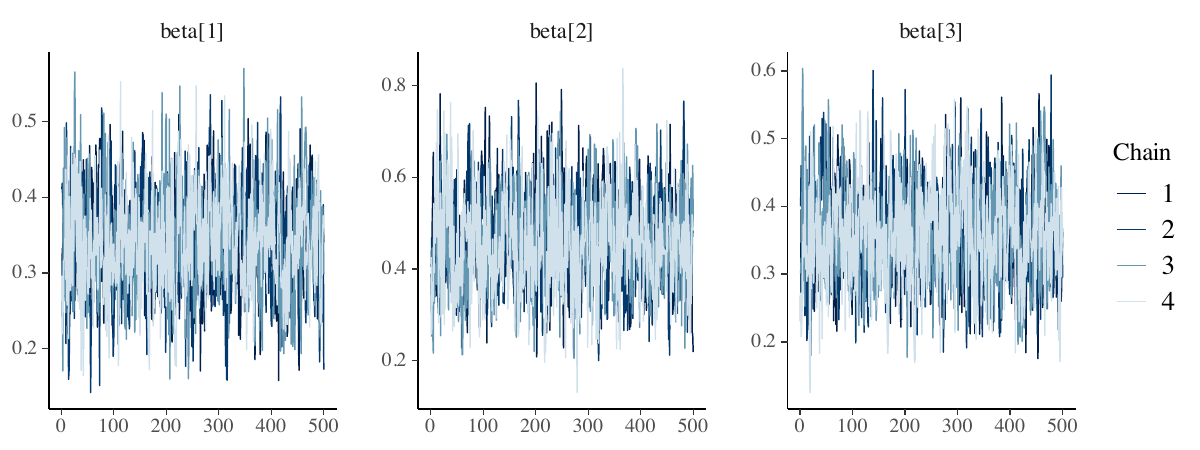}
\label{beta_mixing_A2}
\end{minipage}
}

\subfigure[]{
\begin{minipage}[t]{0.8\linewidth}
\centering
\includegraphics[width=4.0in]{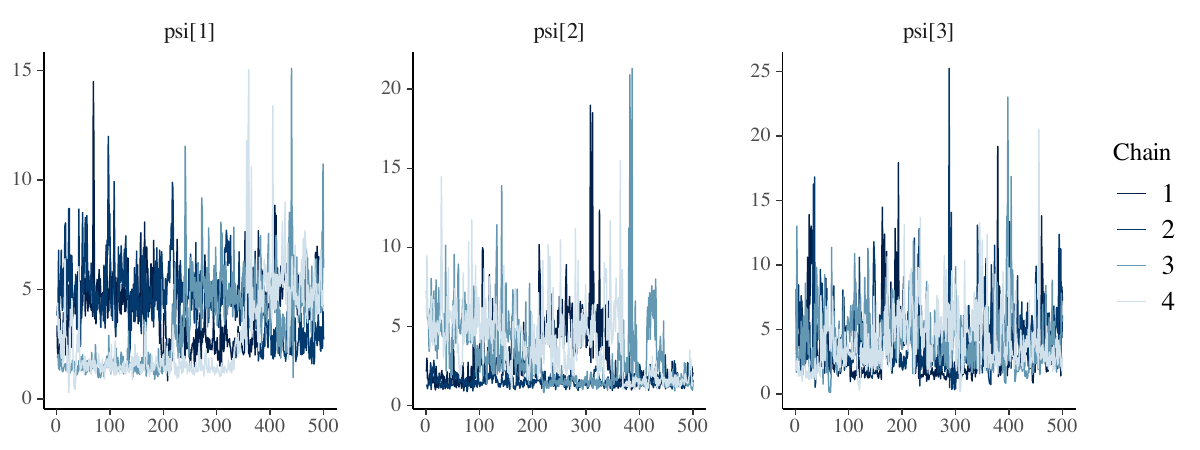}
\label{psi_mixing_A2}
\end{minipage}
}

\subfigure[]{
\begin{minipage}[t]{0.8\linewidth}
\centering
\includegraphics[width=4.0in]{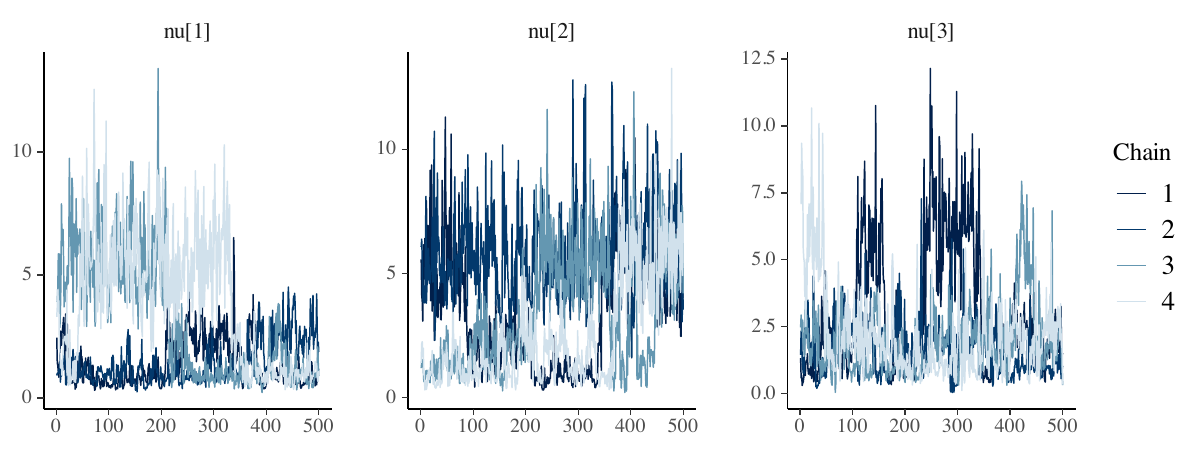}
\label{nu_mixing_A2}
\end{minipage}
}

\caption{\footnotesize Visualization of MCMC traces of other parameters in Setting (a.2). 
(a), trace plots of $\bm{\beta}$; (b), trace plots of the first three components of $\bm{\psi}$; 
(c), trace plots of the first three components of $\bm{\nu}$. }
\label{fig: mixing A2 others}
\end{figure}

\section{Sensitivity analysis}

\subsection{Number of MCMC draws in the tuning procedure}
In general, longer MCMC chains are more possible to provide more reliable posterior approximation. 
Nonetheless, we do NOT want the chain length $N_d$ too large during the hyperparmeter tuning procedure for the sake of computational efficiency. 
This subsection discusses how long the MCMC chains are needed in the ore hyperparameter tuning procedure. 
To determine the chain length needed in our tuning procedure is equivalent to answer the question, ``how many draws do we need for each chain to correctly reflect the variation within a single chain?"
In MCMC practice, this question is closely related to the concept of effective sample size, which is essential in computing the Monte Carlo standard error \citep{gelman2013bayesian}. 
We follow the usual practice that an ESS that is greater than $400$ is sufficient to reflect the MCMC variation \citep{vehtari2021rank}. 
Theoretically, the No-U-Turn sampler in \texttt{Stan} is more effective for sampling ``independent" samples than the random walk M-H sampler \citep{hoffman2014no}; 
empirically, within a single chain, the ratio between the ESS and the total number of draws in \texttt{Stan} is about 90\%, even in very complex models \citep{beraha2021jags}. 
Based on these results, roughly, in a Markov chain sampled by \texttt{Stan}, one may treat a new state as an ``independent" or effective sample. 
Consequently, to attain an ESS of 400, we would expect $N_d = 400/\delta_{adp}$ number of draws, where $\delta_{adp}$ (\texttt{adapt\_delta}) is the target average proposal acceptance probability. 
In \texttt{Stan}, the default $\delta_{adp}$ is set as 0.8.
In \texttt{Stan}, the default $\delta_{adp}$ is set as 0.8.
In conclusion, we recommend to use $N_d = 500$ (after burn-in) to draw sufficiently representative MCMC samples with an acceptable time cost in the tuning procedure. 

We conduct sensitivity analysis to examine the robustness of the hyperparameter tuning algorithm against different choices of the chain length. 
As we mentioned before, we suggest to use more $N_d = 500$ posterior draws in each MCMC chain. 
Hence, we repeat the examples presented in Section 5.1, with chain lengths of 800 and 1000 (after 500 burn-in samples) and re-examine the MCMC mixing status. 

\begin{figure}[!htb]
\centering
\subfigure[]{
	\begin{minipage}[t]{0.45\linewidth}
	\centering
	\includegraphics[width=2.0in]{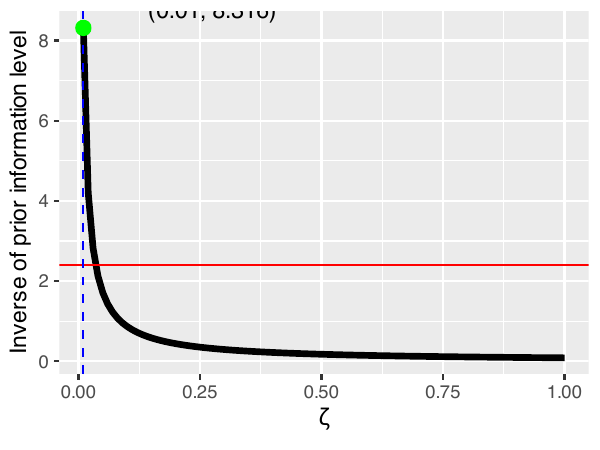}
	\label{zeta_example1_800}
	\end{minipage}
}
\subfigure[]{
	\begin{minipage}[t]{0.45\linewidth}
	\centering
	\includegraphics[width=2.0in]{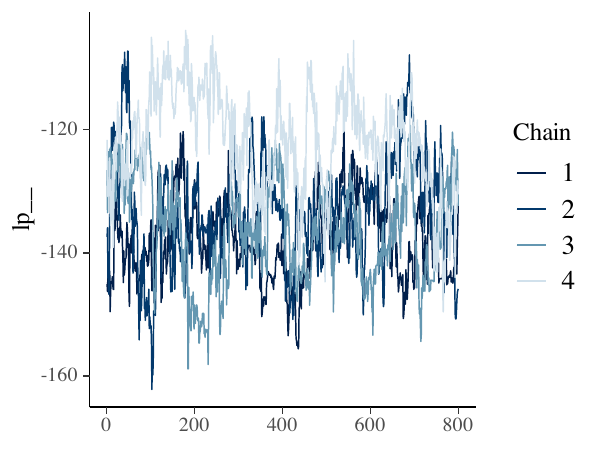}
	\label{lp_example1_800}
	\end{minipage}
}
\subfigure[]{
	\begin{minipage}[t]{0.45\linewidth}
	\centering
	\includegraphics[width=2.0in]{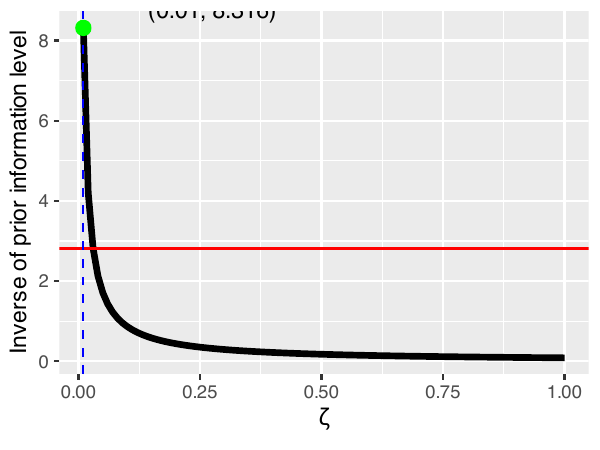}
	\label{zeta_example1_1000}
	\end{minipage}
}
\subfigure[]{
	\begin{minipage}[t]{0.45\linewidth}
	\centering
	\includegraphics[width=2.0in]{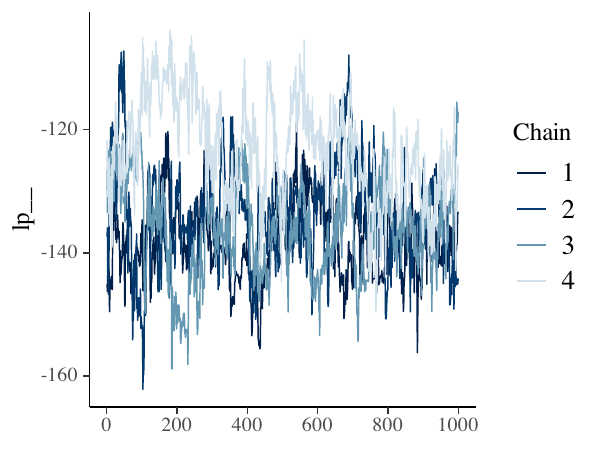}
	\label{lp_example1_1000}
	\end{minipage}
}

\caption{\footnotesize MCMC checking results in setting (b.2) with hyperparameter configuration of $(\eta, \zeta, \rho) = (0.01, 0.01, 1)$ and chain lengths of $800$ and $1000$. 
	(a), the curve of $\widetilde{\mathcal{V}}_{s_{j_0}}(\eta, \zeta)$ with chain length 800. (b), trace plot of chains of \texttt{lp\_\_} with chain length 800.
	(c), the curve of $\widetilde{\mathcal{V}}_{s_{j_0}}(\eta, \zeta)$ with chain length 1000. 
	(d), trace plot of chains of \texttt{lp\_\_} with chain length 1000.
}
\label{fig: example_1_800 & 1000}
\end{figure}

By comparing Figures \ref{zeta_example1_800} and \ref{zeta_example1_1000}, and figures in Section 5.1 of the manuscript, with hyperparameter configuration of $(0.01, 0.01, 1)$, for $N_d > 500$, increasing the chain length changes little about the within-chain MCMC variance. 
Similarly, increasing the chain length cannot resolve the poor mixing issue, as shown by Figures \ref{lp_example1_800} and  \ref{lp_example1_1000}. 

\begin{figure}[!htb]
\centering
\subfigure[]{
	\begin{minipage}[t]{0.45\linewidth}
	\centering
	\includegraphics[width=2.0in]{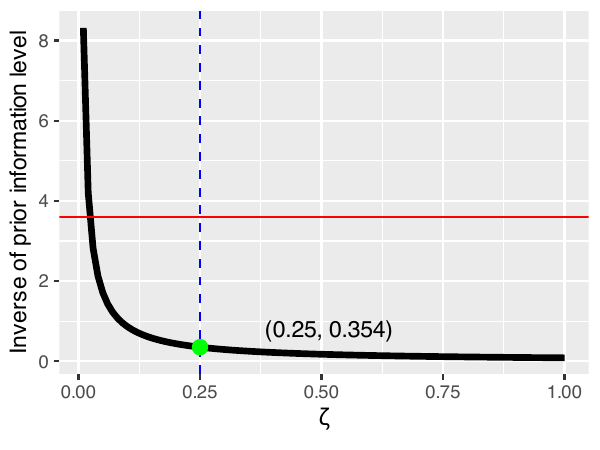}
	\label{zeta_example2_800}
	\end{minipage}
}
\subfigure[]{
	\begin{minipage}[t]{0.45\linewidth}
	\centering
	\includegraphics[width=2.0in]{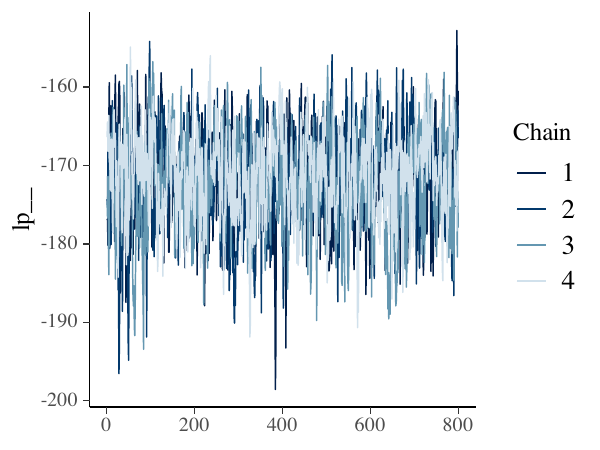}
	\label{lp_example2_800}
	\end{minipage}
}
\subfigure[]{
	\begin{minipage}[t]{0.45\linewidth}
	\centering
	\includegraphics[width=2.0in]{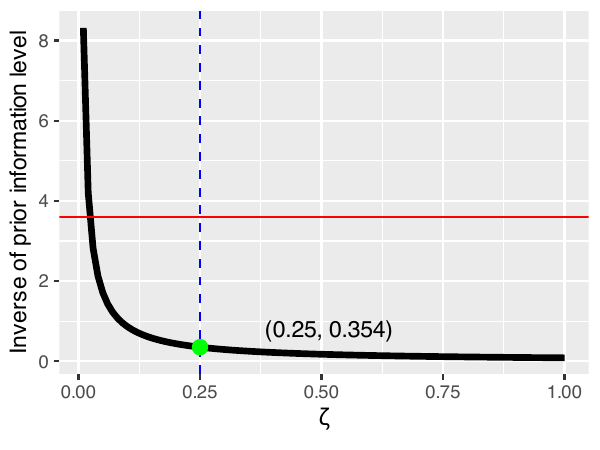}
	\label{zeta_example2_1000}
	\end{minipage}
}
\subfigure[]{
	\begin{minipage}[t]{0.45\linewidth}
	\centering
	\includegraphics[width=2.0in]{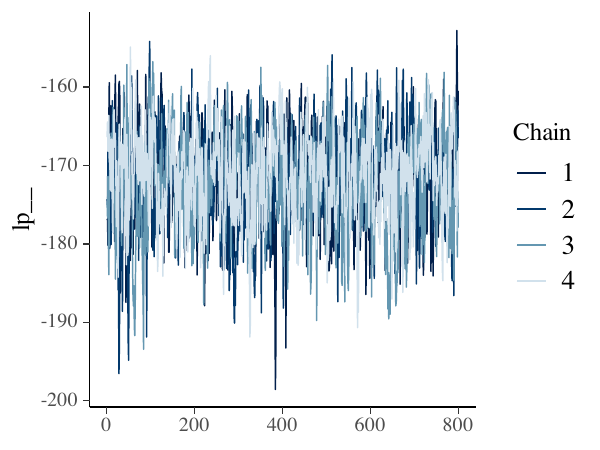}
	\label{lp_example2_1000}
	\end{minipage}
}

\caption{\footnotesize MCMC checking results in setting (b.2) with hyperparameter configuration of $(\eta, \zeta, \rho) = (0.01, 0.25, 1)$ and chain lengths of $800$ and $1000$. 
	(a), the curve of $\widetilde{\mathcal{V}}_{s_{j_0}}(\eta, \zeta)$ with chain length 800. (b), trace plot of chains of \texttt{lp\_\_} with chain length 800.
	(c), the curve of $\widetilde{\mathcal{V}}_{s_{j_0}}(\eta, \zeta)$ with chain length 1000. 
	(d), trace plot of chains of \texttt{lp\_\_} with chain length 1000.
}
\label{fig: example_2_800 & 1000}
\end{figure}

Similarly, with correctly configured hyperparameters $(\eta, \zeta, \rho) = (0.01, 0.25, 1)$, increasing the chain length will not affect the diagnostic result of the proposed sufficient informativeness criterion, as well as the mixing of MCMC chains.

\subsection{Increasing prior informativeness} 
\label{subsec: Increasing prior informativeness}
This subsection discuss the affect of increasing the prior informativeness by changing the hyperparameter configuration. 
Note that we do NOT expect the priors to be too informative to influence sampling. 
Consequently, we focus on increasing $\zeta$ since for $\zeta > 0.25$, increasing $\zeta$ increases the prior informativeness mildly based on the formula of $\widetilde{\mathcal{V}}_{s_{j_0}}(\eta, \zeta)$.
We consider the configuration of $\zeta = 0.5$, and take Settings (a.1) and (a.2) for example. 

\noindent{\textbf{Setting (a.1)}}. ~Figures \ref{zeta_plot_A1 0.5} and \ref{lp_plot_A1 0.5} show that slightly increasing the prior informativeness does NOT change the MCMC mixing (for predictive inference). 
Meanwhile, the ESS of \texttt{lp\-
	\_} is 434, which is also sufficient. 
Nonetheless, increasing $\zeta$ leads to better mixing of $\bm{\beta}$, with $\hat{R}$ statistics being 1.003, 1.006, 1.002 for $\beta_1$, $\beta_2$, and $\beta_3$, respectively. 

\begin{figure}[!htb]
\centering
\subfigure[]{
	\begin{minipage}[t]{0.45\linewidth}
	\centering
	\includegraphics[width=2.0in]{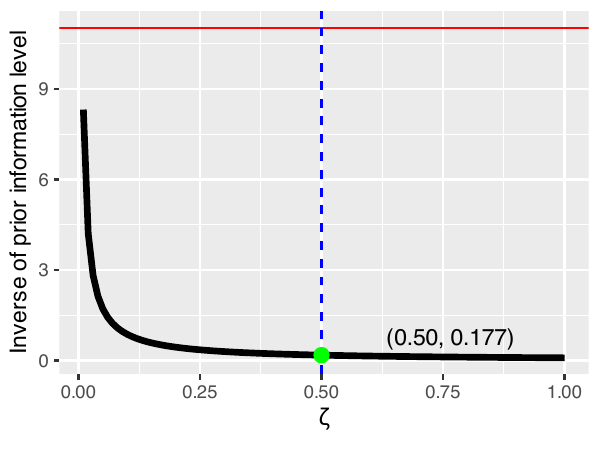}
	\label{zeta_plot_A1 0.5}
	\end{minipage}
}
\vspace{-.5cm}
\subfigure[]{
	\begin{minipage}[t]{0.45\linewidth}
	\centering
	\includegraphics[width=2.0in]{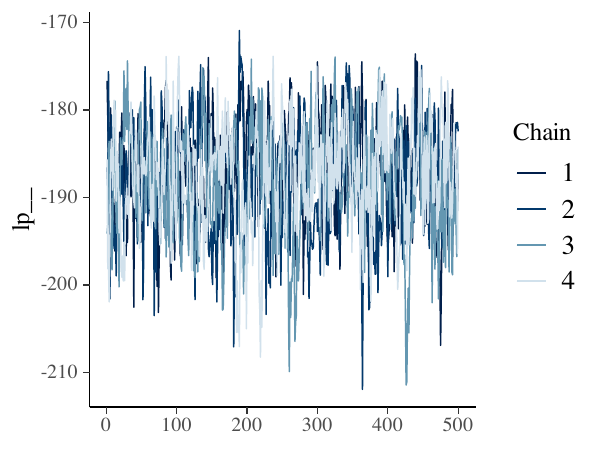}
	\label{lp_plot_A1 0.5}
	\end{minipage}
}

\caption{\footnotesize MCMC checking results in setting (a.1) with hyperparameter configuration of $(\eta, \zeta, \rho) = (0.01, 0.5, 1)$. 
	(a), the curve of $\widetilde{\mathcal{V}}_{s_{j_0}}(\eta, \zeta)$ with $\eta = 0.01$ fixed; horizontal line: the within-chain MCMC variance sampled in setting (a.1). (b), trace plot of chains of \texttt{lp\_\_}. }
\label{fig: mixing A1 0.5}
\end{figure}

\noindent{\textbf{Setting (a.2)}}.~ Similar results are also presented in Figures \ref{zeta_plot_A2 0.5} and \ref{lp_plot_A2 0.5}. 
Furthermore, in this case, the ESS of \texttt{lp\_\_} increases from 520 to 595, and the ESS of $\bm{\beta}$ also increases.

\begin{figure}[!htb]
\centering
\subfigure[]{
	\begin{minipage}[t]{0.45\linewidth}
	\centering
	\includegraphics[width=2.0in]{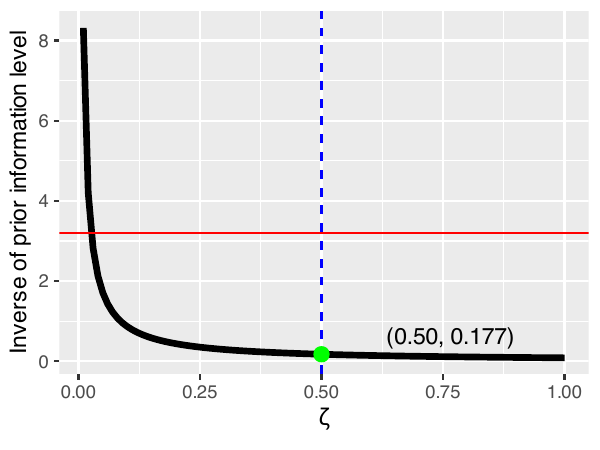}
	\label{zeta_plot_A2 0.5}
	\end{minipage}
}
\vspace{-.5cm}
\subfigure[]{
	\begin{minipage}[t]{0.45\linewidth}
	\centering
	\includegraphics[width=2.0in]{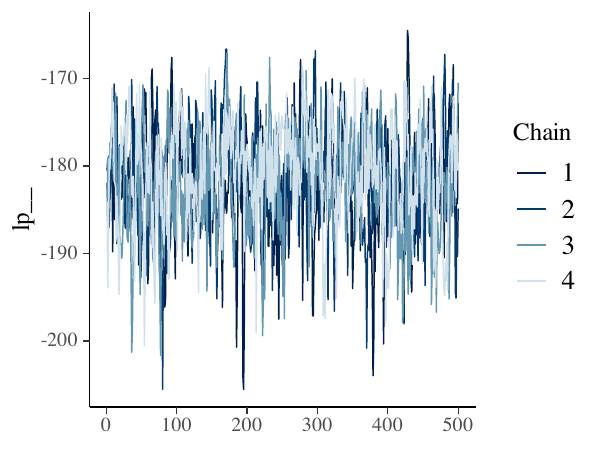}
	\label{lp_plot_A2 0.5}
	\end{minipage}
}

\caption{\footnotesize MCMC checking results in setting (a.2) with hyperparameter configuration of $(\eta, \zeta, \rho) = (0.01, 0.5, 1)$. 
	(a), the curve of $\widetilde{\mathcal{V}}_{s_{j_0}}(\eta, \zeta)$ with $\eta = 0.01$ fixed; horizontal line: the within-chain MCMC variance sampled in setting (a.2). (b), trace plot of chains of \texttt{lp\_\_}. }
\label{fig: mixing A2 0.5}
\end{figure}

In summary, slightly increasing the prior informativeness by setting $\zeta = 0.5$ generally leads to higher ESS than setting $\zeta = 0.25$, especially for the estimation of $\bm{\beta}$.
Consequently, we use the hyperparameter configuration of $(\eta, \zeta, \rho) = (0.01, 0.5, 1)$ throughout all numerical studies. 

\subsection{Number of initial knots}
In this subsection, we check the sensitivity of the number of initial knots. 
We examine the predictive performance on two sets of quantile knots: $4$ and $10$, i.e. setting the as the each $25\%$ and $10\%$ quantiles of the observations. 
We present the results under settings (a.1) and (a.2) for sensitivity analysis.

\begin{figure}[!htp]
\subfigcapskip=-10pt
\centering
\subfigure[]{
	\begin{minipage}[t]{0.45\linewidth}
	\centering
	\includegraphics[width=1.5in]{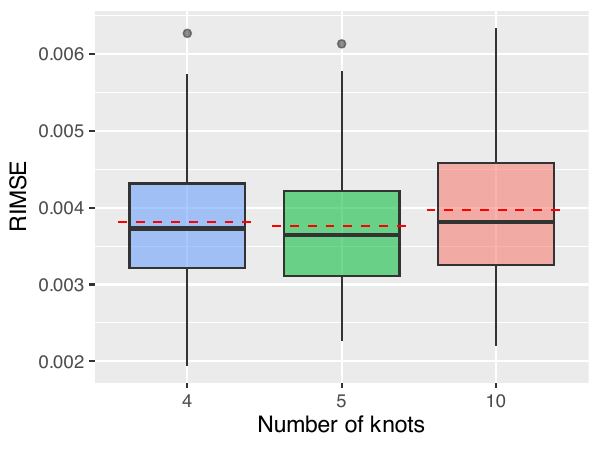}
	\label{BoxGaussSenseRIMSE}
	\end{minipage}
}
\vspace{-.5cm}
\subfigure[]{
	\begin{minipage}[t]{0.45\linewidth}
	\centering
	\includegraphics[width=1.5in]{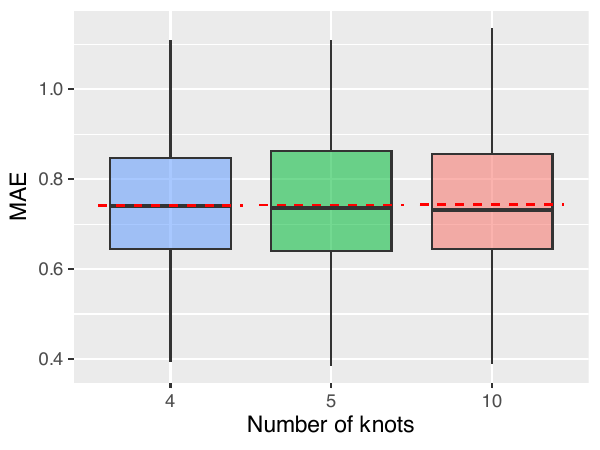}
	\label{BoxGaussSenseMAE}
	\end{minipage}
}
\caption{\footnotesize{Prediction comparison between different number of initial knots under Setting (a.1)} }
\label{Sens_Setting a.1}
\end{figure}

The results under Settings (a.1) and (a.2) are presented in Figures \ref{Sens_Setting a.1} and \ref{Sens_Setting a.2} respectively. 
Under setting (a.1), choosing 4 and 5 initial knots do NOT yield significant difference in RIMSE (two-sided paired t-test $p$-value: 0.169); 
under setting (a.2), choosing 5 and 10 initial knots also yield NO significant difference in RIMSE (two-sided paired t-test $p$-value: 0.095). 
Under two settings, all choices of the number of initial knots yield no significant difference in MAE. 
This sensitivity analysis demonstrates that \texttt{BuLTM} is not sensitive to the choice of the number of initial knots.

\begin{figure}[!htp]
\subfigcapskip=-10pt
\centering
\subfigure[]{
	\begin{minipage}[t]{0.45\linewidth}
	\centering
	\includegraphics[width=1.5in]{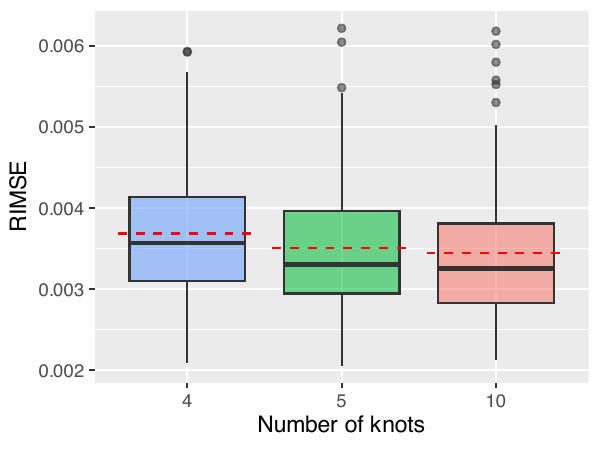}
	\label{BoxMixSenseRIMSE}
	\end{minipage}
}
\vspace{-.5cm}
\subfigure[]{
	\begin{minipage}[t]{0.45\linewidth}
	\centering
	\includegraphics[width=1.5in]{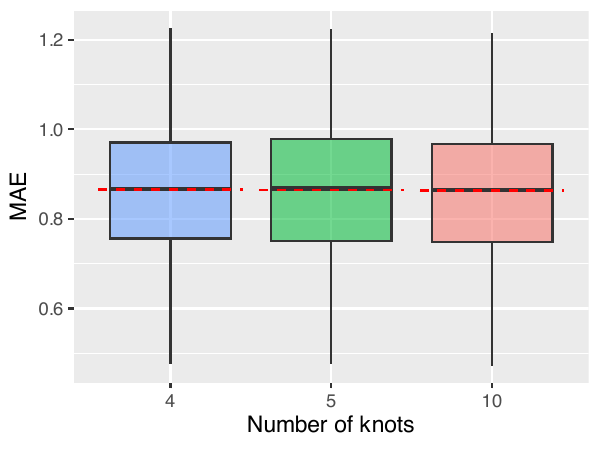}
	\label{BoxMixSenseMAE}
	\end{minipage}
}
\caption{\footnotesize{Prediction comparison between different number of initial knots under Setting (a.2)} }
\label{Sens_Setting a.2}
\end{figure}

\section{Visualize the knot interpolation procedure with censored data}
{We visualize the knot interpolation algorithm through an example in Figure \ref{fig: quantile I spline}. 
	We set the number of initial knots $N_I= 4$, yielding interior knots $s_0, s_1, s_2,$ and $s_3$, which are located at the $0, 25\%, 50\%$ and $75\%$ quantiles of the uncensored observations.  The boundary knots are set at $0$ and $\tau$. 
	Then we compute the absolute difference between the empirical function at each interior knots on $s_j$, for $j=0, \ldots, 3$. 
	Since $|\hat{F}_{\tilde{y}}(s_1) - \hat{F}_{\tilde{y}_c}(s_1)|> 0.05$, we interpolate $s_1^* = \hat{Q}_{\tilde{y}_c}(0.25)$ as the complement knot. 
	Thus, the final interior knots are $(s_0, s_1^*, s_1, s_2, s_3)$. 
	With the two boundary knots $0$ and $\tau$, at the smoothing degree $r=4$, we finally obtain $9$ I-spline functions. 
}

\begin{figure}[!htb]
\centering
\includegraphics[scale = .9]{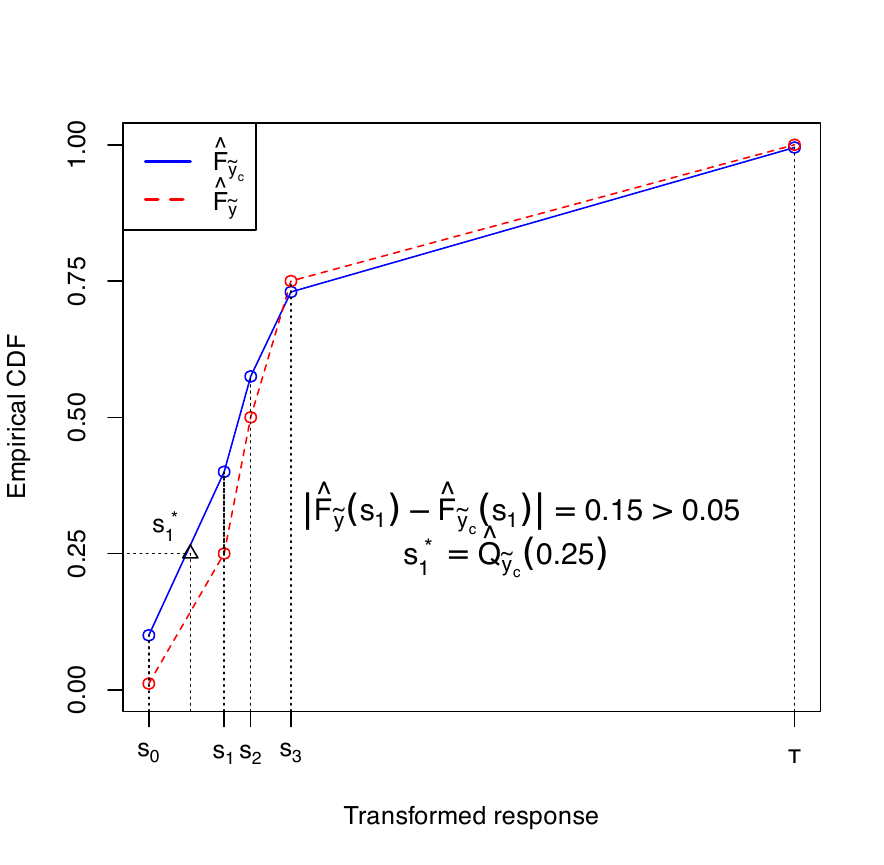}
\caption{\footnotesize
	An example of selecting knots forn quantile-knot I-spline functions with censored data. The number of initial knots $N_I= 4$.  }
\label{fig: quantile I spline}
\end{figure}

\section{Additional results of real-world data analysis }

\subsection{Heart failure clinical records data}
\noindent{\textbf{Parametric estimation}}\\
Results of parametric estimation on the heart failure dataset given by \texttt{BuLTM}, \texttt{TransModel}, and \texttt{spBayesSurv}are displayed in Table \ref{tabHeart}. 
We find that \texttt{BuLTM} is consistent with \texttt{spBayesSurv} in the detection of significance, while \texttt{TransModel} fails to detect the significance of the covariate $Z_9$, serum sodium. 
Existing medical research has evidenced that lower serum sodium was associated with higher in-hospital and 60-day mortality for heart failure patients \citep{klein2005lower}. 
Hence, the results of \texttt{BuLTM} and \texttt{spBayesSurv} are more meaningful and reasonable. 
We conjecture this may be caused by the relatively high censoring rate.

\begin{table}[!htp]
\centering
\scriptsize
\def~{\hphantom{0}}
\caption{\label{tabHeart}\footnotesize{Results of estimated $\bbeta$ in the analysis to heart failure clinical records data. Credible intervals are given on $95\%$ credibility for BuLTM and spBayesSurv. The confidence interval of TransModel is a 95\% Wald-type confidence level.} }
\tabcolsep 4pt
\begin{tabular}{c|cc|cc|cc}
\hline
&  \multicolumn{2}{c}{BuLTM} & \multicolumn{2}{c}{spBayesSurv} & \multicolumn{2}{c}{TransModel} \\
\hline
Covariate              & Estimate &     95\%CI      & Estimate &    95\%CI   &   Estimate &    95\%CI    \\
$Z_1=$ age     & -0.163   & (-0.433, 0.063) &  -4.670  & (-6.182, -3.135) & -4.631 & (-6.474, -2.788)\\
$Z_2 =$ anemia &   -0.013  & (-0.036, -0.001) &  -0.412  & (-0.764, -0.066) & -0.408 & (-0.827,  0.012)\\
$Z_3=$  creatinine phosphokinase & -0.002 & (-0.010, 0.004) & -0.074 & (-0.262, 0.113) &-0.075 & (-0.293, 0.143)\\
$Z_4=$ diabetes & -0.004 & (-0.020, 0.008) & -0.117 & (-0.476, 0.256) & -0.125 & (-0.560, 0.310)\\ 
$Z_5 =$ ejection fraction & 0.022 & (0.008, 0.060) & 0.586 & (0.386, 0.785)&  4.810 & (2.773, 6.847)\\ 
$Z_6=$ high blood pressure & -0.015 & (-0.042, -0.001) & -0.460 & (-0.807, -0.099) & -0.455 & (-0.879, -0.031)\\
$Z_7=$ platelets & 0.076 & (-0.033, 0.389) & 1.303 & (-2.836, 5.327) &  1.384 & (-3.392, 6.160)\\
$Z_8=$ serum creatinine & -0.012 & (-0.033, -0.004) & -0.306 & (-0.421, -0.183) & -0.313 & (-0.453, -0.173)\\
$Z_9=$ serum sodium & 0.939 & (0.787, 0.997) & 41.347 & (3.248, 74.256) &43.077 & (-2.777, 88.931)\\
$Z_{10} = $ sex & 0.009 & (-0.005, 0.033) & 0.222 & (-0.185, 0.625) & 0.224 & (-0.269, 0.716)\\ 
$Z_{11} =$ smoking & -0.005 & (-0.024, 0.010) & -0.133 & (-0.542, 0.282) & -0.148 & (-0.641, 0.345)\\
\hline
\end{tabular}
\end{table}

We use the survival AUC curves to compare the parametric estimation given by the three methods. 
As shown by Figure \ref{HeartScoreAUC}, the parametric estimation given by the three methods generate almost the same AUC, indicating that their parametric estimation has almost the same predictive performances. 

\begin{figure}[!htp]
\centering
\includegraphics[width=2.5in]{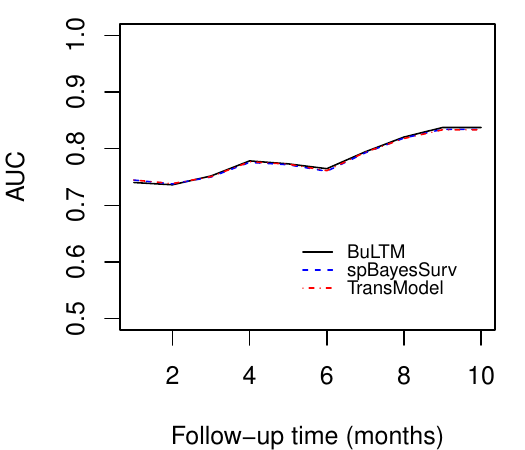}
\caption{\footnotesize  Time dependent survival $\text{AUC}(t)$ computed by estimated relative risks on Heart failure dataset. }
\label{HeartScoreAUC}
\end{figure}

\subsection{Veterans lung cancer data}
\label{subsec: add app veterans}
We analyze the veterans lung cancer dataset from \texttt{R} package \texttt{survival} \citep{Terry2022}. 
It contains $137$ patients from a randomized trial receiving either a standard or a test form of chemotherapy. 
In the study, the failure time is one of the primary endpoints for the trial as $128$ patients were followed to death. 
We include six covariates, the first five of which are $Z_1 = \text{karno}/10$ (karnofsky score), $Z_2 = \text{prior}/10$ (prior treatment, with 0 for no therapy and 10 otherwise), $Z_3 = \text{age}/100$ (years), $Z_4 = \text{diagtime}/100$ (time in months from diagnosis to randomization), and $Z_5 = I (\text{treatment} = \text{test form of chemotherapy})$. The remaining is the covariate of the cell type which has four categories, adeno, squamous, small cell, and large cell.
Thus we include indicator variables to associate with  time-to-death, that is, $Z_6 = I(\text{cell type} = \text{squamous}), Z_7 = I (\text{celltype} = \text{small}), ~\text{and} ~Z_8 = I (\text{celltype} = \text{large})$.  

For the veterans data, we fit the nonparametric transformation model for the veterans data by \texttt{BuLTM} and fit the semiparametric survival models by \texttt{spBayesSurv} and \texttt{TransModel} respectively. 
In this case, \texttt{spBayesSurv} selects the PO model and thus, \texttt{TransModel} specifies $r=1$.

For predictive capability comparison, we conduct 5 runs of 10-fold cross validation. 
We assess the predictive performances through C-index and IBS. 
Figure \ref{Pred_Lung} shows that \texttt{BuLTM} is comparable to all competitors in both C-index and IBS.

\begin{figure}[!htp]
\subfigcapskip=-10pt
\centering
\subfigure[]{
	\begin{minipage}[t]{0.45\linewidth}
	\centering
	\includegraphics[width=2in]{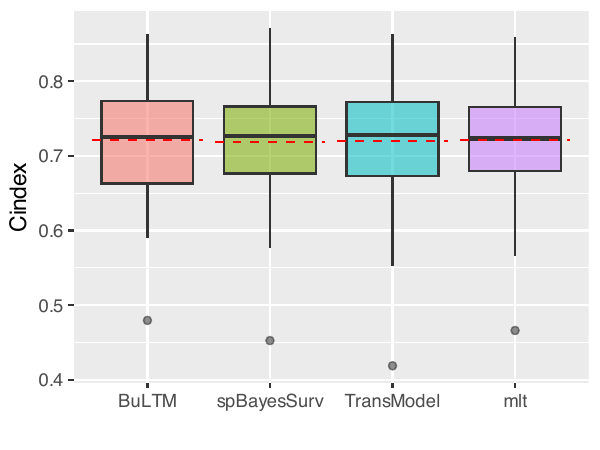}
	\label{Cindex_Lung}
	\end{minipage}
}
\vspace{-.5cm}
\subfigure[]{
	\begin{minipage}[t]{0.45\linewidth}
	\centering
	\includegraphics[width=1.9in]{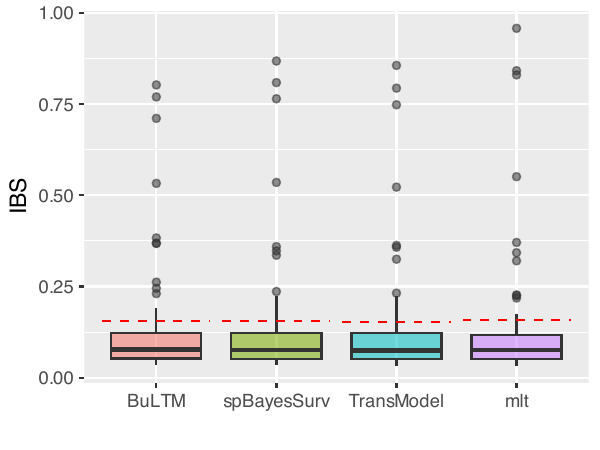}
	\label{IBS_Lung}
	\end{minipage}
}

\caption{\footnotesize Prediction comparison between BuLTM, spBayesSurv, and TransModel on the veterans dataset; (a), C index; 
	(b), IBS; red dashed lines: the mean of the metrics.}
\label{Pred_Lung}
\end{figure}

\noindent{\textbf{{Parametric estimation}}}\\
Results of parametric estimation on the veterans lung cancer dataset given by \texttt{BuLTM}, \texttt{TransModel} and \texttt{spBayesSurv} are displayed in Table \ref{tabLung}. 
The three methods provide similar significance levels for all coefficients. 
Although some signs of estimated coefficients are different, say $\beta_3$ and $\beta_7$, they are not significant since their credible/confidence intervals cover zero. 
That implies qualitative interpretations of the estimates of the regression coefficients under the three models are stable. 
\begin{table}[!htb]
\centering
\footnotesize
\def~{\hphantom{0}}
\caption{\label{tabLung}\footnotesize{Results of estimated $\bbeta$ for veterans administration lung cancer data. Credible intervals are given on $95\%$ credibility for BuLTM and spBayesSurv. The confidence interval of TransModel is a 95\% Wald-type confidence level. }}
\tabcolsep 3pt
\footnotesize
\begin{tabular}{c|cc|cc|cc}
\hline
\multicolumn{1}{c}{} &  \multicolumn{2}{c}{BuLTM} & \multicolumn{2}{c}{spBayesSurv} & \multicolumn{2}{c}{TransModel} \\
\hline
Covariate   & Estimate &     95\%CI     & Estimate &   95\%CI   & Estimate &  95\%CI   \\
$Z_1 $ &  0.119  & (0.045, 0.246)  &  0.617   & (0.449, 0.800)  & 0.553  & (0.368, 0.737)\\
$Z_2 $ &  -0.302  & (-0.951, 0.897) &  -1.391 & (-8.597, 6.028) & -0.388 &  (-8.546, 7.768)\\
$Z_3 $ &{-0.006} & (-0.700, 0.671) &  {1.426}   & (-1.643, 4.477) &  {0.945} & (-2.441, 4.331)\\  
$Z_4$ & 0.081 & (-0.693, 0.730) & 0.033 & (-3.533, 3.469)& 0.010& (-3.475, 3.496) \\ 
$Z_5 $ &  -0.044  & (-0.227, 0.117) & -0.147   & (-0.739, 0.487) & -0.278 &  (-0.963, 0.405)\\
$Z_6 $ &  0.350  & (0.093, 0.694) &  1.387   & (0.396, 2.334)  & 1.995 & (0.063, 3.027)\\
$Z_7 $ &  {-0.005}  & (-0.242, 0.205) &  {0.058}   & (-0.739, 0.916) & {0.413} & (-0.514, 1.342)\\
$Z_8 $ & 0.274 & (0.053, 0.571)  &  1.367   & (0.444, 2.308) & 1.364 & (0.343, 2.385) \\
\hline
\end{tabular}
\end{table}

We assess the parametric estimation results through the survival AUC curves. 
Figure \ref{LungScoreAUC} displays the dynamic AUCs using the estimated relative risks given by \texttt{BuLTM}, \texttt{spBayesSurv}, and \texttt{TransModel} as diagnostics. 
We find \texttt{BuLTM} and \texttt{TransModel} share almost the same survival AUC curves which are higher than that of \texttt{spBayesSurv}. 
Thus, \texttt{BuLTM} appears to provide a competitive parametric estimation on the veterans dataset.
\begin{figure}[!htp]
	\centering
	\includegraphics[width=2.25in]{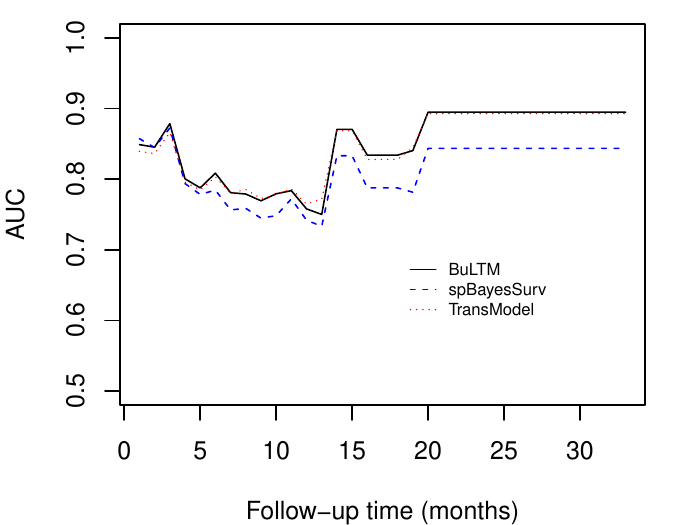}
\caption{\footnotesize  Time dependent survival $\text{AUC}(t)$ computed by estimated relative risks on veterans lung cancer dataset.}
\label{LungScoreAUC}
\end{figure}

\section{Posterior checking}

The sufficiently informative prior elicitation for infinite-dimensional parameters $H$ and $S_\xi$ is not noninformative. 
An objective Bayesian may worry that the prior information may impact the posterior too much such that the  prior-to-posterior update is not data-driven. 
We conduct posterior checking under simulation Setting (a.1) and  to check the difference between the marginal prior and posterior densities of parameters $\bm{\beta}$ and $\bm{\alpha}$. 
The posterior checking results under other settings are similar. 
We use the aforementioned hyperparameter configuration $(\eta, \zeta, \rho) = (0.01, 0.5, 1)$. 
For numerical simplicity, we set $\pi(\bm{\beta}) =  N(0, 10^6)$ as the noninformative prior. 

Figure \ref{figalpha} compares the priors and marginal posteriors of $\bm{\alpha}$, the first 8 coefficients of I-spline functions, where we find
all the coefficients in the I-splines prior vary significantly from the prior except $\alpha_8$. 
This evidences that the prior-to-posterior updating is sufficiently driven by the data. 

\begin{figure}[!htp]
\centering    \includegraphics[width=0.49\textwidth]{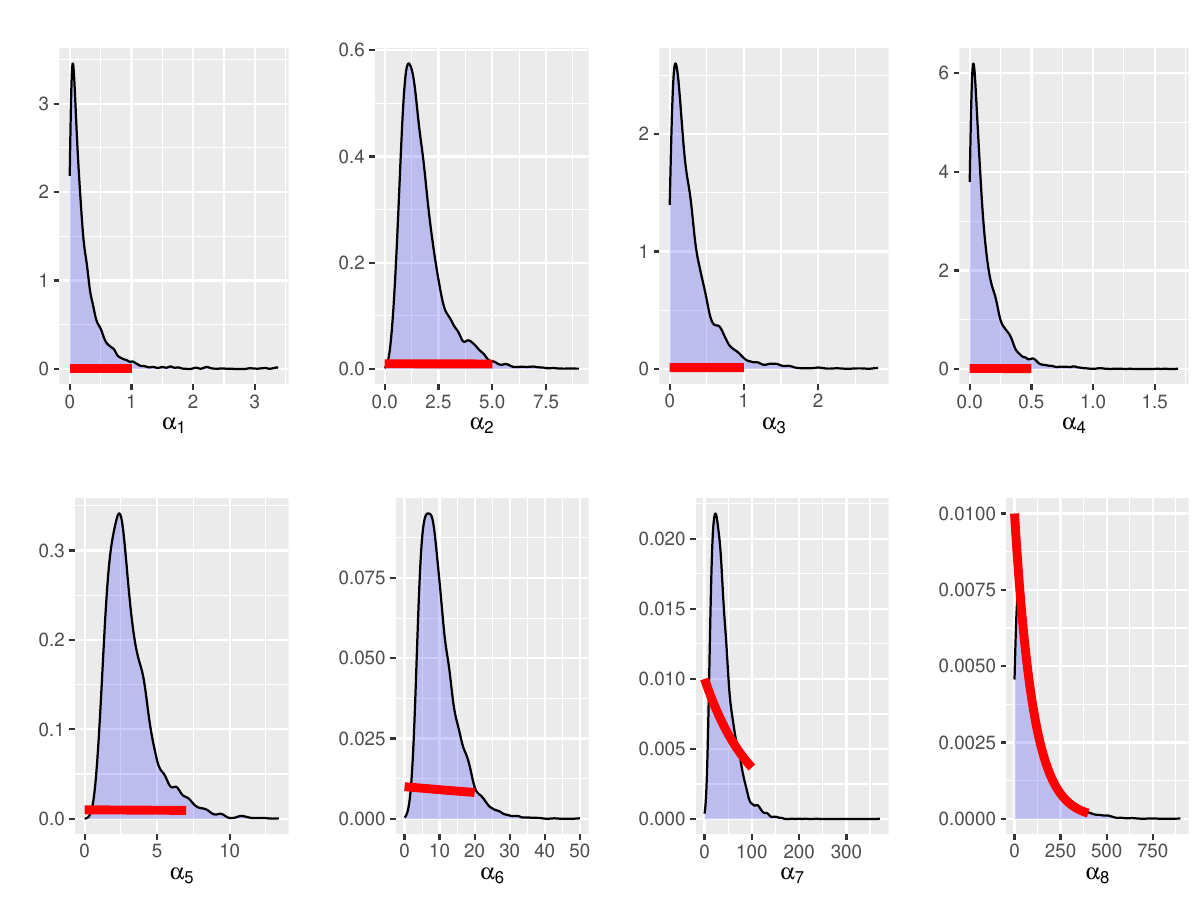}
\caption{\label{figalpha}\footnotesize{Comparison between the the marginal posterior density  and priors of $\alpha_1, \ldots , \alpha_8$. Shaded region, marginal posterior density; Wide line, prior density of $\exp(\eta)$.}}
\end{figure}

We also compare the priors with the marginal posterior of $\bbeta$, the unconstrained parameter sampled from MCMC.
Fig \ref{figbeta} shows an apparent difference between flat priors and marginal posterior of $\bbeta$, demonstrating that the posterior updating is driven by data.  

\begin{figure}[!htp]
\centering
\includegraphics[width=0.49\textwidth]{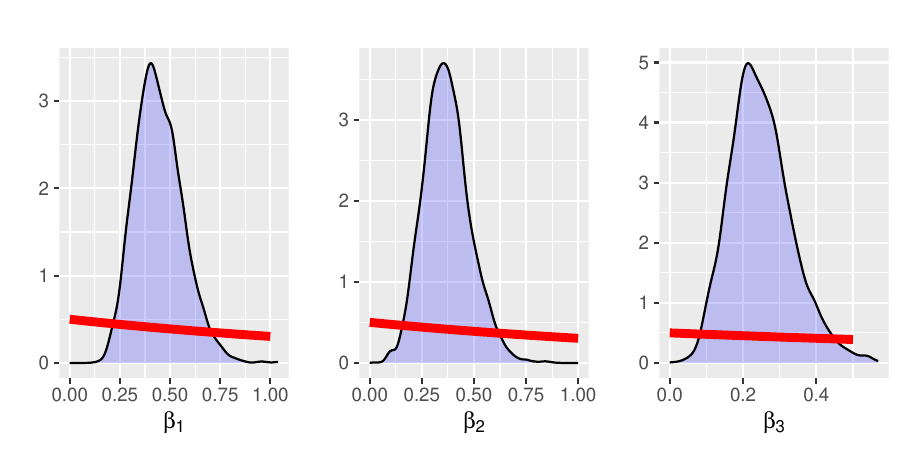}
\caption{\label{figbeta}\footnotesize{Comparison between the the marginal posterior density of $\beta$ without posterior projection and  corresponding priors. The shaded region, posterior density; wide line, flat prior. }}
\end{figure}

\section{Non-applicability of constrained priors }
One may consider other alternative choices of parametric and nonparametric priors for the triplet $(\bbeta, H, S_\xi)$. 
Here we introduce some alternative choices of priors. It includes how to construct constrained priors to make the MTM identified. 
Another construction of I-splines prior with shrinkage prior for $H$ is also given here. 

Our spirit is inspired by Horowitz's normalization conditions \citep{horowitz1996semiparametric}. 
Like the manuscript, we use the unit scale condition that $||\bbeta||=1$ as an equivalent condition of Horowitz's scale normalization. 
Rather than applying posterior projection, we assign the uniform distribution on the $p$-dim unit hypersphere as the prior for the fully identified $\bbeta$. 
It is conducted by the following transformation
$$
\bbeta_* \sim N(0, I), \bbeta = \bbeta_* / ||\bbeta_*||^{1/2}.
$$
Still, we need the location normalization, which assumes that the $H(t_0)=1$ or $h(t_0)=0$  for some finite $t_0$ \citep{horowitz1996semiparametric}.
We adopt the I-spline priors as our initial. 
We formulate $H$ by
$$
H(t) = \sum_{j=1}^{K} \alpha_j B_j(t),
$$
where $K=J+r$ is the number of I-spline functions; see \textit{Section \ref{propSpline}}.  
By the characteristic of I-spline functions on interval $D = (0, \tau)$, if $\sum_{j=1}^{K}\alpha_j = 1$,  $H$ will surely pass the point $(\tau, 1)$. 
Therefore, the location normalization condition is transferred to a sum-to-one restriction, that is, $(\alpha_1, \ldots, \alpha_K)$ fall into a $K$-dim simplex. 
We consider two choices of priors for the $p$-dim simplex.
The first one is the Dirichlet prior
$$
(\alpha_1, \ldots, \alpha_K) \sim \text{Dir}(a_1, \ldots, a_K),
$$
where $\{a_j\}_{j=1}^K$ are hyperparameters of Dirichlet distribution. 
Alternatively, we may consider a kind of transformed prior. For $j=1, \ldots, K$, 
$$
\alpha_j^* \sim \exp(\eta), ~\alpha_j = \alpha_j^* \bigg/ \sum_{j=1}^K \alpha_j^*.
$$
Both these two priors normalize the location of $H$ and therefore, fully identify the transformation function. 

The above priors make the transformation model fully identified. 
However, with these priors, we find that the MCMC procedure by NUTS converges very slowly and suffers from poor mixing.
What's worse, the prediction accuracy is poor. 
These two drawbacks force us NOT to work on a fully identified model.

	\bibliographystyle{apalike}
	\linespread{0.1}
	\selectfont
	\bibliography{main}
	
\end{document}